\documentclass[10pt]{amsart}
\usepackage{hyperref}

\usepackage{graphicx, enumerate, url}
\usepackage{amssymb,mathtools,amsthm}
\usepackage{algorithm}
\usepackage{algpseudocode}
\usepackage{color}
\usepackage{tikz}
\usetikzlibrary{3d,calc}
\usepackage{bm,bbm}

\numberwithin{equation}{section}
\numberwithin{algorithm}{section}

\theoremstyle{plain}
\newtheorem{theorem}{Theorem}[section]
\newtheorem{proposition}[theorem]{Proposition}
\newtheorem{lemma}[theorem]{Lemma}
\newtheorem{corollary}[theorem]{Corollary}
\newtheorem{conjecture}[theorem]{Conjecture}

\theoremstyle{definition}
\newtheorem{definition}[theorem]{Definition}

\theoremstyle{remark}

\usepackage{lipsum}
\usepackage{amsfonts}
\usepackage{epstopdf}

\usepackage{cleveref}
\usepackage{soul}
\usepackage{amssymb}
\usepackage{easybmat}
\usepackage[colorinlistoftodos,bordercolor=orange,backgroundcolor=orange!20,linecolor=orange,textsize=scriptsize]{todonotes}
\usepackage{color}
\usepackage{graphicx, enumerate, url}

\newcommand{\red}[1]{\textcolor{red}{#1}}

\definecolor{darkgreen}{RGB}{16,122,16}

 \renewcommand{\geq}{\geqslant}
 
 \renewcommand{\ge}{\geqslant}
 \renewcommand{\le}{\leqslant}

  \newcommand{\y}{\mathbf{y}}
  \newcommand{\x}{\mathbf{x}}
  \newcommand{\z}{\mathbf{z}}
  \newcommand{\1}{\mathbf{1}}
  \newcommand{\0}{\mathbf{0}}
 \newcommand\ba{\mathbf{a}}

 \newcommand\be{{\mathbf e}}
 
 \newcommand\bF{{\mathbf F}}
 \newcommand\bG{{\mathbf G}}
 
 \newcommand\bi{{\mathbf i}}
 \newcommand\bj{{\mathbf j}}

 \newcommand\bn{{\mathbf n}}

 \newcommand\bp{{\mathbf p}}
 
 \newcommand\bq{{\mathbf q}}

 \newcommand\bu{\mathbf{u}}
 
 \newcommand\bv{\mathbf{v}}
 \newcommand\bV{{\mathbf V}}

 \newcommand\bx{{\mathbf x}}
 
 \newcommand\by{{\mathbf y}}

\newcommand\cA{{\mathcal A}}
 \newcommand\cB{{\mathcal B}}

 \newcommand\cH{{\mathcal H}}
 \newcommand\cI{{\mathcal I}}

 \newcommand\cP{{\mathcal P}}

 \newcommand\cS{{\mathcal S}}
 \newcommand\cT{{\mathcal T}}
 \newcommand\cU{{\mathcal U}}
 \newcommand\cV{{\mathcal V}}
 
 \newcommand\cX{{\mathcal X}}
 \newcommand\cY{{\mathcal Y}}
 \newcommand\cZ{{\mathcal Z}}
 
 \newcommand\rB{{\rm B}}

 \newcommand\rP{{\rm P}}
 
 \newcommand\rS{{\rm S}}

  \newcommand{\e}{\mathbf{e}}

  \newcommand{\gl}{\mathbf{GL}}
  \newcommand{\m}{\mathbf{m}}
  \newcommand{\n}{\mathbf{n}}

  \newcommand{\trans}{^\top}

\newcommand{\orb}{\operatorname{orb}}
\newcommand{\tr}{\operatorname{Tr}}
\newcommand{\rank}{\operatorname{rank}}

\newcommand{\C}{\mathbb{C}}
\newcommand{\F}{\mathbb{F}}

\newcommand{\R}{\mathbb{R}}
\newcommand{\N}{\mathbb{N}}

\newcommand{\Z}{\mathbb{Z}}

\title{Rank of a tensor and quantum entanglement}

\author{Wojciech Bruzda}
\address{Institute of Theoretical Physics, Jagiellonian University, Krakow, Poland}
\email{w.bruzda@uj.edu.pl}
\author[Shmuel Friedland]{Shmuel~Friedland}
\address{Department of Mathematics, Statistics, and Computer Science, University of Illinois at Chicago,  Chicago, Illinois, 60607-7045, USA }
\email{friedlan@uic.edu}
\author[K.~{\.Z}yczkowski]{Karol~{\.Z}yczkowski}
\address{Institute of Theoretical  Physics, Jagiellonian University, Krakow, Poland}
\address{Center for Theoretical Physics, Polish Academy of Science, Warsaw, Poland}
\email{karol.zyczkowski@uj.edu.pl}

\begin{document}

\maketitle

 \begin{abstract} 
 The rank of a tensor is analyzed in context
 of quantum entanglement.
 A pure quantum state  $\bf v$ of a composite
 system consisting of $d$ subsystems with $n$ levels each
is viewed as a vector in the $d$-fold tensor product of $n$-dimensional Hilbert space and can
be identified with a tensor with $d$ indices, each running from $1$ to $n$.
  We discuss the notions of the generic rank and the maximal rank 
 of a tensor and review results known for the low dimensions. 
 Another variant of this notion, called the border rank of a tensor, 
 is shown to be relevant for characterization of orbits of
 quantum states generated by the group of  special linear transformations.
A quantum state ${\bf v}$ is called {\sl entangled},
if it {\sl cannot} be written in the product form, 
 ${\bf v} \ne {\bf v}_1 \otimes  {\bf v}_2 \otimes \cdots \otimes  {\bf v}_d$,
 what implies correlations between physical subsystems.
 A relation between various ranks and norms of a tensor and the entanglement
 of the corresponding quantum state is  revealed.
  \end{abstract}



\tableofcontents
\section{Introduction}
\label{sec:maxmult}
\subsection{Quantum mechanics and tensors}\label{subsec:qmten}
Quantum mechanics  was born in the beginning of twentieth century.  One way of its description was the ``matrix mechanics''  created by Heisenberg-Born-Jordan in 1925 \cite{Wikmm}.   The matrix formalism of quantum mechanics was created by von Neumann \cite{vNe27} and Landau \cite{Lan27}.
The probabilistic nature of quantum mechanics is puzzling, which led Einstein-Podolsky-Rosen to the notion of entanglement \cite{EPR35}.  Entanglement deals with 
quantum systems, for which there exists a physically distinguished
partition of the entire system into $d$ parts
with $d\ge 2$.  The mathematical description of $d$-partite systems involves the
theory of  tensors.  The common physics notation for tensors was introduced by Dirac in 1939 \cite{Dir39}.  It is very concise,  but it is unfamiliar to a good portion of mathematical community, and
 in several cases it is not adequate.

There are several ways to measure the entanglement of a $d$-partite system described by a $d$-tensor $\cT$, which in quantum mechanics
can be called a \emph{state}.  It is a nonzero tensor, which can be thought as a $d$-array with entries $\cT_{i_1,\ldots,i_d}$.  The bi-partite  case, $d=2$, corresponds to matrices, 
so a $2$-tensor will sometimes be denoted  by $T$. 
The case $d\ge 3$ corresponds to tensors with $d$ indices.  
For simplicity of exposition we will call $\cT$ any $d$-array for $d\ge 2$.

The simplest integer invariant of $\cT$ is its rank,  denoted as $r(\cT)$.  Rank-one tensor corresponds to a tensor product of $d$-vectors, which is called a \emph{product} state.   The rank of $\cT$ is the minimum number of summands in the decomposition of $\cT$ to a sum of rank-one tensor.  Thus the higher the rank of a tensor, the more entangled the corresponding quantum state.
A different way to decompose $\cT$ into a sum of rank-one tensors is to have a decomposition with the minimal norm, 
which in quantum physics might be related to energy.  This brings us to the notion of the  \emph{nuclear norm} \eqref{defnucnrm}.  
The minimal number of rank-one terms in the nuclear norm decomposition is called the nuclear rank, and denoted as $r^{\textrm{nucl}}(\cT)$.  
For matrices the rank and the nuclear rank are equal to the standard rank of matrices.  The nuclear norm of  a matrix is a sum of its singular values.  For tensors $d\ge 3$, while the description of the rank and nuclear rank are relatively easy, the numerical computation can be quite extensive.  In computer science termonology, the computation of all the above mentioned quantites in NP-hard.

The aim of this paper is to give a survey of many results on the rank of tensors, and to emphasize the connection to quantum mechanics and quantum information theory. This work is addressed to a broad audience consisting of computer scientists, mathematicians and  physicists. 
%
\subsection{Quantum information and related fields}\label{subsec:qitrf}
In this subsection we discuss in more details 
{\sl Quantum information theory} that our paper deal with.  Since the  publication of the pioneering paper of Ingarden   \cite{In76}, there was 
 an explosion in the theory of  quantum information
and related fields in the last twenty five years.
On the one hand one could witness enormous success related to 
quantum cryptography \cite{BB14} 
and practical implementations of quantum communication \cite{GT07}.
On the other hand  the progress in quantum computation is still moderate \cite{Pre18} 
and some experts raise doubts \cite{Ka16},
whether an operating quantum computer will be ever constructed.

A quantum computer would allow us to solve some problems that are not known to be solvable in polynomial time on a classical computer,
for instance factorization of large integers \cite{Sho97}.  
The practical impossibility of factoring integers to primes is the basis for the widely used RSA cryptographic scheme.  
One of the key advantages of processing of information by a quantum computer 
relies in the possibility 
of going beyond the standard set of digits $'0'$ and $'1'$,
and using nonclassical states, represented by normalized vectors
in a $n$--dimensional, complex Hilbert space ${\mathcal H}_n$. 
They include superposition of classical states,
which after a measurement yield the result $'0'$ with a certain probability $p$
and $'1'$ with probability $1-p$,
often written as $\sqrt{p}|0\rangle + \sqrt{1-p}|1\rangle$.
We used here the notation of Dirac, explained in details in section \ref{subsec:Diracnot},
in which $|0\rangle$ and $|1\rangle$ denote
any two normalized orthogonal vectors ---
the elements of an orthonormal basis in ${\mathcal H}_n$. 
Assumed normalization of vectors has a simple interpretation:
  The squared modulus of each component represents the probability
   of the state to be measured in the corresponding basis state
   and the sum of probabilities is equal to one.

Non-classical properties characterize also the notable {\sl Bell state}
\cite{Bel64}, written
 $(|0\rangle_A \otimes |0\rangle_B +|1\rangle_A \otimes |1\rangle_B)/\sqrt{2}$,
which exhibits quantum entanglement --- the effect
of non-classical correlations between subsystems,
predicted already by Einstein and Schr{\"o}dinger \cite{EPR35,Sch35,Sch36}.
Note that the outcomes of a measurement performed on
both subsystems are perfectly correlated:
if $'0'$ is registered in the subsystem $A$,
the same result will be obtained in the subsystem $B$.

To define {\sl quantum entanglement}
one requires a physical system
consisting of two (or more) distinguished subsystems \cite{Br02,HHHH09}, 
the state of which is described by
 a vector from a complex Hilbert space with a tensor product structure,
 ${\mathcal H}_{AB}= {\mathcal H}_A \otimes {\mathcal H}_B$.
A special role in such a space is played by the product states,
$|\psi_{AB}\rangle = |\phi_A\rangle \otimes |\phi_B\rangle$,
where  $|\phi_A\rangle  \in {\mathcal H}_A$ and  $|\phi_B\rangle  \in {\mathcal H}_B$.
A state $|\psi\rangle  \in {\mathcal H}_{AB}$
is called {\sl entangled} if it is {\sl not} of the product form,
as it carries some quantum correlations  between subsystems
-- for the definitions of necessary notions used in quantum 
mechanics see Appendix  \ref{AppA}.

The motivation for the definition of entanglement
between two predefined systems
stems from the classical probability theory:
 A product quantum state can be compared to a product probability vector,
 $p(x,y)=p_1(x) p_2(y)$,
 which describes independent variables $x$ and $y$,
 while entangled state corresponds to correlated events. 
 It is thus important to characterize degree of entanglement, 
 interpreted as a degree of quantum correlations between subsystems.

Product states allow one to construct a complete orthonormal product basis,
which reads  $|i\rangle_A\otimes |j\rangle_B$ with $i,j=0,\dots n-1$
 for a system consisting of two subsystems $A$ and $B$ with $n$ levels each.
Any pure state $|\psi_T\rangle$
in the tensor product space ${\cH}_A \otimes {\cH}_B$ 
 of a bipartite system (i.e. consisting of two subsystmes)
can be represented in a product basis by a matrix  $T$ of expansion coefficients,
$|\psi_T\rangle=\sum_{i,j=0}^{n-1} T_{ij}|i\rangle \otimes |j\rangle$.

Note that the Hilbert spaces $\cH_A$ and $\cH_B$ are identified with two vector spaces, while ${\cH}_A \otimes {\cH}_B$ is identified with the space of matrices. 
Similarly, the tensor product of $d\ge 3$ vector spaces $\otimes_{j=1}^d \cH_j$ is identified with the space of $d$-mode multiarrays (tensors).  A rank-one tensor  (matrix) corresponds to a product state.
%
\subsection{Matrices and tensors}\label{subsec:matten}
To characterize  entanglement of the state
$|\psi_T\rangle$ it is sufficient
to perform the singular value decomposition (SVD) of the matrix $T$: 
if there is more than one nonzero singular value then the corresponding quantum state is
not of the product form so it is entangled. 
In other words, the bipartite 
pure state is entangled if the rank of the corresponding
 matrix $T$ is larger than one  \cite{HHHH09}.
%

In a similar way any pure state of a multipartite system, consisting of $d>2$ subsystems,
 can be represented 
by a tensor $\cT$ with $d$ indices.  
A standard definition of a rank of  $\cT\ne 0$ 
is the minimum number of product states whose sum is $\cT$.
However, description of entanglement in such
 a multipartite case is  more difficult than for bipartite systems, 
as algebraic transformations on tensors are much more involved than the operations
 on matrices \cite{WGE16,BZ17}.  Tensor product and entanglement optimization are used in modern quantum chemistry \cite{Sz15}.

For matrices there exists a single notion
of the {\it rank of a matrix}, as various ways to introduce this quantity lead to the same
definition.   This is not the case for tensors with $d\ge 3$ modes,
for which different approaches lead to different definitions
of the {\it rank of a tensor} (with various names).  Furthermore, most interesting definitions of the rank of tensors, as the above definition of the rank, are hard to compute \cite{Has90,HL13,FL16,Shi16,SS18}.  In the present paper by `hard' we mean that the complexity of their computations are suspected to be at least NP-complete. 
  
A special phenomenon occurs for tensors that does not hold for matrices: There are tensors of rank greater than one such that   their rank decomposition is unique (up to permutation of order of rank-one summands).   Such tensors are called \emph{identifiable}.
It is conjectured that general tensors of rank less than the generic rank are identifiable \cite{COV14,COV17}. 

Another example is a generalization 
the notion of SVD for the space of tensors \cite{KB09,LMV00,NL08}.
However in general, these decompositions of a tensor as a sum of rank-one factors do not have the property that in each mode the vectors are orthogonal as in the matrix case.

An analysis of tensors as multiarrays can be traced to the two papers of Hitchcock \cite{Hit27, Hit27a} which introduce the notion of $d$-mode tensor and its rank.   In algebraic geometry the study of symmetric tensors, which is equivalent to homogeneous polynomials, was started by Sylvester \cite{Syl04}.  Sylvester discussed the decomposition of a homogeneous polynomial of degree $d$ in $2$-variables as a minimal sum of linear forms of degree $d$.  In modern terminology, a minimal decomposition of a homogeneous polynomial of degree $d$ as a minimal sum of  $d$-powers of linear terms is called the Waring decomposition \cite{La12}.  Waring  raised the general problem of minimum decomposition of an integer as a sum of $d$-powers of integers  \cite{Nat96}.

A good introduction of tensor decompositions in given by Kolda and Bader in \cite{KB09}.  
A reader  interested in results on tensor rank
is advised to take a look at this survey  \cite{KB09} and its extensive
list of 247 references.  Some of them are not directly related to the scope of this work,
and not quoted here. However, we provide a  few relevant
 references  \cite{Ber13,StA09,TBT09,Van17} 
 related to ranks of tensors that are not  mentioned in \cite{KB09}
 and are not discussed in the present  paper.
We also omitted the topic of rank of tensors 
related to the product of matrices, as discussed in \cite{La12}.


The purpose of this review is twofold: first, we aim to  introduce the problems of pure states entanglement to the community of 
 researchers working in applied and pure mathematics, and computer science.
 Second, we wish to present a survey of recent mathematical results concerning the rank of a tensor and its various generalizations that are relevant in studies
  of quantum entanglement.

Note that the rank of a tensor $r(\cT)$ with $d$ modes over the complex numbers, 
can be viewed as a simple integer quantity characterizing entanglement of a pure state of a quantum system composed out of $d$ subsystems. 
  The rank of a nonzero tensor ranges from $1$, which 
  characterizes a product (separable) state,
   to the maximum rank $r_{\max}$, defined for a given size of the tensor.
It is convenient to use the notion `rank of a pure quantum state' $|\psi\rangle$,  
 which means the rank of the corresponding tensor $\cT$.

 As most of the problems in tensors, the problem of 
 finding the rank of a given tensor is NP-hard \cite{Has90}.  A $d$-mode complex-valued tensor of dimensions $\n=(n_1,\ldots,n_d)$
 with its entries  chosen  independently at random (a generic tensor) 
 has a fixed rank $r_{\textrm{gen}}(\n)$  with probability one.  The value of $r_{\textrm{gen}}(\n)$ can be computed in randomized polynomial
  time using Terracini's lemma \cite{Ter16}. The original
  Terracini's lemma characterizes the generic Waring rank of a
   homogeneous polynomial of degree $d$ of $n$-variables over the complex numbers as discussed above.

There is a general conjecture about the value of the generic rank of $3$-mode tensors \cite{Fri12}, see Conjecture \ref{congenrank3}, which is known to hold in special cases \cite{AOP09}. 
 A generalization of this conjecture to $d$-mode tensors is stated as an open problem below \eqref{genrankdbd}.  A related problem is the value 
of the tensor rank of the tensor product of two tensors,
 linked to the issue of characterization of 
 entanglement and finding an optimal decomposition of
 several copies of a given entangled state \cite{CDS08,CCDJW10,YCGD10}.

 The rank of the tensor product of tensors is submultiplicative (the rank of the tensor product
  is not greater than the product of ranks), and can be strictly submulitplicative \cite{CJZ18}.  
 It is then important to ask whether 
  the rank of the tensor product of two generic tensors of the same dimension 
  is equal to the product of their ranks.
 The rank of the direct sum of two tensors is subadditive.  Strassen's direct sum conjecture stated
 that the rank of the direct sum of two tensors is the sum of the ranks \cite{Str73}.  This conjecture was
 disproved by Shitov \cite{Shi17}.  A special case of Strassen's conjecture can be related to the
 above product rank conjecture for the generic case, i.e., the rank of the product of $k$-copies of
 a generic tensor is the $k$-power of the rank of this tensor \cite{CF18}. 
 Yet another related  notion is the {\it border rank} of a tensor: It is the minimum $r$ such that a given tensor is a limit of tensors of rank $r$.
  This notion is fundamental in algebraic geometry and numerical analysis, in particular for problems of approximation
   of tensors by lower rank tensors \cite{KB09}.

A special interesting case is the case of symmetric tensors, which correspond in
quantum physics to problems involving bosons --- particles with an integer spin 
following the Bose--Einstein statistics. 
Symmetric tensors can be viewed as homogeneous polynomials \cite{La12,FW18}.
The relevant notion of the rank of a symmetric tensor, is the symmetric rank, which is the Waring rank
of the homogeneous polynomial discussed above.  Comon's conjecture claimed that the symmetric 
rank and the rank of a symmetric tensor are equal \cite{CGL08}. 
 Recently Shitov gave a counterexample to this conjecture \cite{Shi17}.  
The generic symmetric rank of a symmetric tensor is given in \cite{AH95},
which is the analog of Conjecture  \ref{congenrank3}.

Another way to measure the entanglement of a quantum pure state is the spectral or nuclear norm (see definitions \eqref{defspecnrm} and \eqref{defnucnrm})
 of the corresponding tensor \cite{FL18,DFLW17}, which are dual norms.  
The spectral norm, also known as the injective norm \cite{Rya02} of a tensor, 
determines  the geometric measure of entanglement \cite{WG03} of the 
corresponding pure quantum state.  The nuclear norm, 
also known as projective norm \cite{Rya02}, 
is the minimum value of the ``energy'' of 
a decomposition of a tensor as sum of rank-one tensors -- see the sentence below \eqref{defnucnrm}.
The minimum number of rank-one tensors in the minimal decomposition 
of a given tensor \eqref{defmindecnn} is called the nuclear rank, and denoted $r^{\textrm{nucl}}(\cT)$.   Hence $r(\cT)\le r^{\textrm{nucl}}(\cT)$.
For matrices the nuclear rank is equal to the rank.  
Nuclear rank of tensors possesses similar properties as matrix rank, unlike the tensor rank \cite{FL16}.   That is, if $\lim_{k\to\infty} \cT_k=\cT$ then $\liminf_{k\to\infty} r^{\textrm{nucl}} \cT_k\ge r^{\textrm{nucl}} (\cT)$.
Assume that there exists $\cT$ such that $r(\cT)$ is greater than the generic rank $r_{\textrm{gen}}$.  Then 
$r^{\textrm{nucl}}(\cT)\ge r(\cT)>r_{\textrm{gen}}$, and the nuclear rank of every tensor in a small enough neighborhood of $\cT$ is not less than  $r^{\textrm{nucl}}(\cT)$.  
Therefore there exists an open set of tensors whose nuclear rank is greater than their ranks.
The nuclear rank of the quantum state provides
 another simple integer measure of quantum entanglement. 

The application part of our paper is how to compute (or estimate)  the various ranks we survey. Apart of computing the generic ranks, which can be
done in randomized polynomial time, all other quantities seem to be NP-hard to compute.  A good way to find or estimate
 from below the rank and symmetric rank is to use polynomial equations, in particular the effective
  Nullstellensatz \cite{AF20} --- If $r$ is less than the rank of a tensor then the corresponding system of polynomials equations is not solvable -- see Subsection \ref{subsec:estrankpoleq}.  
 Hence the Bertini software \cite{BHSW06} is an appropriate tool. 
 To compute the nuclear rank we can use the numerical methods 
 and the software suggested in \cite{DFLW17}.

\subsection{Notation}\label{subsec:notation}
We use several notions of rank in various context, hence we present a table that collects all symbols used in the text
with reference to the page on which the objects appears for the first time:
\begin{center}
\begin{tabular}{ l | c | r }
\bf{description} & \bf{notation} & \bf{defined at page} \\
\hline
rank of a matrix $A$               & $r(A)$                            & \pageref{eq:def_rank} \\
rank of a tensor $\cT$             & $r(\cT)$                          & \pageref{eq:def_Tensorrank} \\
generic rank of $\cT\in\C^\n$      & $r_{\textrm{gen}}(\n)$            & \pageref{eq:def_genericrank}\\
maximum rank of $\cT\in\C^\n$      & $r_{\textrm{max}}(\n)$            & \pageref{eq:def_maximumrank}\\
border rank of $\cT\in\C^\n$                          & $r_{\textrm{b}}(\cT)$                     & \pageref{eq:def_borderrank}\\
Kruskal's rank of vectors $\x_1,\ldots,\x_l\in\C^m$   & $r_{\textrm{K}}(\x_1,\ldots,\x_l)$        & \pageref{eq:def_Kruskalrank}\\
symmetric rank of $\cS$                               & $r_{\textrm{s}}(\cS)$                     & \pageref{eq:def_symmetricrank}\\
generic rank of symmetric tensor                      & $r_{\textrm{gen}}(d,n)$                   & \pageref{eq:def_genericsymmetricrank}\\
maximum rank of symmetric tensor                      & $r_{\textrm{max}}(d,n)$                   & \pageref{eq:def_maximumsymmetricrank}\\
nuclear rank of $\cT$                                & $r^{\textrm{nucl}}(\cT)$                  & \pageref{eq:def_nuclearrank}\\
generic nuclear rank of $\cT\in\C^\n$                & $r_{\textrm{gen}}^{\textrm{nucl}}(\n)$    & \pageref{eq:def_genericnuclearrank}\\
maximum nuclear rank of $\cT\in\C^\n$                & $r_{\textrm{max}}^{\textrm{nucl}}(\n)$    & \pageref{eq:def_maximumnuclearrank}\\
symmetric nuclear rank of $\cS$                       & $r_{\textrm{s}}^{\textrm{nucl}}(\cS)$     & \pageref{eq:def_symmetricnuclearrank}\\
\end{tabular}
\end{center}
\subsection{Summary of the paper}\label{subsec:summary}
Section \ref{sec:prelres} discusses basic notions in quantum information theory (QIT) and tensors that will be used in this paper. Subsection \ref{subsec:Diracnot} discusses briefly the Dirac notation, and the notion of a simple physical system.  A simple system is represented by a vector, which is usually of length one (normalized),  and corresponds to a single party system. Subsection \ref{subsec:entangle} discusses complex systems and the notion of entanglement.  A composite system is represented by tensor product of $d$-simple components.  The simplest composite system consists of two parties ($d=2$), and usually referred as a bipartite system.  It is represented by a matrix, which usually has Frobenius norm one (normalized).  The bipartite state is entangled if and only if the corresponding matrix has rank greater than one.  For $d\ge 3$ a $d$-partite system is represented by a $d$-mode tensor, which usually has Hilbert-Schmidt norm one (normalized).  It is entangled if and only if 
the corresponding tensor has rank greater than one.  Subsection \ref{subsec:Kron} discusses the Kronecker product of $d$-vector spaces and their interpretation in QIT.

Section \ref{sec:tenrank} discusses properties and results on tensor rank of general tensors. Subsection \ref{subsec:matrank} recalls the standard properties of matrix rank.  Subsection \ref{subsec:SVD} discusses the well known properties of Singular Value Decomposition (SVD), which is known in physics community as Schmidt decomposition.  In subsection \ref{subsec:tenrank}  we discuss the notions of tensor rank, generic rank and border rank, and some of their properties.

In section \ref{sec:rank3ten} we discuss in details the rank properties of $3$-mode tensors. Subsection \ref{subsec:basres3ten} gives basic results on the rank of $3$-tensors.  In particular we discuss Kruskal's uniqueness theorem \cite{Kr77}.  
Subsection \ref{subsec:rankmn2ten} reviews the formula for the rank of $m\times n\times 2$ tensors discovered in \cite{Ja79}.  In subsection \ref{subsec:strasscon3} we discuss the validity of Strassen's direct sum conjecture for certain $3$-tensors.
Subsection \ref{subsec:genrank} discusses known results and a conjecture on the formula for the generic rank of $3$-tensors.  In subsection \ref{subsec:numprocrgenmnp} we give a randomized polynomial-time algorithm to compute the generic rank of $3$-tensor.  Subsection \ref{subsec:maxrank3t} discusses known results on maximal ranks of $3$-tensors.  

In section \ref{sec:rankdge4} we discuss ranks of $d$-tensors for $d\ge 4$. In subsection \ref{subsec:gencased} we bring known generalizations of results from subsection \ref{subsec:basres3ten}  to $d$-mode tensors.  In subsection \ref{TerraciniLemma} we recall Terracini's lemma, which is a basic tool in understanding and analyzing the set of tensors of rank $r$.  Subsection \ref{subsec:upbdgr} gives an upper bound on the generic rank of tensors using pure combinatorial methods. In subsection \ref{subsec:grqunit} we discuss the  generic rank of $d$-qunits ($d$-mode tensors whose each mode is $\C^n$).  Subsection
\ref{subsec:estrankpoleq} discusses an algorithmic way to find the rank of a tensor using solvability of system of linear equations with many variables.  In subsection \ref{subsec:genidgt} we discuss the problem of
generic identifiability of tensors.  Namely, assuming an integer $r$ is less than a generic rank, 
when a generic tensor of rank $r$ has a unique decomposition as a sum of $r$ rank one tensors?

In section \ref{sec:symten} we discuss symmetric tensors---bosons in physics.  Subsection \ref{subsec:hompol} discusses  basic properties  of symmetric tensors and their relation to homogeneous polynomials. In subsection \ref{subsec:maxsymrk}
we give a known general upper bound for the maximum symmetric rank in terms of generic symmetric rank, and some known values of maximum symmetric rank.
Subsection \ref{subsec:rankWd} shows that the rank of symmetric tensor $|W_d\rangle$ is $d$, while its border rank is $2$.  Here 
$|W_d\rangle\in\rS^d\C^2$ corresponds to the polynomial
$dx_1^{d-1}x_2$.
For instance, the three-qubit $|W\rangle$ state reads $|W_3\rangle=\frac{1}{\sqrt{3}}\left(|100\rangle +|010\rangle+|001\rangle\right)$.
In subsection \ref{subsec:rprodten} we show explicitly that the rank of Kronecker and tensor products can be strictly submulitplicative by considering the ranks of $|W\rangle\otimes_K|W\rangle$ and $|W\rangle\otimes|W\rangle$, which are $7$ and $8$ respectively, while the square of the rank of $|W\rangle$ is $9$.  In a short subsection \ref{subsec:compsymten} we discuss briefly computational methods for symmetric rank of symmetric tensors.  Subsection \ref{subsec:genidst} gives a short account of the results in \cite{COV}, which show that the generic identifiability property of symmetric tensors holds for a rank less than the symmetric generic rank, except a number of known cases.  

Section \ref{sec:nucrank} is mainly devoted to the notion of the nuclear rank of a tensor.  In subsection \ref{subsec:specnrm} we discuss the spectral norm and 
geometric measure of entanglement.  Subsection \ref{subsec:nucnorm} introduces the  nuclear norm and nuclear rank. In subsection \ref{subsec:face}  we discuss the faces of the unit ball of the nuclear norm.  Subsection \ref{subsec:matnnrm} discusses the exposed faces and facets of the unit balls of matrix nuclear and spectral norms.  In subsection \ref{subsec:GHZ} we show that the nuclear rank of the state 
$|GHZ\rangle=|000\rangle + |111\rangle/\sqrt{2}$
 is $2$.  Subsection \ref{subsec:gn3qub} discusses the  generic and maximum nuclear rank of symmetric $3$-qubits.

In Appendix \ref{AppA} we present definitions of some notions used in quantum theory
and discussed in this work.

\section{Preliminary results}\label{sec:prelres}
This Section is organized as follows: In subsection \ref{subsec:Diracnot} we introduce the Dirac notation and the notion of a simple quantum system aimed at mathematicians.  In subsection \ref{subsec:entangle} we describe a composite systems consisting of $d$-simple systems.  Mathematically this is equivalent to the introduction of $d$-tensor product of the corresponding vector spaces.  Next we introduce the notion of entanglement, which corresponds to a tensor of rank greater than one.  We discuss the differences between the bipartite case that corresponds to matrices and $d$-partite case for $d\ge 3$. In subsection \ref{subsec:Kron}  we discuss Kronecker tensor product of tensor spaces and their quantum interpretation.

 \subsection{Dirac notation and simple systems}\label{subsec:Diracnot}
 Let $\cH$ be a finite-dimensional Hilbert space with an inner product  $\langle \x,\y\rangle$.  (So $\cH$ is a vector space  over the field of complex numbers $\C$. We assume that the inner product is linear in $\x$ and bar linear in $\y$, see below.)  We denote by $\cH_n$ an $n$-dimensional Hilbert space.  Let  $\e_1,\ldots,\e_n$ be an orthonormal basis in $\cH_n$.  We identify $\cH_n$ with the column space $\C^n=\{\x=(x_1,\ldots,x_n)\trans\}$, where $\x=\sum_{i=1}^n x_i\e_i$,  $\y^*=(\bar y_1,\ldots,\bar y_n)$ and $\langle \x,\y\rangle=\y^*\x$. Here $\bar z$ stands for the complex conjugate of $z\in\C$.  
 Thus one can write
 \begin{equation*}
 \langle a\x+b\z,\y\rangle=a \langle\x,\y\rangle+b \langle\z,\y\rangle,\quad 
 \langle \x, a\y+b\z\rangle=\bar a \langle\x,\y\rangle+\bar b \langle\x,\z\rangle.
 \end{equation*}
The notation of Dirac, which is routinely  used in the literature on quantum physics, 
can be summarized as follows. The symbol $|i\rangle$ 
represents a basis vector in $\cH_n$
and usually stands for $\e_{i+1}$ for $i \in \{0,\ldots,n-1\}$, or for $\e_i$, where $i\in[n]=\{1,\ldots,n\}$. For simplicity of our exposition the latter convention will be used. 

  A vector in $\cH_n$ is often  denoted as
 $|\psi\rangle=\sum_{i=1}^n x_i|i\rangle$.
 In quantum theory such a vector describes a physical system
 and it is called a {\sl pure quantum state},
 but for brevity the shorter versions  as 
 a `pure state' or a `state' are also used.  
Thus $|\psi\rangle $ corresponds to $\x$ and 
$\langle \psi|$ corresponds to $\x^*$.  Furthermore $\langle \phi |\psi\rangle$ is the inner 
product $\langle \psi, \phi\rangle$.
 For physical applications one assumes that $|\psi\rangle  \ne 0$
 and the states are normalized, 
$\langle \phi |\phi\rangle=||\phi||^2=1$,
 but  sometimes non-normalized states will also be used here.
A given quantum system is described by a state 
in a $n$ dimensional space, $|\psi \rangle  \in {\cH}_n$,
also denoted as $\C^n$.
In physics parlance the spaces $\C^2$, $\C^3$ and $\C^n$ 
are often called the space of \emph{qubits}, \emph{qutrits} and \emph{qunits}, respectively.
As the overall phase is physically not relevant,
 a quantum pure state refers to entire equivalence class, 
 $|\psi\rangle \sim e^{i \alpha} |\psi\rangle$,
 with $\alpha \in [0, 2\pi)$.
Thus the space of normalized pure quantum states 
forms a complex projective manifold, ${\mathbb C}P^{n-1}$, which is 
obtained from $\C^n$ by identifying each complex line through the origin
 as a point in ${\mathbb C}P^{n-1}$.  It has complex dimension $n-1$ and
  real dimension $2(n-1)$.
In the simple case of a single qubit, $n=2$,
the space of pure states is called the {\sl Bloch sphere}, 
 ${\mathbb C}P^{1}=S^2$.

Observe that $|\psi\rangle\langle \phi |$ stands for the corresponding
matrix $\x\y^*$ of rank one.
In particular  $|\psi\rangle\langle \psi |=P_{\psi}$ represents
a projector onto a normalized state $|\psi\rangle$, and it is easy to see
that $P_{\psi}^2=|\psi\rangle\langle \psi |\psi\rangle\langle \psi|=P_{\psi}$,
 where we used the fact that  $\langle \psi |\psi\rangle=1$.
  Besides pure states, in quantum physics one works
with convex combinations of projectors onto pure states,
$\rho=\sum_i a_i |\psi_i\rangle\langle \psi_i |$
with nonnegative weights,  $a_i \ge 0$, $\sum_i a_i=1$,
called {\sl mixed states}.
They are represented by positive semi-definite density matrices, $\rho^*=\rho\ge 0$
-- see Appendix  \ref{AppA} --
but in this work  they appear only sporadically.

\subsection{Composite systems and quantum entanglement}
\label{subsec:entangle}
A key axiom of quantum theory states that a quantum system
composed of two subsystems $A$ and $B$, of dimension $n_1$ and $n_2$,
respectively, is described in a  tensor product Hilbert space,
${\cH}_{n_1 n_2}={\cH}_{n_1} \otimes {\cH}_{n_2}$,
sometimes written ${\cH}_A \otimes {\cH}_B$.  For instance,
 if the system consists of two particles $A$ and $B$ each represented as vectors $\x\in \cH_A, \y\in \cH_B$ then the joint system $AB$ is 
 represented by their tensor product
 or a matrix $X\in \cH_A\otimes\cH_A$. Note that $X=\x\y^\top$ if and only if both
 subsystems  are not entangled, see below.  
 The composite system $AB$ is said to consists of two parts $A$ and $B$.
In a similar way, a physical system composed of $d$ parts
is represented in a tensor product  Hilbert space with $d$ factors.

The shortness of Dirac notation is transparent when considering 
 tensor product of $d$ Hilbert spaces:
  $\cH=\otimes_{k=1}^d \cH_{n_k}=\cH_{n_1}\otimes\cdots\otimes\cH_{n_d}$
 of product dimension, $\prod_{j=1}^d n_j$.
 This tensor product is viewed in quantum physics as the $d$-partite space.  
 Assume that $|\psi_k\rangle\in \cH_{n_k}$ for $k\in[d]$.
The product vector (state) is denoted as 
\[|\psi_1\rangle|\psi_2\rangle\cdots|\psi_d\rangle=|\psi_1\rangle\otimes\cdots\otimes|\psi_d\rangle=\otimes_{k=1}^d|\psi_k\rangle.\]
Assume that $\e_{1,k},\ldots,\e_{n_k,k}$ is an orthonormal basis of $\cH_{n_k}$.  
Then 
$\otimes_{k=1}^d \e_{i_k,k}$ for $i_1\in[n_1],\ldots,i_d\in[n_d]$ 
is an orthonormal basis in $\cH$. In Dirac notation the vector
 $\otimes_{k=1}^d \e_{i_k,k}$ is denoted as
$|i_1\cdots i_d\rangle$ or $ \vert i_1\rangle \otimes \vert i_2\rangle \otimes 
 \dots \otimes \vert i_d\rangle$.  
   
 Any pure quantum  state $|\psi\rangle$ in $\cH$ (a vector of length one)
corresponding to a physical system composed of  $d$ subsystems 
with $n$ levels each and described in a vector space of dimension $n_1 n_2\cdots n_d$,
can be written in a product basis,
\begin{equation} 
\vert \psi\rangle = \sum_{i_1=1}^{n_1}  \sum_{i_2=1}^{n_2} \dots \sum_{i_d=1}^{n_d}
 \cT_{i_1,i_2,\dots, i_d} \; \vert i_1\rangle \otimes \vert i_2\rangle \otimes 
    \dots \otimes \vert i_d\rangle .
\label{tensor11}
\end{equation}
\noindent
Thus, for a multipartite system,  $d>2$, the state $|\psi\rangle$ corresponds to a tensor 
  $\cT\in \otimes_{k=1}^d \C^{n_k}$, whose coordinates 
  in the standard basis are $\cT_{i_1,\ldots,i_d}\in\C$.  
  In the special cases of simple systems, $d=1$,
  or bipartite systems, $d=2$, the tensor  $\cT$ reduces to a vector or to 
  a  matrix, respectively.
The normalization condition,
$\langle\psi| \psi \rangle = ||\psi||^2 =1$,
implies that the Hilbert-Schmidt norm (also known as Frobenius norm for matrices)
\begin{equation}\label{defHSnorm}
\|\cT\|_2=\sqrt{ \sum_{i_1, \dots, i_d=1}^{n_1,\dots, n_d}     | \cT_{i_1,i_2,\dots, i_d}|^2}
\end{equation}
is fixed to unity,  $\|\cT\|_2=1$.  
   
In many parts of this paper we will restrict ourselves to the most interesting case in quantum setting: $n_1=\cdots=n_d=n$.   This 
means that each subsystem lives in the Hilbert space   $\cH_{n}$ of the same size $n$.
Then $\cH$ is denoted either $\cH_n^{\otimes d}$ (physical notation) 
or $\otimes^d \cH_n$
(mathematical notation).  Furthermore, assume that $|\psi_1\rangle=\cdots=|\psi_d\rangle=|\psi\rangle$ then $\otimes_{k=1}^d|\psi_k\rangle$ is denoted either $|\psi\rangle^{\otimes d}$ or $\otimes^d|\psi\rangle$.

In an experimental setting one may change 
the physical partition of the entire
system into a different composition of subsystems.  This change
corresponds to a reshaping tensor in such a way that the total number
of elements is preserved. For instance, a matrix $6\times 6$ describes a bipartite
$6\times 6$ system, while a four-index tensor $\sum_{i_1,i_2,i_3,i_4}\cT_{i_1,i_2,i_3,i_4}|i_1i_2i_3i_4\rangle$ with $i_1,i_2\in\{1,2\}$
and $i_3,i_4\in\{1,2,3\}$ represents a $2\times 2\times 3\times 3$ system composed of two qubits and two qutrits.

A quantum state  $|\psi\rangle \in {\cH}$ is called {\it  separable}
 (and hence non-entangled)  if the state has a product form,
$|\psi\rangle=|\phi_1\rangle \otimes |\phi_2\rangle \otimes \cdots \otimes |\phi_d\rangle$,
so that the rank of the  corresponding tensor  $\cT$ is equal to one.
In all other cases the state is called {\it entangled},
as it has no product form
-- see  Appendix  \ref{AppA}.  

Note  that the term {\it entanglement} has a meaning if the tensor product structure 
is specified. Then the physical partition of the entire system into $d$ subsystems is fixed,
and the  product basis $|i_1 i_2  \dots i_d\rangle=\vert i_1\rangle \otimes \vert i_2\rangle \otimes 
    \dots \otimes \vert i_d\rangle$
in which the state  (\ref{tensor11}) is represented is well defined. 

Characterization of the degree of entanglement is
easier in the case of  bipartite systems. 
A normalized pure state of an $n \times n$ system
can be represented in a product basis by a complex square matrix, 
$|\psi_{AB}\rangle = \sum_{i,j=1}^n T_{ij}|i,j\rangle$,
where states $|i\rangle$  with $i\in[n]$ form an orthonormal basis
in the first subsystem $A$,
while an analogous basis $|j\rangle$ refers to the second subsystem $B$.
Complex expansion coefficients $T_{ij}$ form a vector of length $n^2$,
but it is often convenient to treat them as a square matrix of size $n$.
The quantity $|T_{ij}|^2$ can be viewed as a probability that the system $AB$ 
is measured in the state $|ij\rangle$.

The state $|\psi_{AB}\rangle$
 is called {\sl separable} if and only if the rank of the matrix $T$ is one,
so the state has the product form,   
$|\psi_{AB}\rangle=|\phi_{A}\rangle \otimes |\phi_{B}\rangle$. 
Separability refers to the fact that the  joint physical system $AB$
can be then divided into two separate parts, $A$ and $B$,
and the results of measurements performed separately
on both of  them are not correlated.
Note that the space of separable pure states 
is equivalent to the Cartesian product of two complex
projective manifolds \cite{BZ17}, which forms a Segre embedding, 
 ${\mathbb C}P^{n-1} \times {\mathbb C}P^{n-1} \subset
 {\mathbb C}P^{n^2-1}$.
In the simplest case of two-qubit system, $n=2$,
the set of separable states forms the Cartesian
 product of two Bloch spheres,
 $S^2 \times S^2 \subset  {\mathbb C}P^{3}$.

A pure state which is not separable is called {\sl entangled}
 and one may introduce several measures of quantum entanglement \cite{HHHH09,BZ17},
 which aim to quantify, to what extend a given bipartite state  $|\psi_{AB}\rangle$ 
is {\sl not} of the product form. 
In general,  entanglement measures are not equivalent, 
in the sense that their maximal values are attained for different states.
The simplest quantity is given by  the rank $r$ of the corresponding matrix $T$ of size $n$,
but it is not a smooth function of the state.
Another  possibility is to deal with various norms of $T$.
Assumed normalization of the state,
 $\langle \psi_{AB}|\psi_{AB}\rangle=1$,
implies the following constraint for the Frobenius (Hilbert--Schmidt) norm, 
$||T||_{2}^2=\tr T^*T=1=\sum_{i=1}^n \lambda_i$.
Here $\lambda_i$ denote
the eigenvalues of the positive semidefinite  matrix $T^*T$, 
while $\sigma_i=\sqrt{\lambda_i}$ represent the singular values of $T$.
They arise by the singular value decomposition,
$T=UDV^*$, where $U$ and $V$ are unitary, while
$D$ is a diagonal matrix with non-negative entries
$\sigma_i$ at the diagonal -- see Section \ref{subsec:SVD}.
Then the corresponding bipartite state can be written
using the {\sl Schmidt decomposition},
$|\psi\rangle =\sum_{j=1}^r \sigma_j |j\rangle \otimes |j'\rangle$, 
equivalent to Eq. (\ref{tensor11}), 
where unitaries $U$ and $V$ determine the basis
$|j\rangle$ and $|j'\rangle$, respectively.

Let  $\lambda_{\rm max}$ denote the largest eigenvalue of  $T^*T$,
 so that the spectral norm of $T$ is given by the largest singular value, 
 $||T||_{\infty}=\sigma_{\rm max}=\sqrt{\lambda_{\max}}$. 
Then the state $|\psi_{AB}\rangle$ is separable 
if and only if  $\lambda_{\max}=\sigma_{\rm max}=1$,
so that, the smaller the norm  $||T||_{\infty}$,
(under the restriction that the Frobenius  norm of $T$ is fixed), 
the larger entanglement. Thus, as a measure of entanglement one 
can take a suitable smooth function of $\lambda_{\max}$ (or $\sigma_{\rm max}$),
for instance  $1-\lambda_{\rm max}$
or  $-\log( \lambda_{\rm max})$
advocated  in \cite{WG03}.
The Schmidt decomposition implies 
that the  maximal overlap of the  analyzed state with any product state reads,
$ \max_{| \psi_{sep}\rangle} 
|\langle \psi_{AB} | \psi_{sep} \rangle   |^2=\lambda_{\rm max}$,
where 
$| \psi_{sep} \rangle =|\phi_A\rangle  \otimes | \phi_B \rangle$.
This scalar product determines the minimal distance
of the analyzed state to the manifold of separable (product) states,
 a quantity  called  {\sl geometric measure of entanglement} \cite{Shi95,ZB02,WG03}.
 
In the case of pure states
the natural geodesic distance on ${\mathbb C}P^{n^2-1}$
is equivalent to the  Fubini--Study distance, $D_{FS}^{\rm min}=\arccos(|  \langle \psi_{AB}|\psi_{CD} \rangle |  )$
 also called the quantum angle, as it corresponds to the angle between 
 $|\psi_{AB}\rangle$ and $|\psi_{CD}\rangle$.
One   may also analyze various distances in the space of mixed states (density matrices)
between the projector on the analyzed state, $\rho_{\psi_{AB}}=|\psi_{AB}\rangle \langle \psi_{AB}|$,
and the projector on the closest product state.
Depending on the distance selected \cite{LS02,ZB02}, 
one obtains $D_{HS}^{\rm min}=\sqrt{2(1-\lambda_{\rm max})}$ for the smallest 
Hilbert--Schmidt distance between two states $\rho$ and $\omega$,
namely $D_{HS}(\rho, \omega)=[\tr (\rho-\omega)^2]^{1/2}$.
The smallest value with respect to the trace distance,
$D_{1}(\rho, \omega)=\tr |\rho-\omega|$ with $|X|=\sqrt{X^*X}$
reads
$D_{1}^{\rm min}=2 \sqrt{1-\lambda_{\rm max}}$,
while the Bures distance, 
$D_{B}(\rho, \omega)=[2-2 \tr|\sqrt{\rho} \sqrt{\omega}|]^{1/2}$,
leads to 
$D_{B}^{\rm min}= [2(1-\sigma_{\rm max})]^{1/2}$
-- see also Section \ref{subsec:specnrm}.

Another, more precise entanglement measure can be obtained
from the entire vector of squared singular values of $T$,
which sum to unity,  $\sum_i\lambda_i=1$.
 Hence one can write down the entropy of this vector, 
$S(|\psi_{AB}\rangle)=-\sum_{i=1}^n \lambda_i \log \lambda_i$,
called {\sl entanglement entropy} of the
bipartite pure state $|\psi_{AB}\rangle$.
This quantity, being a continuous function of $|\psi_{AB}\rangle$
and having several  information-theoretical interpretations,
is often considered as a distinguished measure of
 the bipartite entanglement \cite{HHHH09,BZ17}.
 For instance, in the simplest case of two-qubit system,
 $\cH_4=\cH_2\otimes \cH_2$,
 the entropy is maximized for the {\sl Bell entangled state} \cite{Bel64},
 $|\phi_{AB}^+\rangle = (|00\rangle + |11\rangle)/\sqrt{2}$,
 for which $\lambda_1=\lambda_2=1/2$, so that
$S(|\phi_{AB}^+\rangle)=\log 2$.
Recall that $|00\rangle$ is a useful shorthand for
the product state, often written in various ways,
 $|0\rangle_A \otimes |0\rangle_B =|0\rangle \otimes |0\rangle= |0,0\rangle=
 |00\rangle $.

The vector of singular values of the matrix $T$,
  representing the analyzed state,  allows one to compute its trace norm,
 $||T||_{1}= \tr\sqrt{T^*T}=\sum_{i=1}^n \sigma_i$.
Since the bipartite state $|\psi_{AB}\rangle$ is separable 
if and only if $\sigma_{\rm max}=1$ so that $||T||_{1}=1$,
 to construct an alternative measure of bipartite  entanglement
 one can consider the quantity $||T||_{1}-1$. 
 In short,  
 under the normalization assumption,  $||T||_{2}=1$,
 the larger trace norm of $T$, the more entangled state.
 In the case of  a two-qubit system the maximal value of the trace norm,
 $||T||_{1}^{\max}=\sqrt{2}$, is achieved for the Bell entangled state 
 $|\phi_{AB}^+\rangle=(|00\rangle+|11\rangle)/\sqrt{2}$.

For composite systems  with $d\ge 3$ parts,  a similar strategy does not work
as the singular value decomposition of a matrix
has no direct generalization for a tensor with $d$ indices \cite{LMV00a,CHS00}.
However,  the rank of the tensor
 can still serve as one of the simplest measures of
quantum entanglement, as 
 in this setting the rank  $r(\cT)$ is given by the minimal natural number
$r$ such that the  corresponding state  (\ref{tensor11})
can be represented as a superposition of $r$ product states,
\begin{equation} 
\vert \psi\rangle = \sum_{i=1}^{r}  
a_i  \vert \phi_1^{(i)} \rangle \otimes \vert \phi_2^{(i)} \rangle \otimes 
    \dots \otimes \vert  \phi_d^{(i)}\rangle ,
\label{tensor22}
\end{equation} 
with arbitrary complex  coefficients $a_i$. Note that the states $|\phi_j^{(i)}\rangle$
related to the subsystem number $j$, with $1\le j \le d$, need not be orthogonal.
In physics literature one uses the term
{\sl  rank of a composite pure state}, 
which is equal to the rank of the corresponding  tensor $\cT$.
Furthermore, to quantify entanglement of a multipartite state
one uses the {\sl Schmidt measure} \cite{EB01,WGE16},
equal to the log of the rank of the corresponding tensor,
 $E_S(|\psi\rangle)=\log r(\cT)$.
More precise description of multipartite entanglement
can be obtained by studying the 
spectral norm  $||\cT||_{\infty}$ of a tensor \eqref{defspecnrm}, under the assumption that the Hilbert-Schmidt norm \eqref{defHSnorm} is fixed to unity.
An alternative approach is based on the nuclear norm $||\cT||_1$ of a tensor \eqref{defnucnrm},
which can be considered as a generalization of the
matrix trace norm for tensors.

\subsection{Kronecker tensor product}\label{subsec:Kron}
As discussed in \S\ref{subsec:Diracnot}, the tensor product $\cH_A\otimes \cH_B$ corresponds to a
bipartite space, while $\cH=\otimes_{k=1}^d\cH_{n_k}$ corresponds to a $d$-partite space.  Denote by $\cH_B^{\vee}=\{|\psi\rangle^{\vee}, |\psi\rangle\in\cH_B\}$ the dual space to $\cH_B$.  That is, $|\psi\rangle^\vee$ is viewed as a linear function on $\cH_B$: $|\psi\rangle^{\vee}(|\phi\rangle)= \langle \psi| \phi\rangle$.
Then bipartite product $\cH_A\otimes \cH_B^\vee$,  is identified in a classical way with the space of linear operators $L(\cH_A,\cH_B)=\{L,\;L:\cH_B\to \cH_A\}$.  Namely,
$|\theta\rangle\otimes|\psi\rangle^{\vee}(|\phi\rangle)=\langle \psi|\phi\rangle |\theta\rangle$.  It is customary to abuse the notation by identifying $\cH_m\otimes\cH_n$ with the space of $m\times n$ matrices $\C^{m\times n}$.

 Let $A=[A_{i,j}]\in \C^{m\times n}$ and $B=[B_{k,l}]\in \C^{p\times q}$.  
Then the Kronecker tensor product $C=A\otimes_K B$ is the block matrix $C=[A_{i,j}B]\in \C^{(mp)\times (nq)}$.  (Note that $A\otimes B$ is a $4$-tensor in $\C^m\otimes\C^n\otimes\C^p\otimes\C^q$ with $(A\otimes B)_{i,j,k,l}=A_{i,j} B_{k,l}$.)  This is equivalent to
$(\cH_{m}\otimes \cH_{n})\otimes_K(\cH_p\otimes \cH_q)=\cH_{mp}\otimes \cH_{nq},$
where $\cH_m\otimes\cH_p$ and $\cH_n\otimes\cH_q$ are viewed as Hilbert spaces $\cH_{mp}$ and $\cH_{nq}$ respectively.  More generally, given $l$ tensor product spaces $\cH_{n_{1,i},\cdots, n_{d,i}}=\otimes_{j=1}^d \cH_{n_{j,i}}$ for $i\in[l]$, the Kronecker tensor product is defined as the following $d$-tensor product
\begin{eqnarray}\label{defkronprod}
\otimes_K^{i\in[l]} \cH_{n_{1,i},\cdots, n_{d,i}}=\otimes_K^{i\in[l]}(\otimes_{j=1}^d \cH_{n_{j,i}})=\otimes_{j=1}^d (\otimes_{i=1}^l \cH_{n_{j,i}}).
\end{eqnarray}
Note that for $d=2$, the above definition reduces to the standard fact that Kronecker product of $l$ matrices is a matrix.
In \cite{CJZ18} the Kronecker product $\otimes_K$ is denoted by $\boxtimes$.

The Kronecker product \eqref{defkronprod} has the following quantum interpretation \cite{CF18}.
Consider a group of $d$ people that share 
 $l$ states with $d$-subsystems each.
 This means that  the person $j$ controls $l$ subsystems, 
 corresponding to the subspaces $\cH_{n_{j,1}},\ldots,\cH_{n_{j,l}}$. 
 These subsystems  
 are described by vectors in the product Hilbert 
 space 
 $\cH_{n_{j,1}\cdots n_{j,l}}=\otimes_{i=1}^l \cH_{n_{j,i}}$.  
Thus we can interpret the total system shared by $d$ of users,
each  controls $l$ subsystems,  
as a $d$-partite system on the corresponding Hilbert spaces.

It is possible to define the Kronecker tensor product  of $l$  tensor product spaces $\cH_{n_{1,i},\cdots, n_{d_i,i}}$ for $i\in[l]$, where $d_1,\ldots, d_l$ are different.  Set $d=\max\{d_1,\ldots,d_l\}$. For $d_i<d$ set $\cH_{d_i+k,i}=\C, n_{d_i+k,i}=1$ for $k\in[d-d_i]$.    Define $\cH_{n_{1,i},\ldots,n_{d,i}}=\otimes_{j=1}^d \cH_{n_{j,i}}$.  Let  $\cH_i'$ be a tensor product obtained from $\cH_{n_{1,i},\ldots,n_{d,i}}$ by permuting the factors $\cH_{d_i+1,i},\ldots,\cH_{d,i}$ with other factors.  For $d_i=d$ let $\cH_i'=\cH_{n_{1,i},\ldots,n_{d,i}}$.  Then 
$\otimes_K^{i\in[l]} \cH_{n_{1,i},\cdots, n_{d_i,i}}=\otimes_K^{i\in[l]}\cH_i'$.

\section{Tensor rank}\label{sec:tenrank} 
In this section we discuss basic notions and results on matrices and tensors, with the emphasis on the notion of tensor rank.  In subsection \ref{subsec:matrank} we discuss the well known characterizations and properties of matrix rank.  In subsection 
\ref{subsec:SVD} we recall the properties of Singular Value Decomposition (SVD), which is known in the  physics community as Schmidt decomposition.  We also discuss the geometric measure of entanglement of a bipartite state, which has a simple formula in terms of 
the operator norm ($\sigma_{\max}$) of the corresponding matrix.  
The maximally entangled bi-partite Bell state has the minimum $\sigma_{\max}(A)$.  
In subsection \ref{subsec:tenrank}  we discuss the notion of the rank of tensor.  We point out the submultiplicativity of  the rank of tensors under the tensor and Kronecker tensor product.   An example of strict submulitplicativity is given in Lemma \ref{UVtenrin}.  A similar result holds for the subadditivity of tensor rank under the direct sum.  In some cases the strict inequality holds, that is the Strassen's direct sum conjecture is false.   The notions of generic, maximum and border ranks of tensors are also reviewed.
\subsection{Matrix rank}\label{subsec:matrank}
Let $A\in\C^{m\times n}$ be a nonzero matrix.  Then $\rank A$ also written $r(A)$\label{eq:def_rank} is the minimum $k\in\N$ such that 
$A=\sum_{i=1}^k\x_i\y_i^*$, where $\x_i\in\C^m, \y_i\in\C^n$ for $i\in[k]$. 
  Equivalently, for a bipartite state $A\in\cH_m\otimes\cH_n$ the rank of $A$ is the minimum number of summands in the decomposition of $A$ as a sum of the product states.  The rank of zero matrix is $0$.  The rank of a product state is $1$, and the matrix $\x\y^*$, for $\x\in\C^m\setminus\{\0\},\y\in\C^n\setminus\{\0\}$, is called rank-one matrix. 
  
It is quite simple to find the rank of a matrix $A$ using Gauss elimination.  Hence the complexity of finding the rank of $A$ is $O(\min(m,n)^2\max(m,n))$ in exact arithmetic.  Better complexity results can be found in \cite{BCS97}.
There are many equivalent ways to define the rank of a matrix.  We bring together a few of the equivalent definitions and some related inequalities \cite{Fribook,HiL13}:
\begin{lemma}\label{eqderrankmat}
Let $A\in\C^{m\times n}$.   Then each of the integers below is $r(A)$:
\begin{enumerate}
\item The dimension of the row space (subspace spanned by the rows of $A$).
\item The dimension of the column space,  (subspace spanned by the columns of $A$).
\item The dimension of a maximal nonzero minor of $A$. (Minor of $A$: a determinant of a square submatrix of $A$.)
\item The rank of $PAQ$ for two invertible matrices $P$ and $Q$ of dimension $m$ and $n$ respectively.  
\end{enumerate}
Furthermore,
\begin{enumerate}
\item Assume that $B$ is a submatrix of $A$ (obtained by deleting some rows and columns of $A$). Then $r(B)\le r(A)$.
\item Assume that $P\in \C^{m\times m}, Q\in\C^{n\times n}$.  Then $r(PAQ)\le r(A)$.
\item Assume that $A_k\in \C^{m\times n}$ for $k\in\N$, and $\lim_{k\to\infty} A_k=A$.

Then $\liminf_{k\to\infty} r(A_k)\ge r(A)$.
\end{enumerate}
\end{lemma}
The last statement of Lemma \ref{eqderrankmat} is the lower semicontinuity of the matrix rank. 

Let $A\in\C^{m\times n}, B\in \C^{p\times q}$.  Then $A\oplus B$ is the block diagonal matrix $\left[\begin{array}{cc}A&0\\0&B\end{array}\right]\in\C^{(m+p)\times (n+q)}$.
Hence the dimension of the column space of $A\oplus B$ is the sum of the dimensions of the column space of $A$ and $B$, which yields the equality,
 $r(A\oplus B)=r(A)+r(B)$.

A decomposition of $A$ as a sum of $r(A)$ rank-one matrices  is called a rank decomposition. If $r=r(A)>1$ then this decomposition is not unique.  (We ignore the order of the summands.)  For example, if we choose a basis $\x_1,\ldots,\x_r\in\C^m$ in the column space of $A$, then there exists unique basis $\y_1,\ldots,\y_r\in\C^n$  of the column space of $A^*=\bar A\trans$ such that $A=\sum_{i=1}^r \x_i\y_i^*$.  

Let $\gl_n\subset \C^{n\times n}$ be the group of invertible matrices.  Denote by 
\begin{equation}\textrm{orb}(A)=\{B=PAQ, P\in\gl_m,Q\in\gl_n\},\end{equation}
the orbit of $A$ under the action of $\gl_m\times\gl_n$.  Since any basis $\x_1,\ldots,\x_m$ in $\C^m$ is of the form $P\e_1,\ldots,P\e_m$ for a unique $P\in\gl_m$ it follows that orb$(A)$ is the set (quasi (algebraic) variety) of all matrices of rank $r(A)$.  Furthermore, the closure of orb$(A)$ is the (algebraic) variety of all matrices of at most $r(A)$.  In terms of quantum physics this statement
 is written  that the SLOCC transformations of a given bipartite state 
 of rank $r$ do not increase its rank.
The above abbreviation stands for
 {\sl stochastic local operations and classical communication},
 as these experimentally realizable transformations
  play a key role in the theory of quantum information  \cite{BZ17}. 

Observe next that if we choose at random a matrix $A$ in $\C^{m\times n}$ (where each entry has a standard Gaussian distribution) then $r(A)$ is $\min(m,n)$ with probability $1$.  That is, the value of the maximal minor of $A$ of order $\min(m,n)$ is nonzero with probability $1$.  In the language of algebraic geometry, the generic rank of $\C^{m\times n}$ is $\min(m,n)$.  Note that $\min(m,n)$ is also the maximal possible rank of matrices in $\C^{m\times n}$.   That is a generic bipartite state is maximally entangled, if the rank of the tensor is considered as a simple discrete measure of quantum entanglement.
 There exist also continuous measures of entanglement including
 the geometric measure  -- see the next subsection.

The rank of a matrix behaves nicely under the Kronecker tensor product \cite{Fribook}:
\begin{eqnarray}
\label{ranktneprodmat}
r\left(\otimes_{K}^{i\in[l]} A_i\right)=\prod_{i=1}^l r(A_i).
\end{eqnarray}
The reason for that is very simple.  Observe that column space of $B=\otimes_{K}^{i\in[l]} A_i$ is the tensor product of the column spaces of $A_1,\ldots,A_l$.  Hence the dimension of the column space of $B$ is the product of the dimension of the column spaces of $A_1,\ldots,A_l$.

\subsection{SVD or Schmidt decomposition}\label{subsec:SVD}
There is a standard way to make a minimal rank decomposition unique in a generic case.  This is the Singular Value Decomposition (SVD) (in mathematics) or Schmidt decomposition (in physics): For a given $A\in\C^{m\times n}$ there exists a decomposition
\begin{eqnarray}\label{SVDdec}
A=\sum_{i=1}^{r(A)} \sigma_i(A)\bu_i\bv_i^*, \quad \bu_i^*\bu_j=\bv_i^*\bv_j=\delta_{ij}, \quad i,j\in[r(A)].
\end{eqnarray}
Note that $\bu_1,\ldots,\bu_{r(A)}$ and $\bv_1,\ldots,\bv_{r(A)}$ are orthonormal bases of the column spaces of $A$ and $A^*$ respectively. Furthermore $\sigma_1(A)\ge \cdots \ge\sigma_{r(A)}(A)>0$ are the positive singular values of $A$.  Note that
\begin{equation}A A^*\bu_i=\sigma_i(A)^2\bu_i, \quad A^*A\bv_i=\sigma_i(A)^2\bv_i, \quad i\in[r(A)].\end{equation}
That is, the square of the positives singular values of $A$ are the positive eigenvalues of $AA^*$ and $A^*A$.  In particular, the decomposition \eqref{SVDdec} is unique if and only if $\sigma_1(A)> \cdots> \sigma_{r(A)}(A)$ \cite{Fribook}. 
By uniqueness we mean here that each rank-one matrix $\bu_i\bv_i^*$ is unique, 
but $\bu_i$ and $\bv_i$ need not be unique.

We now recall various approximation properties of the SVD decomposition of $A$.  Recall that $\C^{m\times n}$ is \red{a} Hilbert space with the inner product $\langle A, B\rangle=\tr B^*A$, where the trace of a square matrix $C=[C_{i,j}]\in\C^{m\times m}$ is given as $\tr C=\sum_{i=1}^m C_{i,i}$.  Then the Frobenius norm 
(also called Hilbert--Schmidt norm)
of $A$ is given by $\|A\|_F=\sqrt{\tr A^*A}=\sqrt{\sum_{i=1}^{r(A)}\sigma_i(A)^2}$.  The operator norm of $A$ is given by
\begin{align}
\|A\|&=\sigma_1(A)=\max\{\|A\x\|, \|\x\|=1\}\\
&=\max\{|\y^* A\x|, \|\x\|=\|\y\|=1\}=\max\{\Re(\y^* A\x), \|\x\|=\|\y\|=1\}\nonumber.
\end{align}
The real part of a complex number $z$ is denoted by $\Re z$.

Denote by $\Pi(m,n)\subset \cH_m\otimes\cH_n$ all normalized  product states
with the norm set to $1$.
  Assume that $|\psi\rangle\in \cH_m\otimes\cH_n$ is a normalized state.  The geometric measure of entanglement 
can be described \cite{ZB02,WG03} by
the Hilbert--Schmidt distance of $|\psi\rangle$ to $\Pi(m,n)$:
\begin{align}
&\min\{\||\psi\rangle-|\xi\rangle\otimes|\eta\rangle\|,\||\xi\rangle\|=\||\eta\rangle\|=1\}=\\
&\min\{\sqrt{2-2\Re\langle\xi|\psi\rangle|\eta\rangle},\||\xi\rangle\|=\||\eta\rangle\|=1\}=\sqrt{2(1-\sigma_1(|\psi\rangle))\nonumber}.
\end{align}
Hence the maximally entangled states with respect to the geometric measure of entanglement are the Bell states $|\psi\rangle$, which are characterized by $\sigma_i(|\psi\rangle)=\frac{1}{\sqrt{\min(m,n)}}$ for $i\in [\min(m,n)]$.

Furthermore, for each $k\in[r(A)]$ let $B_k\in\C^{m\times n}$ be any element such that $r(B_k)=k$.  Then
the distance of $A$ from the orbit orb$(B_k)$, or its closure, is $\sigma_{k+1}(A)$, and is achieved at $A_k:=\sum_{i=1}^k \sigma_i(A) \bu_i\bv_i^*$. 
Recall that $\sigma_j(A)=0$ for $j>r(A)$. See for example \cite{Fribook}.
\subsection{Definition of a rank of a tensor}\label{subsec:tenrank} 
Let $d>2$ be a positive integer.  Assume that $n_1,\ldots,n_d$ are positive integers.
Let $\n=(n_1,\ldots,n_d)$ and denote $\cH_\n=\otimes_{i=1}^d \cH_{n_i}$ and $\C^{\n}=\otimes_{i=1}^d \C^{n_i}$. The dimension of these vector spaces is $N(\n)=\prod_{i=1}^d n_i$.  A non-normalized $d$-product state $\otimes_{i=1}^d|\psi_i\rangle\in\cH_\n$ corresponds to a rank-one tensor $\otimes_{i=1}^d\x_i\in\C^{\n}\setminus\{0\}$.  Note that $\otimes_{i=1}^d\x_i\in\C^{\n}$ is the zero tensor if and only if at least one of $\x_i$ is a zero vector.  Assume that  $\cT\in\C^{\n}$ is the pure state $|\psi\rangle\in\cH_{\n}$.
Then $\cT$ has a representation \eqref{tensor11}.  The rank of a nonzero tensor $\cT\in\C^{\n}$, denoted as $r(\cT)$\label{eq:def_Tensorrank}, is the minimal number of summands in the representation of $\cT$ as a sum of rank-one tensors.  Equivalently, the rank of the state $|\psi\rangle\in\cH_{\n}$ is the minimal dimension of a subspace spanned by normalized product states that contains $|\psi\rangle$.  The equality \eqref{tensor11} yields that $r(\cT)\le N(\n)$.   Actually, a stronger inequality is known -- see \S\ref{subsec:gencased}:
\begin{eqnarray}\label{upbundrank}
r(\cT)\le \frac{N(\n)}{\max(n_1,\ldots,n_d)}, \quad \cT\in\C^{\n}.
\end{eqnarray}
While the definition of the rank of the tensor is in principle the same as for matrices, the calculation of the rank of a given tensor can be NP-hard even for $3$-tensors \cite{Has90}.   

Let $\bp=(p_1,\ldots,p_d)\in\N^d$.  Assume that $\cU=[\cU_{j_1,\ldots,j_d}]\in\C^{\bp}$.  Recall that $\cT\otimes\cU\in\C^{\bq}$, where $\bq=(\n,\bp)$.
On the other hand the tensor $\cT\otimes_K\cU\in \C^{\n\cdot\bp}$, where $\n\cdot\bp=(n_1p_1,\ldots,n_dp_d)$.  From the rank minimal decomposition of $\cT$ and $\cU$ we deduce the obvious inequalities \cite{CF18}
\begin{eqnarray}\label{rnaktenKrnin}
r(\cT\otimes_K\cU)\le r(\cT\otimes \cU)\le r(\cT)r(\cU).
\end{eqnarray} 
Recall that $\cV=\cT\oplus\cU$ is a tensor in $\C^{\n+\bp}$, such that
$(\cT\oplus\cU)_{i_1,\ldots,i_d}=\cT_{i_1,\ldots,i_d} \textrm{ for } i_k\in[n_k], k\in[d]$,
$(\cT\oplus\cU)_{n_1+j_1,\ldots,n_d+j_d}=\cU_{j_1,\ldots,j_d} \textrm{ for } j_k\in[p_k], k\in[d]$,
and all other entries are zero.  It is straightforward to show that $r(\cT\oplus \cU)\le r(\cT)+r(\cU)$.  Recall that for $d=2$ we have equality in the above inequality.
In \cite{Str73}  Strassen asked if $r(\cT\oplus \cU)= r(\cT)+r(\cU)$ for $3$-tensors.
For general $d$ this problem is sometimes called \emph{Strassen's direct sum conjecture}.  For $d=3$ this is true if $\min(n_1,n_2,n_3,p_1,p_2,p_3)\le 2$, see \S\ref{subsec:strasscon3}.
Some additional cases where Strassen's direct sum conjecture holds are discussed in \cite{BPR19,Tei15}. For border rank this conjecture
 was known to be wrong \cite{Sch81}, see below.
Recently Shitov showed that  even for $d=3$   the direct sum conjecture of 
Strassen  is in general false \cite{ShiS17}.
Let $\oplus^k\cT$ be the direct sum of $k$ copies of $\cT$. 
By definition
 $r(\oplus^k\cT)\le k r(\cT)$
 and \emph{the restricted Strassen's conjecture}  \cite{CF18} is asking,
 whether equality holds,
  $r(\oplus^k\cT)=kr(\cT)$?   
  
  It was shown in \cite{CF18, Sch81} that this equality can be stated in the following form.  Let $\cI(k,d)\in (\C^k)^{\otimes d}$ be the identity tensor: $\cI(k,d)=\sum_{i=1}^k |i\rangle^{\otimes d}$.  One can show that 
$\cI(k,d)\otimes_K\cT$ is $\oplus^k\cT$, if we use the lexicographical order on the standard bases $(\C^k)^{\otimes d}\otimes_K \C^{\n}$
and $r(\cI(k,d))=k$. (It follows from the observation that if we view $\cI(k,d)$ as a matrix in $\C^{k \times k^{d-1}}$, then this matrix has rank $k$.)
Hence the restricted Strassen conjecture (which is still open) is equivalent to
\begin{eqnarray}\label{rStrascon}
r\left(\cI(k,d)\otimes_K\cT\right)=r(\cI(k,d)) r(\cT)=kr(\cT).
\end{eqnarray}  
Assume that the above equality holds for some $\cT$.  Observe that \eqref{rnaktenKrnin} and
 \eqref{rStrascon} imply that equalities in \eqref{rnaktenKrnin}  hold.  This implies \cite{CF18}
 that  $r(\cI(k,d)\otimes \cT)=kr(\cT)$.

  Consider the following simple example.  Let $\cT$ be $2\times 2\times 2$ with entries $\cT_{i,j,k}$ where $i,j,k\in[2]$.  Then the tensor $\cU=\cT\oplus \cT$ is $4\times 4\times 4$ tensor. To give the exact formula for the entires of $\cU$  we relabel $\{1,2,3,4\}$ as the pairs as  $\{(1,1),(1,2),(2,1), (2,2)\}$.  Then the entries of $\cU$ are $\cU_{(i_1,i_2),(j_1,j_2),(k_1,k_2)}$.  These entries are zero unless $i_1=j_1=k_1=l\in[2]$.  The formula for possible nonzero entries of $\cU$ is $\cU_{(l,i_2),(l,j_2),(l,k_2)}=\cT_{i_2,j_2,k_2}$.    The nonzero entries of $2\times 2\times 2$ tensor $\cI(2,3)$  are the entries $(1,1,1)$ and $(2,2,2)$ which are equal to $1$.  It is straightforward to see that $\cI(2,3)\otimes \cT=\cU$. Corollary \ref{varestStrcon} claims that in this example the restricted Strassen conjecture holds.

The \emph{generic} rank of a tensor in $\C^\n$, denoted as $r_{\textrm{gen}}(\n)$\label{eq:def_genericrank}, is the rank of a tensor $\cT\in\C^{\n}$ whose entries are chosen at random, assuming that the entries of tensors in $\C^\n$ are $N(\n)$ independent Gaussian random variables.  We will justify later (\S\ref{subsec:genrank}) the existence  of generic rank, and discuss briefly how to compute efficiently this rank using Terracini's lemma \cite{Ter16}. 
For example the generic rank of an $m\times n$ matrix is $\min(m,n)=r_{\textrm{gen}}(m,n)$.
It is well known that $r_{\textrm{gen}}(2,2,2)=2$, and the state
 $|GHZ\rangle=\frac{1}{\sqrt{2}}(|000\rangle+|111\rangle)=\frac{1}{\sqrt{2}}\cI(2,3)$ serves as an
example of a state with such a rank \cite{Fri12} -- see the discussion after Lemma \ref{ranks222ten}.
The \emph{maximum} rank of a tensor in $\C^\n$, denoted as $r_{\textrm{max}}(\n)$\label{eq:def_maximumrank}, is the maximum possible rank of tensors in $\cT\in\C^{\n}$. 
By definition, for tensors
$r_{\textrm{gen}}(\n)\le r_{\textrm{max}}(\n)$,
while for matrices  equality holds.
Furthermore, the maximal rank for a three-qubit state reads $r_{\textrm{max}}(2,2,2)=3$ and the state $|W\rangle=\frac{1}{\sqrt{3}}(|100\rangle+|010\rangle+|001\rangle)$ saturates the bound -- see discussion following Lemma \ref{ranks222ten}.
Nice results in \cite{BHMT17} state that 
\begin{eqnarray}\label{maxgenrin}
r_{\textrm{max}}(\n)\le 2r_{\textrm{gen}}(\n)-1.
\end{eqnarray}  
See example 16 (2) on page 7 in \cite{BHMT17}. For  $\n=(2,2,2)$ the generic rank is $2$ and $\sigma_{r_{\textrm{gen}}-1}$ is the variety of rank-one tensors.  This variety has projective dimension $3$ in the space of projective dimension $7$.
Note that for $\n=(2,2,2)$ this inequality boils down to $3<4$.  We will outline a short proof of the weaker inequality $r_{\textrm{max}}(\n)\le 2r_{\textrm{gen}}(\n)$ \cite{BT15} later. 

We now discuss the \emph{border} rank of $\cT\in\C^{\n}$, denoted as $r_{\textrm{b}}(\cT)$\label{eq:def_borderrank}, which was discovered in \cite{BLR}.  It is the smallest $k\in\N$ with the following properties:  There exists a sequence $\cT_j, j\in\N$ such that $r(\cT_j)=k$ for all $j\in\N$, and $\lim_{j\to\infty}\cT_j=\cT$.  
By definition, inequality
 $r_{\textrm{b}}(\cT)\le r(\cT)$ holds, which is always saturated for  matrices. 
Rank-one tensor satisfies the equality $r_{\textrm{b}}(\cT)= r(\cT)$.  That is, the set of all tensors of rank one and norm one is closed.
We will show later that $r_{\textrm{b}}(\cT)\le r_{\textrm{gen}}(\n)$ for any $\cT\in\C^\n$.
It is known that $r_{\textrm{b}}|W\rangle =2$, see \S\ref{subsec:basres3ten}.   Actually, this result follows from the above remarks.  
The border rank is subadditive:
$r_{\textrm{b}}(\cT\oplus \cU)\le r_{\textrm{b}}(\cT)+r_{\textrm{b}}(\cU)$, and the inequality  may be strict \cite{BCS97, Sch81}.
Thus the conjecture of Strassen for border rank is false.   More about algebraical methods and criteria of determining the border rank of tensors with
border rank not greater than two, can be found in \cite{Ra11}.

Denote by $\gl(\n)$ the product group  $\gl_{n_1}\times \cdots\times \gl_{n_d}$.  
Then $\gl(\n)$ acts on $\C^{\n}$ as follows. 
Let $(A_1,\ldots,A_d)\in\gl(\n)$.  Assume that $\cT=\sum_{i=1}^{r'} \otimes_{j=1}^d \x_{j.i}$.  Then $(A_1,\ldots,A_d)(\cT)=\sum_{i=1}^{r'} \otimes_{j=1}^d (A_j\x_{j.i})$.  The orbit of $\cT$ under this action 
reads
\begin{equation}\orb(\cT,\gl(n))=\Big\{\sum_{i=1}^{r'} \otimes_{j=1}^d (A_j\x_{j.i}),(A_1,\ldots,A_d)\in\gl(\n)\Big\}.\end{equation} 

This orbit corresponds to all states equivalent to $\cT$ under the SLOCC operations \cite{BZ17}.
Note that each tensor in $\orb(\cT,\gl(n))$ has rank $r(\cT)$.   The closure of the orbit of $\cT$ is denoted by Closure$(\orb(\cT,\gl(n)))$, it contains the set of the states that can be obtained from $\cT$ by SLOCC.  It can happen that this closed set may contain tensors of rank greater than $r(\cT)$. 
 For instance, 
the closure of the set of $2\times 2 \times 2$ tensors obtained from the  $|GHZ\rangle$ state
by changing bases in each copy of $\C^2$ forms  the entire
set of all $2\times 2\times 2$ tensors,
 see \S\ref{subsec:basres3ten}.

To illustrate the challenges 
of finding rank of $d$-mode tensors and other ranks of tensors we present a small survey on the ranks of $3$-tensors.

\section{Ranks of $3$-tensors}
\label{sec:rank3ten}
In this section we survey several known results on the rank of $3$-tensors of dimension $(m,n,p)$ related to quantum entanglement.
 In subsection \ref{subsec:basres3ten} we bring the celebrated Kruskal's theorem that gives a necessary condition
 that a decomposition of a tensor into a sum of rank-one tensors 
 is  unique up to a permutation of the summands.  Subsection \ref{subsec:rankmn2ten} covers the results of J\'aJ\'a \cite{Ja78,Ja79}, which completely characterize the rank of $m\times n\times 2$ tensors.  
 Similar results were obtained by Grigoriev \cite{Gri78,Gri78a}.
 In Subsection \ref{subsec:strasscon3} we recall the results of 
 J\'aJ\'a-Takche  \cite{Ja86} and Buczy\'nski-Postinghel-Rupniewski \cite{BPR19},
 which give sufficient conditions under which the direct sum conjecture of Strassen holds. 
   Subsection \ref{subsec:genrank} discusses the values of generic rank for $3$-tensors.   Theorem \ref{genmnpbig} combined with \eqref{genmnpvbig} implies
    that the generic rank of  such a three-tensor is equal  to $p$ if $(m-1)(n-1)+1\le p\le mn$.
 For $2\le m\le n\le p\le (m-1)(n-1)$ we state a well known conjecture 
 concerning  the value of generic rank.  This conjecture holds in some cases. 
  In subsection \ref{subsec:numprocrgenmnp} we describe a known algorithm to find the generic rank.  Subsection \ref{subsec:maxrank3t} surveys briefly some known results on maximal ranks of $3$-tensors.

\subsection{Basic results on rank of $3$-tensors}\label{subsec:basres3ten}
Assume that $d=3$ and $\n=(m,n,p)$.  Since $3$-tensors represent three-partite system, the order of the parties: Alice, Bob and Charlie is arbitrary.  In some cases we are going to assume 
\begin{eqnarray}\label{mnpord}
2\le m \le n \le p.
\end{eqnarray}
(The reason for the assumption that $m\ge 2$ is that for $m=1$ a $3$-tensor is a matrix.)
Given a $3$-tensor $\cT=[\cT_{i_1,i_2,i_3}]\in\C^{\n}$ we can associate with it four kind of ranks.  The first rank is $r(\cT)$, while the other three ranks  $r_A(\cT)$, $r_B(\cT)$ and $r_C(\cT)$ are corresponding matrix ranks.
Let us first consider $r_C(\cT)$.
View the two parties $\{A,B\}$ (Alice and Bob) as one party, which corresponds to the Hilbert space $\cH_{mn}$.  Then $\cT$ is viewed as a bipartite state $\cT_C\in\cH_{mn}\otimes\cH_p$.  It has $p$ columns $T_k=[\cT_{i,j,k}]_{i=j=1}^{m,n}\in\C^{m\times n}$ for $k\in [p]$.  Each column is a matrix, and $T_k$ is called a frontal slice.  The collection of the $p$ columns $\{T_1,\ldots,T_p\}$ can be viewed as an album of $p$ photos, where the matrix $T_k$ is $k$-th photo.  Then $r_C(\cT)$ is the dimension of the subspace in $\C^{m\times n}$ spanned by $T_1,\ldots,T_p$.  We next observe that $r_C(\cT)\le r(\cT)$.  Indeed, a singular value decomposition  of $T_C$ is
\begin{equation}\cT_C=\sum_{k=1}^{r_C(\cT)}\sigma_k(\cT_C)U_k\otimes \z_k, \quad \tr U_j^*U_k=\z_j^*\z_k=\delta_{jk}, \quad j,k\in[r_C(\cT_C)].\end{equation}
Note that here $U_j$ does not have to be a rank-one matrix.  Observe next that a rank decomposition $\cT=\sum_{i=1}^{r(\cT)} \x_i\otimes\y_i\otimes \z_i$ induces a decomposition $\cT_C=\sum_{i=1}^{r(\cT)} (\x_i\otimes\y_i)\otimes \z_i$ to rank-one vectors in $\cH_{mn}\otimes \cH_p$.  

The ranks $r_A(\cT)$ and $r_B(\cT)$ are defined similarly.
Hence
\begin{equation}\max(r_A(\cT), r_B(\cT), r_C(\cT)) \le r(\cT).\end{equation}

Assume that
\begin{eqnarray}\label{rankone3tendec}
\cT=\sum_{i=1}^r \x_i\otimes\y_i\otimes\z_i.
\end{eqnarray}
Under what conditions $r=r(\cT)$?  A simple sufficient condition is: the set of the matrices $\x_1\otimes\y_1,\ldots,\x_r\otimes\y_r$ and the set of vectors $\z_1,\ldots,\z_r$ are linearly independent.  Indeed, this condition insures that $r_C(\cT)=r\le r(\cT)$.  

Kruskal's conditions \cite{Kr77} gives sufficient conditions for $r=r(\cT)$, and that the above rank decomposition of $\cT$ is unique: 
Any rank decomposition of $\cT$ is a sum of rank-one tensors $\x_1\otimes\y_1\otimes\z_1, \ldots,\x_r\otimes\y_r\otimes\z_r $ in any order.
To state Kruskal's condition we need to define Kruskal's rank of $l$ nonzero
vectors $\x_1,\ldots,\x_l\in\C^m$, denoted as $r_{\textrm{K}}(\x_1,\ldots,\x_l)$\label{eq:def_Kruskalrank}.  Namely, $r=r_{\textrm{K}}(\x_1,\ldots,\x_l)$ if and only if any $r$ vectors in $\{\x_1,\ldots,\x_l\}$ are linearly independent, and there are $r+1$ vectors in  $\{\x_1,\ldots,\x_l\}$ that are linearly dependent.  For example, if $\x_1,\ldots,\x_l\in\C^m$ are chosen at random then $r_{\textrm{K}}(\x_1,\ldots,\x_l)=\min(l,m)$.  We call a set $\{\x_1,\ldots,\x_l\}\subset\C^m$ \emph{generic}, or in \emph{general position}, if $r_{\textrm{K}}(\x_1,\ldots,\x_l)=\min(l,m)$.
We call a decomposition \eqref{rankone3tendec} generic if the three sets of vectors 
$\{\x_1,\ldots,\x_r\},\{\y_1,\ldots,\y_r\},\{\z_1,\ldots,\z_r\}$ are generic.

Theorem of Kruskal \cite{Kr77} yields:
\begin{theorem}\label{Kruskalcond} Let $\cT\in\C^{m\times n\times p}$ have a decomposition \eqref{rankone3tendec}, where each $\x_i,\y_i,\z_i$ is nonzero.
If
\begin{eqnarray*}\label{Kruskalcond1} 
r_{\textrm{K}}(\x_1,\ldots,\x_{r}) + r_{\textrm{K}}(\y_1,\ldots,\y_{r}) +r_{\textrm{K}}(\z_1,\ldots,\z_{r})\ge 2r+2
\end{eqnarray*}
then $r=r(\cT)$ and the decomposition \eqref{rankone3tendec} is unique.
\end{theorem}
Note that for $m=n>1$, $p=2$ and $r=m$ this result is sharp.  Indeed, assume that 
$r_{\textrm{K}}(\x_1,\ldots,\x_m)=r_{\textrm{K}}(\y_1,\ldots,\y_m)=m$.  Note that $r_{\textrm{K}}(\z_1,\ldots,\z_m)=l$ where $l\in[2]$.  If $l=2$, that is any pair $\z_i,\z_j$ is linearly independent then Kruskal's theorem claims that $r(\cT)=m$ and the decomposition \eqref{rankone3tendec} is unique.  Assume now that $\z_1,\ldots,\z_l$ are nonzero colinear vectors.  Then $\cT$ is a matrix of the form $T=\sum_{i=1}^m a_i\x_i\otimes\y_i$ where each $a_i\ne 0$.  Thus $r(\cT)=r(T)=m$ but the rank decomposition of $T$ is not unique, since the rank decompositon of a matrix of rank greater than one is not unique.  Hence the rank decomposition of $\cT$ is not unique.  See \cite{Der13} for more examples showing that Kruskal's theorem is sharp.
See  \cite{Rho10} for a simple proof of Kruskal's theorem.  We will discuss Kruskal's theorem for $d$-mode tensors, where $d>3$, later.  

The following corollary follows from Kruskal's theorem:
\begin{corollary}\label{gendec3ten}  Let $\cT\in \C^{m\times n\times p}$.  
Assume that a decomposition of $\cT$ is generic.  If 
\[\min(r,m)+\min(r,n)+\min(r,p)\ge 2r+2\] 
then $r=r(\cT)$ and the rank decomposition of $\cT$ is unique.
\end{corollary}
We call a $3$-tensor $\cT$ that satisfies the conditions of the above corollary a rank-$r$ tensor with the generic decomposition.

The following theorem \cite{Fri12} gives a characterization of the rank of $3$-tensor:
\begin{theorem}\label{charrank3ten}  Let $\n=(m,n,p)$.  Assume that $\cT\in \C^\n$, and let $T_1,\ldots,T_p\in\C^{m\times n}$ be the $p$ frontal slices of $\cT$.  Then $r(\cT)$ is the minimum dimension of a subspace of $\C^{m\times n}$ spanned by rank-one matrices that contains the subspace spanned by $T_1,\ldots,T_p$.
\end{theorem}

A composite space
$\C^{m\times n}$ is spanned by $mn$ linearly independent tensors of rank one.
Hence $r(\cT)\le mn$.  This yields the inequality \eqref{upbundrank} for $d=3$, as we can assume \eqref{mnpord}.

Denote by $r_{\textrm{max}}(m,n,p)$ the maximum possible rank of tensors in $\C^{m\times n\times p}$.  The inequality \eqref{upbundrank} yields that  $r_{\textrm{max}}(m,n,p)\le \frac{mnp}{\max(m,n,p)}$.
\begin{proposition}\label{valranks}   Let $m,n,p\in\N$. For each $k\in\{1,\ldots,r_{\textrm{max}}(m,n,p)\}$  there exists a tensor $\cT\in\C^{m\times n\times p}$ such
that $r(\cT)=k$.
\end{proposition}
\begin{proof}  Assume that $\cA\in\C^{m\times n\times p}$, and $r(\cA)=r=r_{\textrm{max}}(m,n,p)>1$. 
Write $\cA=x_1\otimes y_1\otimes z_1+\dots+x_r\otimes y_r\otimes z_r$ and $\cB_k=x_1\otimes y_1\otimes z_1+\dots+x_{k}\otimes y_{k}\otimes z_{k}$ for $k\in[r-1]$. 
Clearly, $r(\cB_k)\le k$.  Use the rank decomposition of $\cB_k$, and the fact that $\cA=\cB_k+\sum_{j=k+1}^r x_j\otimes y_r\otimes z_j$
to deduce that $r=r(\cA)\le r(\cB_k)+r-k$ which yields that $r(\cB_k)\ge k$.  Hence $r(\cB_k)=k$.
\end{proof}

 Let us assume that $p=2$.  So $\cT$ has two frontal slices $T_1,T_2\in\C^{m\times n}$.  Let us first examine all possible nonzero ranks of $2\times 2\times 2$ tensors in terms of the matrices in the subspace spanned by the frontal sections:  
 \begin{lemma}\label{ranks222ten} Let $\n=(2,2,2)$ and assume that $\cT\in\C^{\n}\setminus\{0\}$.  Suppose that $T_1,T_2\in\C^{2\times 2}$ are the two frontal slices of $\cT$.  Then
 \begin{enumerate}
 \item $r(\cT)=1$ if and only if $T_1$ and $T_2$ are linearly dependent, and  one of the slices is rank-one matrix.
 \item $r(\cT)=2$ if and only if one of the following conditions hold: 
 \begin{enumerate}
 \item The matrices $T_1$ and $T_2$ are linearly dependent, and  one of the slices is rank-two matrix.
  \item The matrices $T_1$ and $T_2$ are linearly independent, and each matrix in span$(T_1,T_2)$ is singular.
  \item The subspace span$(T_1,T_2)$ contains two linearly independent matrices $X,Y$ such that $X$ is invertible and $X^{-1}Y$ is diagonalizable.
 \end{enumerate}
 \item $r(\cT)=3$ if and only if $T_1,T_2$ are linearly independent, and the span$(T_1,T_2)$ contains two matrices $X,Y$such that $X$ is invertible and $X^{-1}Y$ is not diagonalizable.
 \end{enumerate}
 \end{lemma}  
 
 Consider a rank-two tensor $\cT=\x_1\otimes\y_1\otimes \z_1+\x_2\otimes\y_2\otimes \z_2$ with the generic decomposition.    Corollary \ref{gendec3ten} yields that the rank decomposition of $\cT$ is unique.  
 The orbit of the tensor $\cT$ with respect to 
  the general linear transformations, written
 $\orb(\cT,\gl)$, consists of all rank-two  tensors with the generic decomposition.
  In particular, the GHZ (non-normalized) state $|GHZ\rangle=|111\rangle +|222\rangle$ is a rank-two state with the generic decomposition.  Furthemore the closure of $\orb(|GHZ\rangle,\gl)$ is $\C^\n$. Let $W$
  be the (non-normalized) state
 \begin{equation}|W\rangle=|112\rangle+|121\rangle+|211\rangle=M_1\otimes \e_1+M_2\otimes \e_2\end{equation}
 with
 \begin{equation}M_1=\left[\begin{array}{cc}0&1\\1&0\end{array}\right], M_2=\left[\begin{array}{cc}1&0\\0&0\end{array}\right].\end{equation}
 As $M_1$ is invertible and $M_1^{-1}M_2$ is not diagonalizable, it follows that $r(|W\rangle)=3$. This is essentially \cite[Corollary 3.2.1]{Ja79} after a change of variables in the first factor $\C^2$.  Note that $r_{\textrm{K}}(|1\rangle,|1\rangle,|2\rangle)=1$.  Hence
  the above decomposition of $|W\rangle$ fails to satisfy the conditions of Kruskal's theorem: The sum of Kruskal ranks reads $1+1+1=3$ and $2\times 3+2=8$.   It is easy to show that the above rank decomposition of $|W\rangle$ is not unique.
 It is also known that $r(\cT)=3$ if and only if $\cT\in\orb(|W\rangle,\gl)$. 
(It is enough to show that if $r(\cT)=3$ then $(A_1,A_2,A_3)(\cT)=|W\rangle$ for some $(A_1,A_2,A_3)\in\gl(\n)$.)
 
 We give a short proof of this claim.  Assume that $\cT=T_1\otimes \e_1+T_2\otimes \e_2$ for some $T_1,T_2\in\C^{2\times 2}$.  First note that the action of $A_3$ on $\cT$ is equivalent to choose a different basis $T_1',T_2'$ in span$(T_1,T_2)$.  Choose $X=aT_1+bT_2$ to be invertible, $Y=cT_1+cT_2$ be such that $X^{-1}Y$ is not diagonalizable.  Then $X^{-1}Y$ has a double eigenvalue $\lambda$.  Set $Z=Y-\lambda X$.  Then $X^{-1}Z$ is a rank-one nondiagonalizable matrix.  Hence $\cT_1=(I_2,I_2,A_3)(\cT)=X\otimes \e_1+Z\otimes \e_2$ for some $A_3\in\gl_2$.  Observe next that for $B_1,B_2\in\gl_2$ we obtain $\cT_2=(B_1,B_2,I_2)\cT'=B_1 XB_2\trans\otimes\e_1+B_1ZB_2\trans\otimes \e_2$.  
 Choose $B_1,B_2$ such that $B_1XB_2\trans=I_2$.  Then $\cT_2=I_2\otimes \e_1+C\otimes\e_2$.  Observe that $C$ is similar to $X^{-1}Z$.  As $X^{-1}Z$ is a rank-one nondiagonalizable matrix it follows that $C$ is similar to the Jordan block $M_1^{-1}M_2$.  That is $C=Q^{-1}(M_1^{-1}M_2)Q$.  Let $\cT_3=(Q^{-1},Q\trans,I_2)(\cT_2)=I_2\otimes\e_1+(M_1^{-1}M_2)\otimes\e_2$.
 Finally, $(M_1,I_2,I_2)(\cT_3)=|W\rangle$.
 
 The equality
 \begin{equation}|W\rangle=\lim_{t\to 0}\frac{1}{t} \Big(\big(|1\rangle +t|2\rangle\big)^{\otimes 3}-|1\rangle ^{\otimes 3}\Big)\end{equation}
 shows that the border rank of $|W\rangle$ is $2$.
 \subsection{The rank of $m\times n\times 2$ tensors}\label{subsec:rankmn2ten}
 In this subsection we describe the complete solution to the problem of computing the rank of $\cT=[\cT_{i,j,k}]\in \C^{m\times n\times 2}$.  Denote by
$T_1=[\cT_{i,j,1}]$, $T_2=[\cT_{i,j,2}]\in\C^{m\times n}$ the frontal slices of $\cT$.
It was shown by J\'aJ\'a \cite{Ja79} how to apply the Kronecker theory of the canonical form of a pencil of matrices \cite{Kro90} to determine the rank of a $m\times n\times 2$  tensor.  We present some results of \cite{Ja79} using the
 notions and the results discussed above.

Theorem \ref{charrank3ten} states that $r(\cT)$  is the minimal dimension of  a subspace spanned by rank-one matrices in $\C^{m\times n}$, which contains the subspace $\bV=$ span$(T_1,T_2)$.
We can assume that $T_1$ and $T_2$ are linearly independent, otherwise the rank of $\cT$ is 
$\max(r(T_1),r(T_2))$.  Then one can change a basis in span$(T_1,T_2)$ to $T_1'$, $T_2'$.  This is equivalent to  considering $\cT'=(I_m,I_n, A_3)(\cT)$ corresponding to some $A_3\in\gl(2)$.  Next we consider $\cT_1=(P,Q\trans,I_2)(\cT')$, where $P\in\gl(m)$, $Q\in\gl(n)$.  This corresponds to replacing the pair $(T_1',T_2')$ by $(PT_1'Q,PT_2'Q)$.  It is a classical problem to find the canonical form of a pair of matrices $(A,B)\in \C^{m\times n}\times \C^{m\times n}$ under the simultaneous equivalence: $(A,B)\mapsto P(A,B)Q$, where $P\in\gl(m)$, $Q\in\gl(n)$.  This problem was solved completely by Kronecker \cite{Kro90}.   See the classical exposition in \cite{Gan59}, or a short exposition in \cite[Problems, \S2.1] {Fribook}.
It is common to consider the matrix $A+tB$ with a complex parameter $t$, which is usually called a pencil \cite{Fribook}.

Let us first consider the case where $m=n$ and span$(T_1,T_2)$ contains an invertible matrix.  In this case $(T_1,T_2)$ is called a regular pair.  Equivalently, the pencil $T_2+tT_1$ is called a regular pencil.
So we can assume that $T_1'\in$ span$(T_1,T_2)$ is invertible.  Now choose $P=T_1'^{-1}$, $Q=I_m$ to obtain that the pair $(T_1',T_2')$ is equivalent to the pair $(I_m,A)$.   Note that $r_A(\cT)=m$.  Hence $r(\cT)\ge m$.

All other pairs of  the form $(I_m,B)$ are equivalent to $(I_m,A)$ if $B=QAQ^{-1}$ for some $Q\in\gl(m)$.
We can choose $B$ to be the Jordan canonical form of $A$, or to be the rational canonical form of $A$ \cite{Fribook}. 

We first discuss the case where $A$ is a diagonalizable matrix.  That is, we can choose $B$ to be the diagonal matrix diag$(\lambda_1,\ldots,\lambda_n)$.  Then span$(I_m,B)$ is contained in the span of $m$  linearly independent  rank-one diagonal matrices.  Hence $r(\cT)\le m$. Whence $r(\cT)=m$.  
Vice versa, assume that $r(\cT)=m$.  Then span$(\x_1\y_1\trans,\ldots,\x_m\y_m\trans)$ that contains $T_1$, $T_2$ must contain an invertible matrix.  So $\x_1,\ldots,\x_m$ and $\y_1,\ldots,\y_m$ are linearly independent.  Hence there exist unique $P,Q\trans\in\gl(m)$ such that $P\x_i=Q\trans\y_i=\e_i$ for $i\in[m]$.  Thus $PT_1Q$ and $PT_2Q$ are diagonal matrices.  In particular, $B$ is a diagonal matrix, hence $A$ is diagonalizable.

Assume now that $A$ is not diagonalizable. Hence $r(\cT)>m$.   We now discuss the case where $r(\cT)=m+1$.  

Recall the notion of the companion matrix \cite{Fribook}
which corresponds to monic polynomial $p(t)=t^m-p_1t^{m-1}-\ldots-p_m$:
\begin{equation}C(p)=\left[\begin{array}{cccccc}0&1&0&\cdots&0&0\\0&0&1&\cdots&0&0\\
\vdots&\vdots&\vdots&\vdots&\vdots&\vdots
\\0&0&0&\cdots&0&1\\p_m&p_{m-1}&p_{m-2}&\cdots&p_2&p_1
\end{array}\right].\end{equation}
Then det$(t I_m -C(p))=p(t)$.  Assume that $p(t)=\prod_{i=1}^k(t-\lambda_i)^{m_i}$, where $k\in[m]$, each $m_i$ is a positive integer, $\sum_{i=1}^k m_i =m$, and $\lambda_i\ne \lambda_j$ for $i\ne j$.
Then the Jordan canonical form of $C(p)$ has exactly one Jordan block of order $m_i$ corresponding to the eigenvalue $\lambda_i$ for $i\in[k]$.  Assume that $k<m$.
Hence $C(p)$ is not diagonalizable.  

Suppose that $B=C(p)$.  Let $\x=\e_n, \y=(-1+p_m,p_{m-1},\ldots,p_1)\trans$.
Then $C(p)-\x\y\trans=C(q)$ where $q(t)=t^m-1$.  Hence $C(q)$ is diagonalized, and 
there exists $m$ rank-one matrices $\x_1\y_1\trans,\ldots,\x_m\y_m\trans$ whose span contains $I_m,C(q)$.  Therefore the span of $\x\y\trans,\x_1\y_1\trans,\ldots,\x_m\y_m\trans$ contains $I,C(p)$.  Hence $r(\cT)\le m+1$ and therefore $r(\cT)=m+1$.

Recall next that a matrix $A\in\C^{m\times m}$ is similar to the unique matrix $B=\oplus_{i=1}^l C(p_i)$ where $p_{i+1}(t)$ divides $p_i(t)$ for $i\in [l-1]$. $B$ is called the rational canonical form of $A$ \cite{Fribook}.  ($A$ is similar to $C(p)$ if and only if $l=1$.)  The polynomials $p_1(t),\ldots,p_l(t)$ are called the invariant polynomials of $tI_m-A$ (or simply of $A$.)  Thus $(I_m,B)=\oplus_{i=1}^l (I_{n_i},C(p_i)$, where $n_i$ is the degree of $p_i$ for $i\in[l]$.

The result of J\'aJ\'a \cite{Ja79} can be summarized in the following lemma:
\begin{lemma}\label{rankcTregpair} Let $\cT\in\C^{m\times m\times 2}$.   Let $T_1,T_2$ be two frontal slices of $\cT$.  Suppose that span$(T_1,T_2)$ has dimension $2$ and contains an invertible matrix $X$.  Let $X,Y$ be a basis in span$(T_1,T_2)$, and assume that $X^{-1}Y$ has the rational canonical form $\oplus_{i=1}^lC(p_i)$.  If $p_1$ has simple roots then $r(\cT)=m$.  Suppose that $p_1,\ldots,p_k$ have multiple roots, and $p_{k+1}$ has simple roots if $k<l$. Then $r(\cT)=m+k$.
\end{lemma}
Indeed, observe first that if $p_1$ has simple roots, then all other $p_i$ have also simple roots, as each $p_i$ divides $p_1$.  Hence each $C(p_i)$ is diagonalizable, and whence $X^{-1}Y$ is diagonalizable.
Therefore $r(\cT)=m$.  Suppose now that $p_i$ does have multiple roots.  Then $C(p_i)$ is not diagonalizable.   Hence there exists $(A_1,A_2,A_3)\in\gl(m,m,2)$ such that $(A_1,A_2,A_3)(\cT)=\oplus_{i=1}^l\cT_i$, where $\cT_i\in \C^{m_i\times m_i\times 2}$, where the two frontal slices of $\cT_i$ are $I_{m_i}, C(p_i)$.  If $p_i$ have simple roots then $r(\cT_i)=m_i$.  Otherwise $r(\cT_i)=m_i+1$.  Therefore $r(\cT)\le m+k$.
It is shown in \cite{Ja79} that $r(\cT)=m+k$.

The rank of tensors $\cT\in\C^{m\times n\times 2}$ which do not satisfy conditions of Lemma \ref{rankcTregpair} is determined by the following theorem \cite{Ja79}:
\begin{theorem}\label{Jjajathm}
Assume that $\cT\in \C^{m\times n\times 2}$ where the two frontal slices $T_1,T_2$ are linearly independent.  Suppose furthermore that either $m\ne n$ or $m=n$ and the span$(T_1,T_2)$ does not contain an invertible matrix.  Then there exists $(A_1,A_2,A_3)\in\gl(m,n,2)$ such that $(A_1,A_2,A_3)(\cT)=\oplus_{i=1}^p\cT_i$, $\cT_i\in\C^{m_i\times n_i\times 2}$ for $i\in[p]$.   Either $p>1$, $m_p=n_p$ and the subspace spanned by the two frontal slices of $\cT_p$ contains an invertible matrix, or
$|m_p-n_p|=1$.  If $p>1$ then for all other $i\in [p-1]$ one has the equality $|m_i-n_i|=1$.  
\begin{enumerate}
\item Suppose that $n_j=m_j+1$.  Then the two frontal slices of $\cT_j$ are $[I_{m_j}\0]$, $[\0 I_{m_j} ]\in\C^{m_j\times n_j}$.  In this case $r(\cT_j)=n_j$.
\item
 Suppose that $m_j=n_j+1$.  Then the two frontal slices of $\cT_j$ are $[I_{n_j}\0]\trans$, $[\0 I_{n_j} ]\trans\in\C^{m_j\times n_j}$.  In this case $r(\cT_j)=m_j$.  
 \end{enumerate}
 Finally, $r(\cT)=\sum_{i=1}^p r(\cT_i)$.
 \end{theorem}
 
 To see that the $r(\cT_j)$ in the case (1) is $n_j$ we do as follows:  We extend $\cT_j\in \C^{m_j\times n_j\times 2}$ to $\hat\cT_j\in\C^{n_j\times n_j\times 2}$ by adding a row $n_j$ to the two frontal sections.  To the first section  we add the row $\e_{n_j}\trans$  to obtain $I_{n_j}$,  and to the second section we add the zero row to obtain the Jordan block $J_{n_j}$ corresponding to the eigenvalue $0$. $J_{n_j}$ is the companion matrix of $p(t)=t^{n_j}$.
Hence $r(\hat\cT_j)=n_j+1$.  We showed above that there are $n_j+1$ rank-one matrices whose linear combinations span $I_{n_j}$ and $J_{n_j}$.  As in the case
$p=1, m_p=n_p $ and $\rank \cT= n_1+1$ discussed in the beginning of this section, we can assume that  one of these rank-one matrices is of the form $\e_{n_j}\e_1\trans$.  Hence, if we delete the last row of the other $n_j$ rank-one matrices, they will span the two frontal slices of $\cT_j$.  Thus $r(\cT_j)\le n_j$.  It is straightforward to show that $r(\cT_j)>m_j$.   The case (2) can be shown similarly.
The main result of this theorem is its last part.

We now bring one application of this theorem \cite{AS79}:
\begin{eqnarray}\label{maxrankmn2}
r_{\textrm{max}}(m,n,2)=
\begin{cases}
m+\lfloor \frac{n}{2}\rfloor \textrm{   for } 2\le m\le n\le 2m,\\
2m \textrm{ for } 2\le m,\; 2m< n.
\end{cases}
\end{eqnarray}
First observe that the second case and the first case with $n=2m$ is a simple consequence of Theorem \ref{charrank3ten} when applied to horizontal sections $H_1,\ldots, H_n\in\C^{m\times 2}$ of $\cT$.  In that case $r(\cT)\le 2m$ because the whole space $\C^{m\times 2}$ is spanned by $2m$ rank-one matrices.  Assume now that $A_1,\ldots,A_{2m}$ are linearly independent.  Then these matrices span $\C^{m\times 2}$ and $r(\cT)\ge 2m$.

We now discuss the first case of \eqref{maxrankmn2}.  Let us consider first the case $m=n$.  For $m=2$ we know that the maximal rank is $3=2+\lfloor 2/2\rfloor$.  Furthermore, equality holds if and only if $\cT\in\orb(W)$.

Let us consider the case $m=3$.  Then we have two choices in Theorem  \ref{Jjajathm}.
First, $\cT$ is a direct sum of two singular pencils of dimensions $1\times 2$ and $2\times 1$.  In this case the rank of $\cT$ is $4$.  The other choice is that the two frontal sections form a regular pair.  Then Lemma \ref{rankcTregpair} yields that the maximal rank is $4$. 

Next consider the case $m=4$. If the two sections form a nonsingular pencil then $r(\cT)\le 6$.  Equality is achieved if the Jordan canonical form $X^{-1}Y$ in Lemma 
\ref{rankcTregpair} forms two nilpotent Jordan blocks of order $2$.  Other choices have smaller rank.  

We now deduce the general formula for the case $m=n$ as follows.  For $m$ even, we have $r_{\textrm{max}}(m,m,2)=3m/2$ which is achieved for a nonsingular two frontal slices, which are equaivalent to $(I,C)$, where $C$ is a sum of $m/2$ nilpotent Jordan blocks.  If $m\ge 3$ is odd, we have two possible ways to achieve the maximum rank $(3m-1)/2$.  First, a nonsingular pair $(I,C)$ where $C$ is a sum of $(m-1)/2$ nilpotent Jordan blocks and one Jordan block.  Second, a direct sum of $3\times 3$ singular pair of rank $4$, and a regular pair $(I,C)$ of order $m-3$ with the maximal rank $3(m-3)/2$.

For the case $m<n\le 2m$ the maximum possible rank is obtained as follows.  First, we consider the sum of $n-m$ copies of singular pairs of $1\times 2$.  This part contributes $2(n-m)$  to the rank of $\cT$.  If $n=2m$ we are done.  Otherwise we are left with a regular pencil of order $m-(n-m)=2m-n$ with the maximal rank 
$2m-n +\lfloor (2m-n)/2\rfloor$.  Hence the maximal rank is
\begin{equation}2(n-m)+2m-n+\lfloor (2m-2n)/2+n/2\rfloor=m+\lfloor n/2\rfloor.\end{equation}
\subsection{Validity of Strassen's direct sum conjecture for certain $3$-tensors}\label{subsec:strasscon3}
The results of J\'aJ\'a and Takche  \cite{Ja86} yield:
\begin{theorem}\label{valstrascon}  Let $\cT\in\C^{\n}, \cU\in\C^{\bp}$, where $\n=(n_1,n_2,n_3),\bp=(p_1,p_2,p_3)$.  Then $r(\cT\oplus\cU)=r(\cT) + r(\cU)$ if one of the following conditions holds:
\begin{enumerate}
\item $2\in \{n_1,n_2,n_3,p_1,p_2,p_3\}$.
\item $2\in \{n_i n_j -n_k,p_i p_j-p_k\}$ for some $i,j,k$ satisfying $\{i,j,k\}=\{1,2,3\}$.
\end{enumerate}
\end{theorem}
Use induction on $k$ to deduce the following result:
\begin{corollary}\label{varestStrcon}
Let $\cT\in\C^{\n}$, where $\n=(n_1,n_2,n_3)$.  Then
\begin{eqnarray}\label{varestStrcon1}
r(\cI(k,3)\otimes_K \cT)=r(\cI(k,3)\otimes \cT)=kr(\cT)
\end{eqnarray}
if one of the following conditions holds:
\begin{enumerate}
\item $2\in\{n_1,n_2,n_2\}$.
\item $2\in \{n_i n_j -n_k\}$ for some $i,j,k$ satisfying $\{i,j,k\}=\{1,2,3\}$.
\end{enumerate}
\end{corollary}
A recent paper \cite{BPR19} gives additional conditions where Strassen's additivity conjecture holds.   Namely, Theorem \ref{valstrascon} holds if one of the following conditions is satisfied: 
\begin{enumerate}
\item $(p_1,p_2, p_3)=(p_1, 3,3)$.
\item $r(\cU)\le 6$.
\item $\max(r_A(\cU),r_B(\cU),r_C(\cU))+2\ge r(\cU)$.
\end{enumerate}
\subsection{Generic rank of $3$-tensors}\label{subsec:genrank}
We first observe that $r_{\textrm{gen}}(m,n,p)$ is a symmetric function in the positive integer variables $m,n,p$.  Let us fix $m,n\in\N$ and assume that $2\le m\le n$.  We first observe the simple equality
\begin{eqnarray}\label{genmnpvbig}
r_{\textrm{gen}}(m,n,p)=r_{\textrm{max}}(m,n,p)=mn \textrm{ for } p\ge mn.
\end{eqnarray}
Indeed, let $T_1,\ldots,T_p\in\C^{m\times n}$ be the $p$ frontal sections of $\cT\in \C^{m\times n\times p}$.  Hence $T_1,\ldots,T_p$ are chosen at random, where each entry of each $T_k, k\in[p]$ is independent Gaussian random variable.
 As $p\ge mn$, every set of $mn$ matrices out of  $T_1,\ldots,T_p\in\C^{m\times n}$ is linearly independent.  Hence the subspace spanned by $T_1,\ldots,T_p$ is $\C^{m\times n}$.
Theorem \ref{charrank3ten}  yields that $r(\cT)=mn$.  Apply Theorem \ref{charrank3ten} to deduce that for $p\ge mn$ one has $r_{\textrm{max}}(m,n,p)=mn$.

It is left to discuss the case where $p<mn$.  We now bring the following well known result, see \cite{Fri12} and references therein:
\begin{theorem}\label{genmnpbig} Assume that $2\le m\le n$ and $(m-1)(n-1)+1\le p\le mn-1$. Then
$r_{\textrm{gen}}(m,n,p)=p.$
\end{theorem}

We now outline briefly the proof of this theorem which will need some basic notions and results in algebraic geometry that we will be using in this paper. 
A good reference on a basic algebraic geometry is  \cite{Har}.
A set $V\subset \C^N$ is called a \emph{variety} if it is zero of a finite number of polynomials in $N$ complex variables.   The algebra of polynomials in $N$ complex variables is denoted by $\C[\x], \x=(x_1,\ldots,x_N)\trans\in\C^N$.  $V$ is called \emph{irreducible} if it is not a union of two varieties, each strictly contained in $V$.  Assume that $V$ is an irreducible variety. There is a strict subvariety  of V, called Sing $V$, which consists of singular points of $V$, such that $M=V\setminus$Sing $V$ is a connected complex manifold.  The complex dimension of $M$ is called the dimension of $V$, and denoted by dim $V$.  In general, a point $\z\in\C^N$ is called \emph{generic}, or in \emph{general position} if $\z\in\C^N\setminus V$.  Usually, $V$ will depend on the property
that one studies.
  
A variety $V\ne \{\0\}$ is called a \emph{projective}  if for each $t\in\C\setminus\{\0\}$ we have that $tV=V$.  Note that a projective irreducible variety $V$ satisfies $\dim V\ge 1$.  A simplest irreducible projective variety of dimension $d$ will be a subspace $L\subset \C^N$ of dimension $d$.  A basic result in algebraic geometry says that given a projective irreducible variety $V$, $d=\dim V\in [N-1]$ then for each subspace $L$ of dimension $N-d+1$ the intersection $V\cap L$ contains at least one line, i.e., a subspace of dimension $1$.  Furthermore, there exist a subvariety $W(V)$ on the ``space'' of all vector spaces in $\C^N$ of dimension $N-d+1$, such that  $V\cap L$ has a constant number of lines for each $L\not\in W(V)$, which is denoted as deg $V$.  Moreover, for $L\not\in W(V)$, each set of $\min(N,\deg V)$ lines in $V\cap L$ is linearly independent.  Note that if $V$ is also a subspace then $\deg V=1$.

 We now consider the variety of rank-one matrices plus zero matrix in $\C^{m\times n}$.  We view  $\C^{m\times n}$ as $\C^{mn}$.  This variety is called the Segre variety Seg$(\C^{m\times n})$.  This variety has dimension $m+n-1$ and has one singular point $A=0$.  Let us take a vector space $L$ of dimension $mn-(m+n-1)+1=(m-1)(n-1)+1$.  The above results yield that each such subspace contains a rank-one matrix.   One can compute the degree of Seg$(\C^{m\times n})$ and it is not less than $(m-1)(n-1)+1$ \cite{Fri12}.  Hence a generic $(m-1)(n-1)+1$  dimensional subspace of $\C^{m\times n}$ is spanned by $(m-1)(n-1)+1$ rank-one matrices.  Consider now a generic tensor $\cT\in \C^{m\times n\times ((m-1)(n-1)+1)}$.  Let $T_1,\ldots,T_{(m-1)(n-1)+1}$ be its frontal sections.  Hence span$(T_1,\ldots,T_{(m-1)(n-1)+1})$ is a generic subspace of dimension $(m-1)(n-1)+1$, which has a basis consisting of rank-one matrices.  Theorem \ref{charrank3ten}  yields that $r(\cT)=(m-1)(n-1)+1$.
Similar arguments yield Theorem \ref{genmnpbig} for $p$ that satisfies the inequalities $(m-1)(n-1)+1<p<mn$.
 
 Letting $m=2$ in Theorem \ref{genmnpbig} we deduce that
   \begin{eqnarray}\label{2mpcase}
   r_{\textrm{gen}}(2,n,p)=p \textrm{ for } n\le p\le 2n. 
   \end{eqnarray}  
 Thus it is left to determine the generic rank in the critical range
 \begin{eqnarray}\label{critrange}
 3\le m\le n\le p\le (m-1)(n-1).
 \end{eqnarray}
 
For any subset $X$ of $\C^{n_1}\times\cdots \times \C^{n_d}$  we denote  the closure of $X$ in the standard Euclidean topology on $\C^{n_1}\times\cdots \times \C^{n_d}$
  by  $\mathrm{Closure}(X)$.  
  Let us now recall the lemma of Terracini's \cite{Ter16} --  see \cite{Fri12}. 
 \begin{lemma}\label{Frmap} Let $m,n,p$ be positive integer greater than $1$.  
 Fix $r\in\N$ and consider the polynomial map $\bF_r:(\C^m\times\C^n\times\C^p)^r\to \C^{m\times n\times p}$ given as follows:
 \begin{eqnarray*}
 \bF_r(\x_1,\y_1,\z_1,\ldots,\x_r,\y_r,\z_r)=\sum_{i=1}^r \x_i\otimes\y_i\otimes\z_i. 
 \end{eqnarray*}
 Then 
 \begin{enumerate}
 \item The set $\bF_r((\C^m\times\C^n\times\C^p)^r)$ is the set of all tensors in $\C^{m\times n\times p}$ of rank at most $r$.
 \item The set $\mathrm{Closure}(\bF_r((\C^m\times\C^n\times\C^p)^r))$ is the set of all tensors in $\C^{m\times n\times p}$ of border rank at most $r$.
 \item The set $\mathrm{Closure}(\bF_r((\C^m\times\C^n\times\C^p)^r))$ is an irreducible variety $V_r$ in $\C^{m\times n\times p}$.
 \item There exists a subvariety $W_r\subset V_r$ such that $\bF_r((\C^m\times\C^n\times\C^p)^r)\supset V_r\setminus W_r$.
 \item The dimension of $V_r$ is the maximal rank of the Jacobian of $\bF_r$.
 \item There exists a subvariety $U_r\subset (\C^m\times\C^n\times\C^p)^r$ such that the rank of the Jacobian of $\bF_r$ for each point not in $U_r$ is $\dim V_r$.
 \item The generic rank $r_{\textrm{gen}}(m,n,p)$ is the minimal $r$ such that $\dim V_r=mnp$.
 \item $\bF_r((\C^m\times\C^n\times\C^p)^r)\subseteq \bF_{r+1}((\C^m\times\C^n\times\C^p)^{r+1})$.  Equality holds if and only if $r\ge r_{\textrm{max}}(m,n,p)$.
 In particular $\bF_r((\C^m\times\C^n\times\C^p)^r)=\C^{m\times n\times p}$ if and only if $r\ge r_{\textrm{max}}(m,n,p)$.
 \end{enumerate}
 \end{lemma}
 
 We now give a lower bound for  $r_{\textrm{gen}}(m,n,p)$.  Denote by \begin{eqnarray*}\label{defSegmnp}
 \textrm{Seg}(\C^{m\times n\times p})=\{\x\otimes\y\otimes \z, \x\in\C^m,\y\in\C^n,\z\in\C^p\},
 \end{eqnarray*} 
 the variety of all tensors of rank at most $1$.   So Seg$(\C^{m\times n\times p})$, the Segre variety, is a projective variety of dimension $m+n+p-2$, with one singular point $\0$.
 Observe that the polynomial map $\bF_r$ can be viewed as an $r$-secant map $\tilde \bF_r:(\textrm{Seg}(\C^{m\times n\times p}))^r \to \C^{m\times n\times p}$.  Note that the dimension of the variety $(\textrm{Seg}(\C^{m\times n\times p}))^r$ is $r(m+n+p-2)$.  Hence Lemma \ref{Frmap} yields that $r_{\textrm{gen}}(m,n,p)(m+n+p-2)\ge mnp$.  Introducing a new quantity $r_0(m,n,p)$ we obtain a lower bound:
 \begin{eqnarray}\label{rgenlowbd}
 r_0(m,n,p):=\Big\lceil\frac{mnp}{m+n+p-2}\Big\rceil \le r_{\textrm{gen}}(m,n,p).
 \end{eqnarray}
 
 Note that for the case $p=(m-1)(n-1)+1$ Theorem \ref{genmnpbig} yields equality.  For $p >(m-1)(n-1)+1$ one can have strict inequality in \eqref{rgenlowbd}.  For example for $m=n=3$ and $p=5,6$ we have equality in the above inequality, while for $p=7$ we have a strict inequality.
 
 We now state the conjecture on the value of  $r_{\textrm{gen}}(m,n,p)$ in the critical range \cite{Fri12}:
 \begin{conjecture}\label{congenrank3}
 Assume that $m,n,p$ are integers satisfying \eqref{critrange}.  Then equality in \eqref{rgenlowbd} holds unless $(m,n,p)=(3,2k+1,2k+1)$ for $k\in\N$.  In this exceptional case it is known \cite{St83} that 
 $r_{\textrm{gen}}(3,2k+1,2k+1)=r_0(3,2k+1,2k+1)+1$.
 \end{conjecture}
 
 It was shown in  \cite{St83} that for $(m,n,p)=(3,2k+1,2k+1)$ the tensors of border rank at most $r_0(3,2k+1,2k+1)$ forms a hypersurface in $\C^{(3,2k+1,2k+1)}$.
The conjecture holds for $(3,n,n)$ and  $n\ge 3$ \cite{St83}, for $(4,n,n)$ and $n\ge 3$ \cite{AOP09}, and $(n,n,n)$ for  $n\ge 4$ \cite{Lik85,AOP09}.

 We conclude this section with a short outline of a weaker version of the inequality \eqref{maxgenrin} \cite{BT15} for $3$-tensors.  Set $r=r_{\textrm{gen}}(m,n,p)$.  It is enough to consider the case  where $r<r_{\textrm{max}}(m,n,p)$.   Observe that a finite union of subvarieties of $\C^N$ is a subvariety of some hypersurface $H(p)=\{\x\in \C^N, p(\x)=0\}$.  Identify $\C^{m\times n\times p}$ with $\C^{mnp}$.  Theorem  \ref{Frmap} implies that $V_r=\C^{m\times n\times p}$ and $V_{r-1}$ is a strict subvariety of $\C^{m\times n\times p}$. There exists a polynomial $p\in\C^{mnp}[\x]$ such that  $H(p)\supseteq V_{r-1}\cup W_r$.  Hence all tensors in $\C^{m\times n\times p}\setminus H(p)$ have rank $r$.  Let $\cT\in\C^{m\times n\times p}$ such that $r(\cT)=r_{\textrm{max}}(m,n,p)$.
 So $\cT\in H(p)$.  Recall that there exists a line through $\cT$ that intersects $H(f)$ at a finite number of points.  Choose two points $\cT_1$, $\cT_2$ on this line which do not lie in $H(p)$.  Hence $r(\cT_1)=r(\cT_2)=r$ and $\cT$ is a linear combination of $\cT_1$ and $\cT_2$.  Hence $r(\cT)\le 2r$.

The inequality  \eqref{maxgenrin} can be improved to $r(\cT)\le 2r-2$ if the closure of  tensors of rank $r-1$ is a hypersurface  \cite{BHMT17}.
 
 \subsection{A numerical way to compute     $r_{\textrm{gen}}(m,n,p)$}\label{subsec:numprocrgenmnp}
 View $f(\x,\y,\z)=\x\otimes\y\otimes\z$, as $f:\C^m\times\C^n\times\C^p\to \C^m\otimes\C^n\otimes\C^p$.  
 Then the Jacobian of $f(\x,\y,\z)$ is given by a rectangular block matrix \cite{Fri12}:
 \begin{eqnarray}\label{jacobianRGen}
 &&Df(\x,\y,\z)=[A_{\x}(\y,\z) | A_{\y}(\x,\z) | A_{\z}(\x,\y)],\\
&&A_{\x}(\y,\z)=[e_{1,1}\otimes\y\otimes\z|...|e_{m,1}\otimes\y\otimes\z]\nonumber\\
&&A_{\y}(\x,\z)=[\x\otimes e_{1,2}\otimes\z|...|\x\otimes e_{n,2}\otimes\z]\nonumber\\
&&A_{\z}(\x,\y)=[\x\otimes\y\otimes e_{1,3}|...|\x\otimes\y\otimes\e_{p,3}]\nonumber
\label{JacobianBlocks}
\end{eqnarray}
We assume that we are in the critical range \eqref{critrange}. 
 For a positive integer $r$ we define $\bF_r$ as in Lemma \ref{Frmap}. 
Then the Jacobian $D\bF_r$ is given by
\begin{eqnarray}
D\bF_r(\x_1,\y_1,\z_1,\ldots,\x_r,\y_r,\z_r)=\left[\begin{array}{c}Df(\x_1,\y_1,\z_1)\\ \vdots\\Df(\x_r,\y_r,\z_r)\end{array}\right].
\end{eqnarray}
We fix a positive integer $N$.  We start our procedure with $r=r_0(m,n,p)$ \eqref{rgenlowbd} and $j=1$.  Next we select $r$  triplets $\x_i\in\C^m,\y_i\in\C^n,\z_i\in\C^p$ at random for $i\in[r]$.
 It is enough to assume that components of each vector are drawn from the standard Gaussian distribution.
We compute the rank of $D\bF_r(\x_1,\y_1,\z_1,\ldots,\x_r,\y_r,\z_r)$, denoted as $R$.  If $R=mnp$ then $r_{\textrm{gen}}(m,n,p)=r$, and we stop our procedure. If $R<mnp$  and $j<N$ we set $j=j+1$ and repeat the above procedure.  If $j=N$ and $R<mnp$ we conclude that $r<r_{\textrm{gen}}(m,n,p)$. We set $r=r+1$ and repeat until the procedure stops.

One may assume that the generic rank $r_{\textrm{gen}}(m,n,p)$ is equal to $r_0(m,n,p)$, which is often the case. For $3\le n\le p\le 20$ in the critical range \eqref{critrange}, numerical results show that $r_0(m,n,p)$ is the generic rank, except 
the cases of $(3,2k+1,2k+1)$ with $k\in[9]$.  In these exceptional cases $r_{\textrm{gen}}(3,2k+1,2k+1)=r_0(3,2k+1,2k+1)+1$. 
That is, Conjecture \ref{congenrank3} holds for  $3\le n\le p\le 20$ in the critical range \eqref{critrange}. Such anomalies for generic rank are analogous to those reported earlier for $(3,3,3)$ and $(3,5,5)$.

It was shown in
~\cite{St83} that for all positive integers $k$ 
\begin{equation}r_{\textrm{gen}}(3,2(k+1),2(k+1))=r_0(3,2(k+1),2(k+1))\end{equation}
and  \begin{equation}r_{\textrm{gen}}(3,2k+1,2k+1)=r_0(3,2k+1,2k+1)+1.\end{equation}

 \subsection{Known results on maximal ranks of $3$-tensors}\label{subsec:maxrank3t}
Besides the exact values \eqref{maxrankmn2}, we know the following results.
The table in \cite{AS79} gives all the values of the maximal rank  $r_{\textrm{max}}(3,3,p)$ for $\n=(3,3,p)$ for $p\in [9]\setminus\{5\}$.  It is known that $r_{\textrm{max}}(3,3,5)\in\{6,7\}$.
We give their table:
\begin{equation}\begin{array}{cccccccccc}p&1&2&3&4&5&6&7&8&9\\r_{\textrm{max}}(3,3,p)&3&4&5&6&\{6,7\}&7&8&8&9\end{array}\end{equation}
Recall the table of the generic rank of $3 \times 3 \times p$ tensor.  
\begin{equation}\begin{array}{cccccccccc}p&1&2&3&4&5&6&7&8&9\\ r_{\textrm{gen}}(3,3,p)&3&3&5&5&5&6&7&8&9\end{array}\end{equation}
We now explain briefly this formula.
Recall \eqref{genmnpvbig} and Theorem \ref{genmnpbig}.  
First, $r_{\textrm{gen}}(3,3,1)$ is the maximal possible rank of $3\times 3$ matrix which is $3$.  Now $r_{\textrm{gen}}(3,3,2)=r(2,3,3)$. As $3=(2-1)(3-1)+1$ it follows from Theorem \ref{genmnpbig} that $r_{\textrm{gen}}(3,3,2)=3$.  The equality $r_{\textrm{gen}}(3,3,3)=5$ is well known and is stated in the big table of generic rank in \S\ref{subsec:grqunit}.  Again for $p\ge (3-1)(3-1)+1=5$ we get that $r_{\textrm{gen}}(3,3,p)=p$ for $5\le p\le 9$.  As $r_{\textrm{gen}}(3,3,3)\le r_{\textrm{gen}}(3,3,4)\le r_{\textrm{gen}}(3,3,5)$ it follows that $r_{\textrm{gen}}(3,3,4)=5$.

The papers \cite{AS79, AL80} give the following upper bounds on the rank of $3$-tensors 
\begin{eqnarray}
&&r_{\textrm{max}}(m, n, n)\le \Big\lfloor\frac{(m+1)n}{2}\Big\rfloor \label{Atk1}, \quad 3\le m,n,\\
&&r_{\textrm{max}}(m,n,p) \le m + \Big\lfloor\frac{p}{2}\Big\rfloor n, \quad 3\le m\le n, 3\le p \label{Atk2}\\
&&r_{\textrm{max}}(m,n,mn-u)=mn - \Big\lceil \frac{u}{2}\Big\rceil, \quad 3\le m\le n, u \le \min(4,m,n).
\label{Atk3}
\end{eqnarray}
These results and the known results that $r(m,m,m)=\big\lceil\frac{m^3}{3m-2}\big\rceil$ for $m>3$ \cite{Lik85} yield:
\begin{eqnarray}
&&r_{\textrm{max}}(4,4,4)\le 10, \quad r_{\textrm{gen}}(4,4,4)=7\label{maxgenrank4},\\
&&r_{\textrm{max}}(5,5,5)\le 15, \quad r_{\textrm{gen}}(5,5,5)=10,\label{maxgenrank5}\\
&&r_{\textrm{max}}(6,6,6)\le 21, \quad r_{\textrm{gen}}(6,6,6)=14 \label{maxgenrank6},\\
&&r_{\textrm{max}}(7,7,7)\le 28, \quad r_{\textrm{gen}}(7,7,7)=19 \label{maxgenrank7}.
\end{eqnarray}
\section{Ranks of $d$-tensors for $d\ge 4$}\label{sec:rankdge4}
This section discusses ranks of $d$-tensors for $d\ge 4$, where $\n=(n_1,\ldots,n_d), 2\le n_1\le \cdots\le n_d$.  In subsection \ref{subsec:gencased}  we give a well known characterization of the tensor rank in terms of the dimension of the minimal subspace spanned by rank-one $(d-1)$-mode tensors that contains all frontal sections of the given tensor.  
Next we bring a generalization of the Kruskal uniqueness theorem by Sidiropoulos and Bro \cite{SB00}.  Theorem \ref{genndbig} states the known result, under which conditions
 the generic rank $r_{\textrm{gen}}(\n)$ of a tensor is equal to $n_d$.  Subsection \ref{TerraciniLemma} states the lemma of Terracini.  In subsection \ref{subsec:upbdgr}  we give an upper bound on the generic rank of a tensor using purely combinatorial methods. 
 This upper bound is sharp for the case of $d$ subsystems with $n$ levels each: $n_1=\cdots=n_d=n$, for which perfect codes exist \cite{CGG02}. Subsection \ref{subsec:grqunit} concentrates on the generic rank of $d$ qunits. In particular,
  we provide a table of generic ranks of tensors with $d$ indices running from $1$ to $n$,
  for which  these values are known.  Subsection
\ref{subsec:estrankpoleq} discusses an algorithmic way to find the rank of a tensor using solvability of a system of linear equations with several variables.  
In subsection \ref{subsec:genidgt} we discuss the problem of
generic identifiability of tensors.  Namely, assuming an integer $r$ is less than a generic rank, 
when a generic tensor of rank $r$ has a unique decomposition as a sum of $r$ rank-one tensors?  The results of Chiantini-Ottaviani-Vannieuwenhoven \cite{COV14} show that if  
$\prod_{i=1}^d (n_i-1)\le 15 000$ then  
the generic identifiability property holds except in a number of known cases.
\subsection{General case}\label{subsec:gencased}  We now bring the analog of Theorem \ref{charrank3ten} for a general $d\ge 3$.
\begin{theorem}\label{charrankdten}  Let $\n=(n_1,\ldots,n_d)$, where $d\ge 3$ and $n_i\ge 2$ for $i\in[d]$.  Assume that $\cT=[\cT_{i_1,\ldots,i_d}]\in \C^{\mathbf{n}}$, and let $\cT_1=[\cT_{i_1,\ldots,i_{d-1},1}],\ldots,\cT_{n_d}=[\cT_{i_1,\ldots,i_{d-1},n_d}]\in\C^{(n_1,\ldots,n_{d-1})}$ be the $n_d$ frontal slices of $\cT$.  Then $r(\cT)$ is the minimum dimension of a subspace of $\C^{(n_1,\ldots,n_{d-1})}$ spanned by rank-one tensors that contains the subspace spanned by $T_1,\ldots,T_{n_d}$.
\end{theorem}
In particular, we deduce that $r(\cT)\le N(\n)/n_d$.  Apply the above theorem to a mode $k\in[d]$ to deduce that $r(\cT)\le N(\n)/n_k$.  This proves \eqref{upbundrank}.

The following generalization of Kruskal's theorem for $d$-tensors is  
 due to Sidiropoulos and Bro \cite{SB00}:
\begin{lemma}\label{Kruskaldge3}  Assume that $3\le d$, $2\le n_1\le \cdots \le n_d$
are integers.   
Assume that
\begin{eqnarray}\label{Kruskaldg3r}
\cT=\sum_{i=1}^r \otimes_{j=1}^d \x_{i,j},
\quad \x_{i,j}\in\C^{n_j}\setminus\{\0\}, i\in[r], j\in[d].
\end{eqnarray}
If 
\begin{equation}\label{Krusdinq}
\sum_{j=1}^d r_K(\bx_{1,j},\ldots,\bx_{n_j,j})\ge 2r+(d-1)
\end{equation}
then $r=r(\cT)$ and the decomposition \eqref{Kruskaldg3r} is unique.
\end{lemma}
To establish  the uniqueness of a rank decomposition using the above lemma, it seems that in many cases it is beneficial to view
a $d$-tensor as a $3$-tensor in $\cH_{n_1}\otimes \cH_{n_2}\otimes \cH_{n_3,\ldots,n_d}$ \cite{Fri16}:
\begin{lemma}\label{Kruskaldg3}
Assume that $3\le d$, $2\le n_1\le \cdots \le n_d$
are integers.  Decompose the multiset $\{n_1,\ldots,n_d\}$ to a union of three nonempty disjoint multisets $S_1\cup S_2\cup S_3$, which induce three vectors $\n_1\in\N^{|S_1|}, \n_2\in \N^{|S_2|}, \n_3\in\N^{|S_3|}$.  ($|S_k|$ is the cardinality of $S_k$.)  Then $\C^{\n_1}\otimes \C^{\n_2}\otimes \C^{\n_3}$ is obtained from   $\otimes_{j=1}^d \C^{n_j}$ by permuting factors $\C^{n_1},\ldots,\C^{n_d}$.  Thus each $\cT\in\C^{\n}$ induces $\hat \cT\in\otimes_{k=1}^3 \in\C^{\n_k}$.  Assume that $\cT$ has a decomposition \eqref{Kruskaldg3r}.
Let $\hat\cT=\sum_{i=1}^r \otimes_{k=1}^3 \cT_{i,k}$, 
where each $\cT_{i,k}\in\C^{\n_k}$ is a rank-one tensor induced decomposition of $\cT$ .  View each $\cT_{i,k}$ as a vector in $\C^{N(\n_k)}$.  If
\begin{eqnarray}\label{rTKineq}
r_{\textrm{K}}(\cT_{1,1},\ldots,\cT_{r,1}) + r_{\textrm{K}}(\cT_{1,2},\ldots,\cT_{r,2}) +r_{\textrm{K}}(\cT_{1,3},\ldots,\cT_{r,3})\ge 2r+2
\end{eqnarray}
then $r(\cT)=r(\hat\cT)=r$ and the above decomposition of $\hat \cT$ and the corresponding decomposition of $\cT$  is unique up to a permutation
 of rank-one tensors in the decomposition. 
\end{lemma}

\begin{proposition}\label{Kruskaldg3d} Let the assumptions of Lemma \ref{Kruskaldg3} hold.  Moreover, assume that $N(\n_1)\le N(\n_2)\le N(\n_3)$.  Suppose that $\cT$ has a decomposition \eqref{Kruskaldg3r},
where all $\x_{i,j}$ are in general position.  (The entries of each $\x_{i,j}$ are chosen from independent $N(0,1)$ Gaussian distribution.)   Then  $r=r(\cT)$ and  the decomposition \eqref{Kruskaldg3r} of $\cT$ is unique up to a permutation of summands for the following values of $r$:
\begin{enumerate} 
\item If $r\le N(\n_1)+N(\n_2)-2$ and $N(\n_3)\ge N(\n_1)+N(\n_2)-2$.
\item If $r\le N(\n_1)+N(\n_2)-3$, $N(\n_1)\ge 3$ and $N(\n_3)=N(\n_1)+N(\n_2)-3$.
\item If $r\le \frac{1}{2}(N(\n_1)+N(\n_2)+ N(\n_3)-2)$, $N(\n_1)\ge 4$ and $N(\n_1)+N(\n_2)-4\ge N(\n_3)$.
\end{enumerate}
 \end{proposition}
 \begin{proof}  
 The results in \cite{Fri16} yield that $ r_{\textrm{K}}(\cT_{1,k},\ldots,\cT_{r,k})=\min(r,N(\n_k))$ for $k\in [3]$.  Suppose first that $2\le r\le N(\n_1)$.  Then the left hand side of \eqref{rTKineq} is $3r$.  As $r\ge 2$ the inequality \eqref{rTKineq} holds.  Hence $r(\cT)=r$. 
 
Assume now that $N(\n_1)<r\le N(\n_2)$.  Then $ r_{\textrm{K}}(\cT_{1,1},\ldots,\cT_{r,1})=N(\n_1)$  and $ r_{\textrm{K}}(\cT_{1,k},\ldots,\cT_{r,k})=r$ for $k\in \{2,3\}$. Then the inequality \eqref{rTKineq} holds.  Hence $r(\cT)=r$.  Assume that $N(\n_2)\le r\le N(\n_3)$.  Then $ r_{\textrm{K}}(\cT_{1,k},\ldots,\cT_{r,k})=N(\n_k)$ for $k\in [2]$ and $ r_{\textrm{K}}(\cT_{1,3},\ldots,\cT_{r,3})=r$.  Then the inequality \eqref{rTKineq} is equivalent to $r\le N(\n_1)+N(\n_2)-2$.  Therefore \emph{(1)} holds.
 
Suppose that $N(\n_3)=N(\n_1)+N(\n_2)-3$.  As $N(\n_3)\ge N(\n_2)$ we deduce that $N(\n_1)\ge 3$.  Suppose that $N(\n_2)\le r\le N(\n_3)=N(\n_1)+N(\n_2)-3$.  Then the above arguments show that \eqref{rTKineq} holds.  For $r=N(\n_1)+N(\n_2)-2$ we obtain that $ r_{\textrm{K}}(\cT_{1,k},\ldots,\cT_{r,k})=N(\n_k)$ for $k\in[3]$.  For this value of $r$ the inequality \eqref{rTKineq} does not hold.  Thus \emph{(2)} is the best one can obtain for the case $N(\n_3)=N(\n_1)+N(\n_2)-3$.

Assume that $N(\n_1)+N(\n_2)-4\ge N(\n_3)$.  As $N(\n_3)\ge N(\n_2)$ we deduce that $N(\n_1)\ge 4$.  Suppose that $r\le N(\n_3)$.
Then the above arguments show that \eqref{rTKineq} holds. Assume that $r>N(\n_3)$.
Then $ r_{\textrm{K}}(\cT_{1,k},\ldots,\cT_{r,k})=N(\n_k)$ for $k\in [3]$.    Hence \eqref{rTKineq} is equivalent to  $r\le \frac{1}{2}(N(\n_1)+N(\n_2)+ N(\n_3)-2)$.  This establishes \emph{(3)}.
 \end{proof}

It seems that the best way to group the multiset $\{n_1,\ldots,n_d\}$ to $S_1\cup S_2\cup S_3$ is in such a way that $\n_1=(n_1)$, $N(\n_2)\le N(\n_3)$ and $N(\n_3)-N(\n_2)$ is smallest possible.  (Note that $N(\n_1)N(\n_2)N(\n_3)=N(\n)$.)
The following Corollary reveals the advantage of our decomposition of a tensor $\cT$ as a three tensor:
\begin{corollary}\label{Kruskaldg3c}
Assume that $d=2p+1, p\ge 2$, $ n_1=\cdots= n_d=n\ge 2$ are integers.  Suppose that $\cT$ has a decomposition  
$\sum_{i=1}^r\otimes_{j=1}^d \x_{i,j}$,
where all $\x_{i,j}$ are in general position.  (The entries of each $\x_{i,j}$ are chosen from independent $N(0,1)$ Gaussian distribution.)   If $r$ satisfies the following inequalities then $r=r(\cT)$:
\begin{eqnarray*}
r\le 
\begin{cases}
n^p \textrm{ if } n=2,3,\\
n^p-1 +\frac{n}{2} \textrm{ if } n\ge 4,
\end{cases}
\end{eqnarray*}
and the above decomposition of $\cT$ is unique.
 \end{corollary}
 
 We now bring the known analog of Theorem \ref{genmnpbig}:
 \begin{theorem}\label{genndbig} Assume that $3\le d$, $2\le n_1\le \cdots \le n_d$ are integers.  Then
 \begin{eqnarray}
 r_{\textrm{gen}}(\n)=n_d \;\;\text{\sl for }\;\;\prod_{j=1}^{d-1} n_j+d-1-\sum_{j=1}^{d-1}n_j\le n_d\le N(\n).
 \end{eqnarray}
 \end{theorem}
 The proof is similar to the proof of Theorem \ref{genmnpbig}.
Denote  by $M(\n):=1-d+\sum_{j=1}^d n_j$ the dimension of the Segre variety in $\C^{\n}$.
 Let $\n'=(n_1,\ldots,n_{d-1})$.  Hence a generic subspace in $\C^{\n'}$ of dimension $N(\n')-M(\n')+1$ intersects the Segre variety in $\C^{\n'}$ in a finite number of points, whose linear span is this subspace.  Use Theorem \ref{charrankdten}  to deduce Theorem \ref{genndbig} for $n_d=N(\n')-M(\n')+1$.  Similar arguments yield the theorem for $n_d > N(\n')-M(\n')+1$.
 
Introducing the generalized version of the lower bound given in \eqref{rgenlowbd} we obtain the following lower bound for the generic rank
 \begin{eqnarray}\label{genrankdbd}
r_0(\n):= \Bigg\lceil\frac{N(\n)}{M(\n)}\Bigg\rceil\le  r_{\textrm{gen}}(\n).
 \end{eqnarray}
 
Assume that $2\le n_1\le \cdots\le n_d$.  Is the above inequality optimal for $n_d\le N(\n')-M(\n')$?  
 
 In \S\ref{subsec:grqunit} we show some affirmative results for the case $n=n_1=\cdots=n_d$, which we call $d$-qunit states, or simply $d$-qunits.   We now discuss in detail Terracini's lemma in the general setting.
\subsection{Terracini's lemma}\label{TerraciniLemma} 
We recall the results in \cite{Fri12}.   For a fixed $r\in\N$ consider the map 
\begin{equation}\label{defFrd}
\begin{aligned}
\bF_r: (\C^{n_1}\times\cdots\times\C^{n_d})^r\to \C^{\n},\\
\bF_r(\x_{1,1},\ldots,\x_{d,1},\ldots,\x_{1,r},\ldots,\x_{d,r})=\sum_{i=1}^r \otimes_{j=1}^d \x_{j,i}.
\end{aligned}
\end{equation}
The set $\bF_r((\C^{n_1}\times\cdots\times\C^{n_d})^r)$ is a constructible set, of dimension $d(\n,r)$, whose closure is an irreducible variety in $\C^{\n}$.  (A constructible set of dimension $k$ in $\C^m$ is a finite union of irreducible varieties whose maximal dimension is $k$ minus a union of a finite number of constructible sets of dimension at most $k-1$ \cite{Har}.)   The dimension $d(\n,r)$ is the rank of the Jacobian matrix of $\bF_r$ at a generic point 
\[(\x_{1,1},\ldots,\x_{d,1},\ldots,\x_{1,r},\ldots,\x_{d,r})\in (\C^{n_1}\times\cdots\times\C^{n_d})^r.\]
The following results are known \cite{Fri12}:
\begin{enumerate}
\item $d(\n,r_{\textrm{gen}}(\n))=N(\n)$.
\item The sequence $d(\n,r)$ is strictly increasing for $r\in[r_{\textrm{gen}}(\n)]$.
\item $d(\n,r)=N(\n)$ for each integer $r>r_{\textrm{gen}}(\n)$.
\end{enumerate}

The rank of the Jacobian $D\bF_r$ at the point $(\x_{1,1},\ldots,\x_{d,1},\ldots,\x_{1,r},\ldots,\x_{d,r})$ is the dimension
of the subspaces spanned by the following  vectors
\begin{eqnarray*}\label{FriTer}
\left(\otimes_{j=1}^{k-1} \x_{j,i}\right)\otimes \e_{l_k,k,i}\left(\otimes_{j=k+1}^d\x_{j,i}\right), \quad l_k\in[n_k], i\in [r].
\end{eqnarray*}
Here $\e_{1,k,i},\ldots,\e_{n_k,k,i}$ is a basis in $\C^{n_k}$ for $k\in [d]$ and $i\in[r]$,
since for each rank-one component $\otimes_{j=1}^d\x_{j,i}$ one can have a
 different basis in each component $\C^{n_j}$. 
\subsection{An upper bound on $r_{\textrm{gen}}(\n)$}\label{subsec:upbdgr}
We now give an upper bound on the generic rank using pure combinatorial methods.
Consider the standard basis in $\C^{\n}$: 
\begin{eqnarray*}\label{standbas}
\otimes_{j=1}^d \e_{l_j,j}, \quad \e_{l_j,j}=(\delta_{l_j1},\ldots, \delta_{l_jn_j})\trans, \ l_j\in[n_j], \ j\in[d].
\end{eqnarray*}
That is, each element in the basis corresponds to a $d$-tuple $(l_1,\ldots,l_d)$, where $l_j\in [n_j]$ for $j\in [d]$.
Denote by $[\n]$ the set of such of such $d$-tuples:
\begin{eqnarray*}\label{defJn}
[\n]:=[n_1]\times\cdots\times[n_d]=\big\{\mathbf{l}=(l_1,\ldots,l_d), \; l_j\in[n_j], j\in[d]\big\}.
\end{eqnarray*}
Recall the Hamming distance on $[\n]$ is given by the formula:
\begin{equation}\textrm{dist}\big((l_1,\ldots,l_d), (m_1,\ldots, m_d)\big)=p,\end{equation}
if $m_j\ne l_j$ for exactly $p$ indices.
Denote by $O(\mathbf{l})$ the set of all points in $[\n]$ whose distance from $\mathbf{l}$ is at most $1$.
Note that the cardinality of $O(\mathbf{l})$, denoted as $|O(\mathbf{l})|$, is  $M(\n)$.

A subset $A\subseteq [\n]$ is called a dominating set of $[\n]$ if $\cup_{\mathbf{l}\in A} O(\mathbf{l})=[\n]$.
The cardinality of each dominating set $A$ satisfies the inequality $|A|M(\n)\ge N(\n)$.
Denote by $\cA(\n)$ the set of dominating sets.
Let $\gamma(\n):=\min\{|A|, A\in \cA(\n)\}$ be the minimum cardinality of the dominating set.

A subset $B$ of $[\n]$ is called $3$-separated set if the Hamming distance between any two elements of $B$  is at least $3$.  Note that if $B$ is $3$-separated then $|B|M(\n)\le N(\n)$.
Denote by $\cB(\n)$ the set of $3$-separated sets of $[\n]$. Let $\kappa(\n):=\max\{|B|, B\in\cB(\n)\}$ be the maximum cardinality of $3$-separable set.
Observe that $\gamma(\n)\ge \kappa(\n)$.  A maximum $3$-separated set $B$ is called a $1$-perfect code if $\gamma(\n)=\kappa(\n)$.   That is,  the Hamming distance between every two elements of $B$ is at least $3$,  and for each $\bp\in[\n]$ there exists $\bq\in B$ such that dist$(\bp,\bq)\le 1$.
The following result is due to \cite{CGG02}:
\begin{lemma}\label{combinthm}  Assume that $3\le d$ and $2\le n_1\le\cdots\le n_d$ be integers.  Then the following assertions hold:
\begin{enumerate} 
\item The inequality $r_{\mathrm{gen}}(\n)\le \gamma(\n)$ holds.
\item For each $r\in[k(\n)]$ the closure of $\bF_r\left(\C^{n_1}\times \cdots\times \C^{n_d}\right)$ is an irreducible variety of dimension at most $rM(\n)$.
In particular, if the dimension of $\bF_r\left(\C^{n_1}\times \cdots\times \C^{n_d}\right)$ is $rM(\n)$ then most of tensors of rank $r$ have exactly $\deg F_r$ different rank decomposition.
\end{enumerate}
\end{lemma}

The above inequalities for dominating and $3$-separated sets yield that $|B|=\frac{N(\n)}{M(\n)}$.  In particular $\frac{N(\n)}{M(\n)}$ is an integer.  Furthermore, the inequality \eqref{genrankdbd} and Lemma \ref{combinthm} yield that $r_0(\n)=r_{\textrm{gen}}(\n)$.

It is known \cite{Thom} that $1$-perfect code exists if 
\begin{eqnarray*} 
n_1=\cdots =n_d=n=q^l, \;d=\frac{n^{a+1}-1}{n-1}, \quad q \textrm{ is prime},\quad l,a\in\N, a\ge 2.
\end{eqnarray*}
Use Lemma \ref{combinthm} to deduce that in this case $r_{\textrm{gen}}(\n)=n^{d-a-1}$.

Let $G=(V,E)$ be a simple graph on the set of vertices $V$ and edges $E$.
Recall that $A\subseteq V$ is a dominating set if each vertex $v$ not in $A$ is adjacent to some vertex in $A$.  Then $\gamma(G)$ is called the domination number of $G$, if $\gamma(G)$ is the minimum cardinality of a dominating set in $G$.  
We will show below that $\gamma(\n)=\gamma(G(\n))$, where $G(\n)=([\n],E(\n))$ is the induced graph on $[\n]$ by the Hamming distance.

The domination number of $G$ is a solution to the following minimum problem in $|V|$ variables $x_v,v\in V$ whose values are in $\{0,1\}$.  For each $\x=(x_v)_{v\in V}\in\{0,1\}^V$ we denote by supp $\x$ the subset $\{v\in V, x_v=1\}$.  Then supp $\x$ is a dominating set in $V$ if and only if the following inequalities hold
\begin{eqnarray}\label{domsetin}
x_v+\sum_{u, (u,v)\in E} x_u\ge 1 \textrm{ for all } v\in V.
\end{eqnarray}
Hence $\gamma(G)$ is the minimum of $\sum_{v\in V} x_v$ on $\x\in\{0,1\}^V$ subject to \eqref{domsetin}.  It is known that  computing $\gamma(G)$ for general graphs is an NP-complete problem \cite{KV12}.

A greedy algorithm to find an upper bound for $\gamma(G)$ is as follows:  Let $G_1=G$.  Suppose that at the stage $k\in[V]$ we have the graph $G_k=(V_k,E_k)$, where $V_k$ is a nonempty subset of $V$, and  $G_k$ is the induced subgraph of $G$ by the set $V_k$.  We choose a vertex $v_k\in V_k$ of a maximum degree in $G_k$.  Let  $O_k\subset V_k$ be the neighbors of $v_k$ in $G_k$.  Then $V_{k+1}=V_k\setminus\{\{v_k\}\cup O_k\}$.
If $V_{k+1}=\emptyset$ then $A=\{v_1,\ldots,v_k\}$ is the dominating set.  Otherwise set \red{$k$ as $k+1$}.  

Recall the standard linear programming (LP) relaxation of the above minimal problem on $\{0,1\}^V$ \cite{CCPS}.  Namely, we replace the condition $x_v\in \{0,1\}, v\in V$ by the condition $0\le x_v\le 1, v\in V$.  Thus we consider the minimum $\sum_{v\in V}x_v$ satisfying the inequalities \eqref{domsetin} for $\x\in [0,1]^V$.  Denote this minimum by $\beta(G)$.  

The following result is well known \cite{Pet96}:  Let $G=(V,E)$ be a simple graph with maximal degree $\Delta(G)$.  Denote by $A(G)\subseteq V$, a dominating  set obtained by the above greedy algorithm.  Then
\begin{eqnarray}
\beta(G)\le \gamma(G)\le |A(G)|\le O(\log \Delta(G))\beta(G).
\end{eqnarray}
 
Recall that $G$ is called regular, if the degree of each vertex is $\Delta(G)$.
 One can show that for a regular graph $G$ one arrives at the inequality: 
 \begin{equation}\beta(G)\le \frac{\#V}{\Delta(G)+1}.\end{equation}
 Indeed, define $x_v=\frac{1}{\Delta(G)+1}$ for each $v\in V$.  Then the conditions \eqref{domsetin} are satisfied.  
As the following equality is true, 
  $\sum_{v\in V}x_v= \frac{\#V}{\Delta(G)+1}$,
  the above inequality holds.
Thus we showed that for regular graph $G$ we have the following inequalities:
 \begin{eqnarray}\label{domsetreg}
  \frac{\#V}{\Delta(G)+1}\le \gamma(G)\le |A(G)|\le O(\log \Delta(G))\frac{\#V}{\Delta(G)+1}.
 \end{eqnarray}
 (Recall the notation $O(m)$ for some function $f:\N\to [0,\infty)$.  Namely, there exists a universal $K>0$ so that $f(m)\le Km$ for all $m\in\N$.)
 
We now apply these results to estimate from above the generic rank.  Let $G(\n)=([\n], E(\n))$.  Two vertices $\mathbf{l},\mathbf{m}\in[\n]$ are adjacent if dist$(\mathbf{l},\mathbf{m})=1$.
Observe that $G(\n)$ is a regular graph with $\Delta(G(\n))=M(\n)-1$.
 It is easy to show that $A\subseteq [\n]$ is a dominating set if and only if $A$ is a dominating set in $G(\n)$.  That is, $\gamma(\n)=\gamma(G(\n))$.  Thus \eqref{domsetreg} applies to  $G(\n)$.  We do not know how good is the upper bound on $\gamma(n)$ given in \eqref{domsetreg} in the general case.  Apply Lemma \eqref{combinthm}  to deduce sandwich bound:
 \begin{eqnarray}\label{lowupgrankbd}
 r_0(\n)\le r_{\textrm{gen}}(\n)\le O\Big(\log \Big(\sum_{i=1}^d (n_i-1)\Big)\Big)r_0(\n).
 \end{eqnarray}

 Let us consider the following simple examples for $d=3$ and $n_1=n_2=n_3=3$.
 Choose
 \[A=\{(1,1,1),(2,2,2),(3,3,3),(1,1,2),(2,2,3),(3,3,1)\}\]
 end 
 \[B=\{(1,1,1),(2,2,2),(3,3,3).\]
 Then $A$ is a dominating set and $B$ is $3$-separated set.
  Recall that $r_{\textrm{gen}}(3,3,3)=5< |A|=6$.
 It is straightforward to show that $B$ is a maximal $3$-separated set.  So $\kappa(3,3,3)=3$.

 \subsection{The generic rank of $d$-qunits}\label{subsec:grqunit}
 Let $n^{\times d}=(n,\ldots,n)\in\N^d$.  Then $r_{\textrm{gen}}(n^{\times d})$ is the generic rank of $d$-qunits.  Inequality \eqref{genrankdbd} yields
\begin{eqnarray}\label{genqunit}
r_{\textrm{gen}}(n^{\times d})\ge \lceil \theta(n^{\times d})\rceil, \quad \text{where} \quad \theta(n^{\times d})=\frac{n^d}{d(n-1)+1}.
\end{eqnarray}
In previous subsection we showed that equality holds if $[n]^{d}$ has $1$-perfect code \cite{CGG02}.

It was shown in \cite{CGG11} that equality holds in \eqref{genqunit} for $n=2$ and any $d\ge 2$.  That is, the generic rank of $d$-qubits is $\lceil 2^d/(d+1)\rceil$.

We now recall some results in \cite{AOP09} for $r_{\textrm{gen}}(n^{\times d})$.  First, assume that $\theta(n^{\times d})$ is integer.  (Thus $d=\frac{n^{a+1}-1}{n-1}$ for $a\in\N$.)
Then $r_{\textrm{gen}}(n^{\times d})=\theta(n^{\times d})$.  Second, assume that $\theta(n^{\times d})$ is not an integer.  Let
$\lfloor \theta(n^{\times d})\rfloor\equiv_{mod \;n}\delta(n^{\times d})\in \{0,\ldots,n-1\}$.  Then
\begin{eqnarray}
 r_{\textrm{gen}}(n^{\times d})=
 \begin{cases}
 \lceil \frac{n^d}{d(n-1)+1}\rceil$ if $\delta(n^{\times d})=n-1,\\
 r_{\textrm{gen}}(n^{\times d})\le \lceil \frac{n^d}{d(n-1)+1}\rceil +n-1-\delta(n^{\times d}).
 \end{cases}
 \end{eqnarray}
 We now provide a few examples of the above equalities and inequalities.  
 According to  \cite{CGG11} for $n=2$ and the known table of the values of $r_{\textrm{gen}}(n^{\times d})$, which is given later, in all below cases the upper bound on  $ r_{\textrm{gen}}(n^{\times d})$ is a strict inequality.
\begin{eqnarray}
&&\theta(2^{\times 4})=16/5,\; \lfloor \gamma(2^{\times 4})\rfloor=3, \;\delta(2^{\times 4})=1, \;r_{\textrm{gen}}(2^{\times 4})=4,\\
&&\theta(2^{\times 5})=32/6, \;\lfloor \gamma(2^{\times 5})\rfloor=5, \;\delta(2^{\times 5})=1,\;r_{\textrm{gen}}(2^{\times 5})=6,\\
&&\theta(2^{\times 8})=256/9, \;\lfloor \theta(2^{\times 5})\rfloor=28,\; \delta(2^{\times 5})=0, \;r_{\textrm{gen}}(2^{\times 8})=29<30,\\
&&\theta(3^{\times 3})=27/7,\; \lfloor \theta(3^{\times 3})\rfloor=3, \;\delta(3^{\times 3})=0,\; r_{\textrm{gen}}(3^{\times 3})=5<6, \\
&&\theta(3^{\times 5})=243/11, \;\lfloor \theta(3^{\times 5})\rfloor=22,\; \delta(3^{\times 5})=1,\; r_{\textrm{gen}}(3^{\times 5})=23<24,\\
&&\theta(3^{\times 6})=729/13,\; \lfloor \theta(3^{\times 6})\rfloor=56,\; \delta(3^{\times 3})=2,\; r_{\textrm{gen}}(3^{\times 6})=57.
 \end{eqnarray}
 
 Known values of generic rank, $r_{\textrm{gen}}(n^{\times d})$, 
  for the system of $d$ qunits are listed in Table 1.
\begin{table}[!ht]
\label{tab1}
{\scriptsize
\begin{equation}
\begin{BMAT}{r|rrrrrrrrr}{c|ccccccccccccccc}
d\setminus
 n &         2 &         3 &         4 &              5 &              6 &         7 &                 8 &        9 &                 10\\
 2 &    ^{[a,b]}2 &         3 &         4 &              5 &              6 &         7 &                 8 &        9 &                 10\\
 3 &   ^{[b,c,d]}2 &    ^{[c,d]}5 &    ^{[c,d]}7 &       ^{[a,c,d]}10 &        ^{[c,d]}14 &   ^{[c,d]}19 &          ^{[a,c,d]}24 &   ^{[c]}30 &   ^{[c]}36\\
 4 &   ^{[a,b,d]}4 &     ^{[d]}9 &   ^{[a,d]}20 &         ^{[d]}37 &         ^{[d]}62 &    ^{[d]}97 &         {\bf 142} & {\bf 199}&          {\bf 271}\\
 5 &    ^{[a,b]}6 &        {\bf 23} &  ^{[a,b]}64 &      {\bf 149} &        ^{[a]}300 & {\bf 543} &                 . &        . &                  .\\
 6 &   ^{[a,b]}10 &    ^{[a]}57 &   ^{[a]}216 &      ^{[a,b]} 625 &       ^{[a]}1506 &         . &                 . &        . &                  .\\
 7 &    ^{[a,b]}16 &       {\bf 146} & {\bf 745} &              . &              ^{[a]} 6^5&         . &         ^{[a]}41944 &        . & ^{[a]} 156250\\
 8 &    ^{[b]}29 &       {\bf 386} &         . &              . &              . &        ^{[a,b]} 7^{6}&                 . &        . &                  .\\
 9 &   ^{[a,b]}52 &      {\bf 1036} &         . &              . &              . &         . &                ^{[a,b]} 8^{7} &        . &                  .\\
10 &   ^{[a,b]}94 &         . &         . &              . &    ^{[1]}1185612 &         . &                 . &       ^{[a,b]} 9^8&      ^{[a]}109890110\\
11 &   ^{[b]}171 &         . &         . &    ^{[a]}1085070 &              . &         . &                 . &        . &                  .\\
12 &  ^{[a,b]}316 & ^{[a]}21258 &         . &              . &              . &         . &                 . &        . &                 ^{[a]} 10^9\\
13 &  ^{[a,b]}586 &         ^{[a,b]} 3^{10}&         . &   ^{[a]}23032135 &              . &         . &                 . &        . &                  .\\
14 &  ^{[b]}1093 &         . &         . &              . & ^{[a]}1103720622 &         . &                 . &        . &                  .\\
15 &  ^{[b]}2048 &         . &         . &              . &              . &         . &                 . &        . &                  .\\
16 & ^{[a,b]}3856 &         . &         . & ^{[a]}2347506010 &     ^{[a]} 2^{16}3^{12}&         . & ^{[a]}2490928997440 &        . & ^{[a]}68965517241380
\addpath{(10,15,1)dddllldlldldlddlddddddl}
\end{BMAT}
\label{table_of_granks}
\end{equation}
}
\caption{Values of the generic rank calculated for several $n$-level systems.
Polygonal chain defines two areas in the array. Numbers in the upper-left part
corresponding to the right-hand side of the formula \eqref{genqunit} are confirmed
numerically.
Numbers decorated with $^{[a,b,c,d]}$ on the left represent known results according to the reference in \cite{AOP09, CGG11, BCS97, CBLC09}, respectively,
while numbers in bold are results obtained in this work by numerical calculations.}
\label{T1}
\end{table}

Table 2 provides a comparison between the generic rank $r_{\textrm{gen}}$, the maximal ranks $r_{\textrm{max}}$
and the maximal number $R_U$ of terms 
in the shortest representation of a pure state of $d$ 
subsystems with $n$ levels each
in an orthogonal product basis in $\mathcal {H}_n^{\otimes d}$.
The upper bound $R_U = n^d - dn(n-1)/2$   
follows directly from the work of Carteret, Higuchi and Sudbery \cite{CHS00}.
They demonstrated that out of $n^d$ entries of any tensor $\cT$
one can set to zero $n(n-1)/2$ entries
by performing a single unitary rotation which affects a single index.
As there are $d$ independent indices, for which such a transformation can be applied,
 the total number of entries which can be set to zero is 
 $dn(n-1)/2$. This  explains the bound $R_U$ stated above.
\begin{table}[h]
\caption{Generic ranks $r_{\textrm{gen}}$, maximal ranks $r_{\textrm{max}}$
and the maximal number $R_U$ of terms 
in the shortest representation of a pure state of $d$ 
subsystems with $n$ levels each
in an orthogonal product basis in $\mathcal {H}_n^{\otimes d}$.
Numbers in {\bf bold} denote exact results,
other numbers denote upper bounds obtained in \cite{CBLC09} and in \cite{Fri12,BT15}.}
 \smallskip
\hskip -0.5cm
{\renewcommand{\arraystretch}{1.21}
\begin{tabular}[h]
{c|ccc|ccc|ccc|ccc} %
\hline \hline
 \ \ \ \ \   & ~ & $n=2$ & ~ & ~ & $3$ & ~ & ~  & $4$ & ~ 
& ~ & $5$ &  
\\
\hline
    $d$   &$r_{\textrm{gen}}$  & $r_{\textrm{max}}$ & $R_U$ &$r_{\textrm{gen}}$  & $r_{\textrm{max}}$ & $R_U$ 
            &$r_{\textrm{gen}}$  & $r_{\textrm{max}}$ & $R_U$ &$r_{\textrm{gen}}$  & $r_{\textrm{max}}$ & $R_U$ \\
  $2$      & ${\bf 2}$  & ${\bf 2}$ & ${\bf 2}$ 
           & ${\bf 3}$  & ${\bf 3}$ & ${\bf 3}$ 
           & ${\bf 4}$  & ${\bf 4}$ & ${\bf 4}$ 
           & ${\bf 5}$  & ${\bf 5}$ & ${\bf 5}$  \\
  $3$      & ${\bf 2}$  & ${\bf 3}$ & ${\bf 5}$ 
           & ${\bf 5}$  & ${\bf 5}$ & ${\bf 18}$ 
           & ${\bf 7}$  & ${ 13}$ & ${\bf 46}$ 
           & ${\bf 10}$  & ${20}$ & ${\bf 95}$  \\
 $4$       & ${\bf 4}$  & ${\bf 4}$ & ${\bf 12}$ \  
           & ${\bf 9}$  & $18$ & ${\bf 69}$ 
           & ${\bf 20}$  & $40$ & ${\bf 232}$  \ \
           & ${\bf 37}$  & $74$ & \! \!\! ${\bf 585}$  \  \  \\
\hline \hline
\end{tabular}
}
\label{tab:rank}
\end{table}
\subsection{Estimating and computing the rank of a tensor using polynomial equations}\label{subsec:estrankpoleq}
It is clear to many researchers in the field that the rank of a tensor over a given given field $\F$ is equivalent to solvability of corresponding system of polynomial equations over $\F$.  See for example \cite{Cho10}, where the author deals with ranks of tensors over the real numbers.
In this subsection we report briefly on the approach outlined in \cite{AF20} to estimate and to compute the rank of a tensor using  polynomial equations. We start with the following lemma:
\begin{lemma}\label{estrankpoleq} Assume that $d\ge 3$ and $\cT\in\C^{\n}\setminus\{0\}$ is given.  Fix $r\in\N$ and consider the equality \eqref{Kruskaldg3r} as a system of polynomial equations in the entries of unknown vectors $\x_{i,j}$ for $i\in[r],j\in[d]$:
\begin{eqnarray}\label{estrankpoleq1}
\begin{aligned}
f_{k_1,\ldots,k_d}(\x_{1,1},\ldots,\x_{d,1},\ldots,\x_{1,r},\ldots,\x_{d,r})\equiv(\sum_{i=1}^r \otimes_{j=1}^d \x_{i,j}-\cT)_{k_1,\ldots,k_d}=0,\\ k_j\in [n_j], j\in[d].
\end{aligned}
\end{eqnarray}
Then $r( \cT)>r$ if and only if the following equivalent conditions hold:
\begin{enumerate}[(a)]
\item The above system of polynomial equations is not solvable.  
\item There  exist  $N(\n)$ polynomials
$g_{k_1,\ldots,k_d}(\x_{1,1},\ldots,\x_{d,1},\ldots,\x_{1,r},\ldots,\x_{d,r})$ for $ k_j\in [n_j]$, $j\in[d]$ of degree at most 
\begin{equation*}
L(r,\n)=d^{-1 +\min(N(n),r\sum_{i=1}^d n_i)}.
\end{equation*}
such that the following identity holds:
\begin{equation}\label{estrankpoleq2}
\sum_{k_1=1}^{n_1}\cdots\sum_{k_d=1}^{n_d} g_{k_1,\ldots,k_d}f_{k_1,\ldots,k_d}=1.
\end{equation}
\end{enumerate}
Furthermore, $r( \cT)\le r$ if and only if the following equivalent conditons hold:
\begin{enumerate}[(i)]
\item The system of polynomial equations \eqref{estrankpoleq1} is solvable. 
\item
There  are no $N(\n)$ polynomials
$g_{k_1,\ldots,k_d}(\x_{1,1},\ldots,\x_{d,1},\ldots,\x_{1,r},\ldots,\x_{d,r})$ for 

\noindent
$k_j\in [n_j]$, $j\in[d]$ of degree at most $L(r,\n)$
such that the identity \eqref{estrankpoleq2} holds.
\end{enumerate}
\end{lemma}
\begin{proof}
Clearly, $r<r( \cT)$ if and only if the system  \eqref{estrankpoleq1} is not solvable.  Hilbert's Nullstellensatz states that non-solvability of \eqref{estrankpoleq1} is equivalent to the existence of polynomials $g_{k_1,\ldots,k_d}(\x_{1,1},\ldots,\x_{d,1},\ldots,\x_{1,r},\ldots,\x_{d,r})$ that satisfy the identity \eqref{estrankpoleq2}.
  The claim that the degree of each $g_{k_1,\ldots,k_d}$ is at most $L(r,\n)$ is due to K\'ollar \cite{Kol}.

Suppose that $r(\cT)\le r$.  Then the system \eqref{estrankpoleq1} is solvable.
(Some of $\x_{i,j}$ can be zero.)  The identity \eqref{estrankpoleq2} cannot hold.  Part (b) yields that there
are no polynomials $g_{k_1,\ldots,k_d}(\x_{1,1},\ldots,\x_{d,1},\ldots,\x_{1,r},\ldots,\x_{d,r})$ for $ k_j\in [n_j]$, $j\in[d]$ of degree at most $M(r,\n)$
that satisfy \eqref{estrankpoleq2}.
\end{proof}

We now explain briefly how to use this lemma effectively of estimate or to compute the rank of $\cT$.  For $i\in[d]$ let $T_i$ be the $n_i\times (N(\n)/n_i)$ matrix obtained from $\cT$ by viewing $[n_1]\times\cdots\times [n_d]$ as $[n_i]\times ([n_1]\times\cdots [n_{i-1}]\times [n_{i+1}]\cdots\times [n_d])$. (We partition $\cT$ to a bipartite state.)
Let $r_i(\cT)$ be the matrix rank $r(T_i)$, which is easy to compute. (For $d=3$ those are ranks $r_A(\cT),r_B(\cT),r_C(\cT)$ introduced in \S\ref{subsec:basres3ten}.) As in \S\ref{subsec:basres3ten} we have $r(\cT)\ge rm=\max(r_1(\cT),\ldots,r_d(\cT))$.  Fix $r\ge rm$.  Write each $g_{k_1,\ldots,k_d}$ as a polynomial of degree $L(r,\n)$ with unknown monomial coefficients.  Then view the identity \eqref{estrankpoleq2}
as a huge system of linear equations in the unknown coefficients of monomials of $g_{k_1,\ldots,k_d}$ for $ k_j\in [n_j]$, $j\in[d]$.  If this system of linear equations is solvable we deduce that $r<r(\cT)$.  If this system is not solvable then $r\ge r(\cT)$.  To find $r(\cT)$ we start an algorithm with the above procedure for $r=rm$.  If this system of linear equations corresponding to \eqref{estrankpoleq2} is not solvable then $r(\cT)=r$.  Otherwise set $r=r+1$ and repeat the above procedure.  The main drawback of this algorithm for finding $r( \cT)$ is an exponential number of variables and equations in $d$.
\subsection{Generic identifiability of tensors}\label{subsec:genidgt}
\begin{definition}\label{identifdef}
Assume that $d\ge 3$.  A tensor $\cT\in \C^{\bn}$ is identifiable if its rank decomposition as a sum of rank-one decomposition
is unique up to the order of summation. 
\end{definition}
Note that for $d=2$ any matrix $T$ with $r(T)>1$ is not identifiable.
Lemma \ref{Kruskaldg3} gives sufficient conditions on
the identifiability of a tensor.
The obvious question arises, what happens if the inequality \eqref{rTKineq} does not hold.  A simple example is the case of $\cT\in\otimes^3\C^2$ discussed after Lemma \ref{ranks222ten}.  Namely a generic decomposition of $\cT$ as a sum of two rank-one tensors is a unique rank decomposition  of $\cT$, since it satisfies the condition \eqref{rTKineq}.  If $r(\cT)=3$ then any rank decomposition of $\cT$ does not satisfy \eqref{rTKineq}, and its rank decomposition is not unique.  
In particular a symmetric decomposition of $|W\rangle=\sum_{i=1}^3 \otimes^3\x_i$ is not unique.  It is shown in \cite{Der13} that Kruskal's theorem fails if we replace $2r+2$ in the right hand side of  \eqref{rTKineq} by  $2r+1$.

Consider the space of tensors $\C^{\n}$ where $d\ge 3$ and  $2\le n_1\le \cdots\le n_d$.  Fix $r>1$ and assume that there exists $\cT\in\C^{\n}$ such that $\rank \cT=r$.
Let $\bF_r$ be the map defined by \eqref{defFrd}.
A tensor $\cT$ is called a \emph{random tensor} in $\bF_r((\C^{n_1}\times \cdots\times \C^{n_d})^r)$ if $\cT=\bF_r(\x_{1,1},\ldots,\x_{d,1},\ldots,\x_{1,r},\ldots,\x_{d,r})$,
 and the coordinates of vectors $\x_{1,1},\ldots,\x_{d,r}$  are sampled from independent Gaussian complex-valued distribution.  We say that the identifiability property
  holds if for a random $\cT$ of the above form $r=r(\cT)$, and the decomposition $\cT=\bF_r(\x_{1,1},\ldots,\x_{d,1},\ldots,\x_{1,r},\ldots,\x_{d,r}))$ is
   unique up to a permutation of the summands $\otimes_{j=1}^d \x_{j,i}$.   Recall the inequality
\eqref{genrankdbd}.  By counting the parameters we deduce that if $ \frac{N(\n)}{M(\n)}$ is not an integer then for $r= r_{\textrm{gen}}(\n)$ identifiability property fails.
Thus it makes sense to consider the identifiability property for $r<r_{\textrm{gen}}(\n)$.
Theorem \ref{genndbig} states that $r_{\textrm{gen}}(\n)=n_d$ if 
\begin{equation}
\big(\prod_{j=1}^{d-1} n_j\big)+d-1-\sum_{j=1}^{d-1}n_j\le n_d.
\end{equation}
If $n_d$ satisfies the above inequality then the identifiability property holds if and only if \cite{BCO}:
\begin{equation}
r\le \big(\prod_{j=1}^{d-1} n_j\big)+d-2- \sum_{j=1}^{d-1}n_j.
\end{equation}
Thus it is enough to consider the identifiability property for
\begin{equation}
n_1\le \cdots\le n_{d-1}\le n_d\le \big(\prod_{j=1}^{d-1} n_j\big)+d-2- \sum_{j=1}^{d-1}n_j.
\end{equation}
It is shown in \cite{COV14} that if the above inequalities hold, and $\prod_{i=1}^d (n_i-1)\le 15 000$ 
then the identifiability property holds for $r<r_{\textrm{gen}}(\n)$ except the following cases of $\n=(n_1,\ldots,n_d)$ and the corresponding $r$:
\begin{eqnarray}
(4,5,5) \textrm{ and } r=5,\\
(5,5,5) \textrm{ and } r=6,\\
(4,7,7) \textrm{ and } r=8,\\
(3,3,n,n)  \textrm{ and } r=2n-1,\\
(3,3,3,3,3) \textrm{ and } r=5.
\end{eqnarray}
Domanov and de Lathauwer studied the identifiability property for $3$-mode tensors using mostly matrix methods in \cite{DD13I,DD13II,DD15,DD17}.

\section{Symmetric tensors}\label{sec:symten}
This section is devoted to symmetric tensors, which can be applied
 in quantum physics to describe systems of bosons.
In subsection \ref{subsec:hompol} we recall the known one-to-one correspondence between symmetric $d$-mode tensors and homogeneous polynomials of degree $d$.  In particular, the symmetric rank of a symmetric tensor is the Waring rank of a homogeneous polynomial.  Next we bring the celebrated result of 
Alexander-Hirschowitz \cite{AH95} that gives the formula for the generic symmetric rank except a number of known cases.  Subsection \ref{subsec:maxsymrk} discusses the recent upper bound of Buczy\'{n}ski-Han-Mella-Teitler on the maximum symmetric rank in terms of the generic symmetric rank \cite{BHMT17}.  We also provide 
 some known values of the maximum symmetric rank.
Subsection \ref{subsec:rankWd} shows that the rank of the
symmetric tensor  representing the state $|W_d\rangle$ is equal to $d$, 
while its border rank is $2$.
In subsection \ref{subsec:rprodten} we show explicitly that the rank of Kronecker and tensor products of quantum states  can be strictly submulitplicative. This is achieved  by considering the ranks of $|W\rangle\otimes_K|W\rangle$ and $|W\rangle\otimes|W\rangle$, which read $7$ and $8$ respectively, while the square of the rank of $|W\rangle$ is $9$.  In a short subsection \ref{subsec:compsymten} we discuss briefly computational methods for symmetric rank of symmetric tensors.  Subsection \ref{subsec:genidst} gives a short account of the results in \cite{COV}, which show that the generic identifiability property of symmetric tensors holds for a rank less than the symmetric generic rank, except a number of known cases.  
\subsection{Basic properties and relation to homogeneous polynomials}\label{subsec:hompol} A tensor $\cS=[\cS_{i_1,\ldots,i_d}]\in \otimes^d\C^n$ is called \emph{symmetric} if the value of the coordinates $\cS_{i_1,\ldots,i_d}$ does not change under the permutation of indices.  We denote by $\rS^d\C^n\subset\otimes^d\C^n$ the subspaces of all $d$-mode symmetric tensors over $\C^n$.   In physics this space is called the $(d,n)$ \emph{boson space}.
A symmetric $\cS$ is rank-one tensor if and only if $\cS=\otimes^d\x$, where $\x\in\C^n\setminus\{\0\}$.  There exists
one-to-one correspondence between symmetric tensors and the space of all homogeneous polynomials of degree $d$ in $n$ complex
 variables denoted as $\rP(d,n)$.  Indeed, let
$f(\x)=\langle \cS,\otimes^d\bar \x\rangle$, where $\langle\cdot,\cdot\rangle$ is the standard inner product in $\otimes^d\C^n$.  Then $f(\x)\in\rP(d,n)$.  Conversely, each polynomial $f(\x)\in\rP(d,n)$ induces a unique $\cS\in\rS^d\C^n$ as we explain below.  

We now introduce the standard multinomial notation.
Let $\Z_+$ be the set of all nonnegative integers. Denote by $J(d,n)$ the set
$J(d,n)= \big\{\bj=(j_1,\ldots,j_n)\in\Z_+^n,\; j_1+\cdots+j_n=d\big\}.$
Recall that $|J(d,n)|$, the cardinality of the set $J(d,n)$, is $n+d-1 \choose d$.
For $\x=(x_1,\ldots,x_n)\trans\in\C^n$ and $\bj=(j_1,\ldots,j_n)\in J(d,n)$ let $\x^{\bj}$ be the monomial $x_1^{j_1}\cdots x_n^{j_n}$.   Define $c(\bj)=\frac{d!}{j_1!\cdots j_n!}$.  Then $f(\x)\in\rP(n,d)$ expressed as a sum of monomials is given by:
\begin{equation}\label{defpolx1}
f(\x)=\sum_{\bj\in J(d,n)} c(\bj)f_{\bj}\x^{\bj}.
\end{equation}
Suppose that $f(\x)=\langle\cS,\otimes^d\bar \x\rangle$.  Then the correspondence between $f_{\bj}$ and the entries of  $\cS=[\cS_{i_1,\ldots,i_d}]$ is as follows.  Assume that $(i_1,\ldots,i_d)\in [n]^d$ is fixed.  For each $l\in[n]$ let $j_l$ be the number of times that $l$ appears in the sequence $i_1,\ldots,i_d$.  Set $\bj=(j_1,\ldots,j_n)$.  Then $f_{\bj}=\cS_{i_1,\ldots,i_d}$.  

Thus $\dim \rS^d\C^n={n+d-1\choose d}={n+d-1\choose n-1}$.  This dimension is usually significantly smaller than $\dim\otimes^d \C^n= n^d$.  For example for $n=2$, the space $\rS^d\C^2$, the boson $d$-qubit space, has dimension $d+1$, while the space $\otimes^d\C^2$, of $d$-mode qubits is $2^d$.  Thus $(d,n)$ bosons are much less entangled that then $d$-qunits. 

A symmetric rank decomposition of $\cS\in\rS^d\C^n$ is a decomposition of $\cS$ to a sum of rank-one symmetric tensors.  This is analogous to the Waring decomposition of $f\in\rP(d,n)$ to a sum of linear terms to the power $d$: $f(\x)=\sum_{i=1}^r \langle\x,\ba_i\rangle^d$, where $\ba_i\in\C^n\setminus\{\0\}$.  The minimal number of summands in
symmetric rank decomposition of symmetric $\cS$ is called the \emph{symmetric rank} of $\cS$, and denoted as $r_{\textrm{s}}(\cS)$\label{eq:def_symmetricrank}.  This is equivalent to the Waring rank of $f(\x)=\langle \cS,\otimes^d\bar\x\rangle$.  
The following inequality holds by definition,
$r(\cS)\le r_{\textrm{s}}(\cS)$.

We now recall two positive results when $r(\cS)=r_{\textrm{s}}(\cS)$.  For  a $(d,n)$ symmetric tensor $\cS$ denote by $r_A(\cS)$ the matrix rank of $\cS$ viewed as a bipartite state in $\C^n\otimes (\otimes^{d-1}\C^n)$.  As in \S\ref{subsec:basres3ten}
we deduce that $r_A (\cS)\le r(\cS)$.  In \cite{Fri16} it is shown that if $r(\cS)\in\{r_A(\cS),r_A(\cS)+1\}$ then $r(\cS)=r_{\textrm{s}}(\cS)$.  It is shown in \cite{ZHQ16}
that if $\cS\in\rS^d\C^n$ and $r_{\textrm{s}}(\cS)\le d$ then $r(\cS)=r_{\textrm{s}}(\cS)$.
 However, even for general $3$-mode symmetric tensors one has a strict inequality $r(\cS)<r_{\textrm{s}}(\cS)$ \cite{Shi17}.

As for general tensors, one can define a \emph{generic rank of $(d,n)$ symmetric tensor} as the symmetric rank of a random $\cS\in\rS^d\C^n$.
Denote by $r_{\textrm{gen}}(d,n)$\label{eq:def_genericsymmetricrank} the generic rank of $(d,n)$ symmetric tensor.
Note that  $\sum_{i=1}^r \langle\x,\ba_i\rangle^d$ has $rn$ complex parameters.  The dimension count yields the inequality
\begin{eqnarray}\label{symrankin}
r_{\textrm{gen}}(d,n)\ge \Bigg\lceil\frac{{n+d-1\choose d}}{n}\Bigg\rceil.
\end{eqnarray}
The celebrated Alexander-Hirschowitz result \cite{AH95} claims that equality holds in the above inequality except the following cases \cite{BO08}:
\begin{eqnarray*}
n=3, \quad d=4,\\
n=4, \quad d=4,\\
n=5, \quad d=3,\\
n=5, \quad d=4.
\end{eqnarray*}
In all the exceptional cases the value of generic rank is $\Big\lceil\frac{{n+d-1\choose d}}{n}\Big\rceil+1$.  Furthermore, in these exceptional cases, all tensors of border rank at most $\Big\lceil\frac{{n+d-1\choose d}}{n}\Big\rceil$ form a hypersurface in $\rS^d\C^n$.  (We thank G. Ottaviani for pointing out this fact to us.)

\subsection{Maximum symmetric rank}\label{subsec:maxsymrk}
Denote by $r_{\textrm{max}}(d,n)$\label{eq:def_maximumsymmetricrank} the maximum rank of $(d,n)$ symmetric tensors.
The following analogue of \eqref{maxgenrin} is proved in \cite{BHMT17}:
\begin{eqnarray}\label{maxgenrsym}
r_{\textrm{max}}(d,n)\le 2r_{\textrm{gen}}(d,n)-1.
\end{eqnarray}
This bound can be further improved \cite{BHMT17} to $r_{\textrm{max}}(d,n)\le 2r_{\textrm{gen}}(d,n)-2$ if the variety of all symmetric tensors of border symmetric rank at most $r_{\textrm{gen}}(d,n)-1$ is a hypersurface.  This assumption holds in all the above exceptional cases.

We now discuss briefly the known maximum ranks.  The first nontrivial case is $r_{\textrm{max}}(2,3)$.  As $r_{\textrm{max}}(2,2,2)=3$ and $r(|W\rangle)=3$ we deduce that $r_{\textrm{max}}(3,2)=3$.  
Observe that the relation
$r_{\textrm{gen}}(3,2)=r_{\textrm{gen}}(2,2,2)=2$ implies
that inequality \eqref{maxgenrsym} is not sharp in this case.

The following maximum ranks are known.  We also display the value of the generic symmetric rank in these cases:
\begin{eqnarray}
&&r_{\textrm{max}}(d,2)=d \quad \textrm{\cite{CS11}, \cite[\S3.1]{BT15}},\quad r_{\textrm{gen}}(d,2)=\Big\lceil \frac{d+1}{2}\Big\rceil,\\
&&r_{\textrm{max}}(3,3)=5 \quad \textrm{\cite[\S96]{Seg42}, \cite{CM96}, \cite{LT10}},
\quad r_{\textrm{gen}}(3,3)=4,\\
&&r_{\textrm{max}}(4,3)=7  \quad \textrm{\cite[\S97]{Seg42}, \cite{Kle99}, \cite{Par15}},
\quad r_{\textrm{gen}}(4,3)=6,\\
&&r_{\textrm{max}}(5,3)=10 \quad \textrm{\cite{Par15},\cite{BT16}}, \quad r_{\textrm{gen}}(5,3)=7.
\end{eqnarray}
In \S\ref{subsec:rprodten} we show that $r_{\textrm{max}}(3,4)\ge 7$. See also \cite{BT16}.

\subsection{The rank of $|W_d\rangle$}\label{subsec:rankWd}
Denote by $|W_d\rangle\in\rS^d\C^2$ the symmetric tensor corresponding to polynomial
$dx_1^{d-1}x_2$:
\begin{eqnarray}
|W_d\rangle=\sum_{j=0}^{d-1}(\otimes^{d-j-1}\e_1)\otimes \e_2\otimes(\otimes^{j}\e_1).
\end{eqnarray} 
Hence $r(|W_d\rangle)\le d$.
We claim that $r(|W_d\rangle)=r_{\textrm{s}}(|W_d\rangle)=d$.  As $|W\rangle=|W_3\rangle$, we know that $r(|W_3\rangle)=3$.  We first claim that   $r(|W_d\rangle)=d$ \cite{Bla13}.  We prove that by induction on $d=k\ge 3$.  Suppose that  $r(|W_k\rangle)=k$  for $k\ge 3$.   Assume to the contrary that 
\begin{eqnarray*}
|W_{k+1}\rangle=\sum_{i=1}^r \otimes_{j=1}^{k+1} \x_{i,j}, \quad \x_{i,j}\in\C^2, \;j\in[k+1], \;i\in[r], \; r<k+1.
\end{eqnarray*}
Observe that 
$|W_{k+1}\rangle=|W_{k}\rangle\otimes \e_1+\e_1^{\otimes k}\otimes\e_2.$
Hence span$(\x_{1,k+1},\ldots,\x_{r,k+1})=\C^2$.   For $\y\in\C^2$ let $|W_{k+1}\rangle\times \y=\sum_{i=1}^r (\y\trans\x_{i,k+1})\otimes_{j=1}^k \x_{i,j}$ be the contraction with respect to the last coordinate.  Choose $\x_{l,k+1}$ which is linearly independent to $\e_1$.  Let $\y\in\C^2\setminus\{\0\}$ satisfy $\y\trans\x_{l,k+1}=0$.
Hence $\y\trans\e_1\ne 0$ and we fix $\y$ by letting $\y\trans\e_1=1$.  Thus
$\cT=|W_{k+1}\rangle \times \y$ equals to $\sum_{i\in[r]\setminus\{l\}} (\y\trans\x_{i,k+1})\otimes_{j=1}^k \x_{i,j}$.  Therefore $r(\cT)\le k-1$.  Observe next that $\cT=|W_k\rangle +(\y\trans \e_2)\e_1^{\otimes k}$.  
Furthermore, 
$\cT$ is a symmetric tensor which corresponds to the polynomial 
\begin{eqnarray*}
f(x)=dx_1^{d-1}x_2+(\y\trans \e_2) x_1^d=dx_1^{d-1}(x_2+((\y\trans \e_2)/d)x_1).
\end{eqnarray*}
Change coordinates $(x_1,x_2)$ to $(x_1, x_2+((\y\trans \e_2)/d)x_1)$ to deduce that $\cT$ is in the orbit of $|W_k\rangle$.  Hence $r(\cT)=k$ which contradicts our assumption that $r(|W_{k+1}\rangle)<k+1$.  Thus $r(|W_{k+1}\rangle)=k+1$.

Observe finally that
$d=r(|W_d\rangle) \le r_{\textrm{s}}(|W_d\rangle)\le r_{max}(d,2)=d.$

We close this subsection with the well known fact that $r_{\mathrm{b}}(|W_d\rangle)=2$.  As the set of rank-one states is closed it follows
that $r_{\mathrm{b}}(|W_d\rangle)\ge 2$.  On the other hand we have the equality
\begin{eqnarray}
|W_d\rangle=\lim_{t\to 0} \frac{1}{t}\Big((\e_1+t\e_2)^{\otimes d}-\e_1^{\otimes d}\Big).
\end{eqnarray}
\subsection{Tensor rank of product of tensors}\label{subsec:rprodten}
\begin{lemma}\label{UVtenrin}
Let $\cU\in\C^{\m},\cV\in \C^{\n}$ be two tensors, where $\m=(m_1,\ldots,m_p),\n=(n_1,\ldots,n_q)$.  Then
\begin{equation}\label{ranktenprodin}
r\left(\cU\otimes_K\cV\right)\le r(\cU\otimes\cV)\le r(\cU)r(\cV).
\end{equation}
Furthermore
\begin{enumerate}[(a)]
\item
If $\cU\in\C\otimes\C^d\otimes\C^d$ and $\cV\in \C^2\otimes \C^n\otimes\C^m$ then equalities hold in \eqref{ranktenprodin}.  
\item
Suppose that $\cU=\cV=|W\rangle$.  Then strict inequalities hold in \eqref{ranktenprodin}.  More precisely
\begin{equation}\label{valranWW}
r(|W\rangle\otimes_K|W\rangle)=7, \quad r(|W\rangle\otimes |W\rangle)=8.
\end{equation}
\end{enumerate}
\end{lemma}
\begin{proof}
If $p=q$ then we have the inequalities \eqref{rnaktenKrnin}.  The general case follows from the same arguments.
The equality $r\left(\cU\otimes_K\cV\right)=r(\cU)r(\cV)$ yields the equality $r(\cU\otimes\cV)=r(\cU)r(\cV)$.
It is easy to show that if either $p=1$ or $q=1$ then equality holds in \eqref{ranktenprodin}.  Indeed, it is enough to assume that $q=1$ and $\cV\ne 0$.  As in the proof that $r(|W_d\rangle)=d$ we deduce equality by contracting the last index in $\cU\otimes\cV$.  Since for
matrices $r\left(\cU\otimes_K\cV\right)=r(\cU)r(\cV)$, it follows that for $p=q=2$ equality holds in \eqref{ranktenprodin}.  Corollary \ref{varestStrcon} gives an example when one has equalities in \eqref{rnaktenKrnin} for special two $3$-tensors.

We now show part (a) of the Lemma.
Proposition 22 in \cite{CJZ18} gives the following application of Theorem \ref{Jjajathm}:
\begin{eqnarray*}
r(\cX\otimes _K\cY)=r(\cX\otimes \cY)=r(\cX)r(\cY), \quad \cX\in\C\otimes\C^d\otimes\C^d, \cY\in \C^2\otimes \C^n\otimes\C^m.
\end{eqnarray*}
 Note that $\cX$ can be viewed as a matrix $X\in \C^d\otimes \C^d$. 
 Any matrix $T\in\C^{p\times q}$ can be trivially extended to a bigger matrix $X\in\C^d\times \C^d$ for $d= \max(p,q)$ by adding additional  zero rows or columns.
It is straightforward to show that $r(T\otimes \cY)=r(X\otimes \cY)$.
Hence
\begin{eqnarray}\label{CJZthm}
r(X\otimes \cY)=r(\cX)r(\cY), \quad X\in\C^p\otimes\C^q, \cY\in \C^2\otimes \C^n\otimes\C^m.
\end{eqnarray}
A special case of this equality is proved independently in the first part of Proposition 9 in \cite{CF18}.

We now show part (b) of the Lemma.
First we show that one can have strict inequalities in \eqref{CJZthm} for $\cU=\cV=|W_3\rangle$.   Assume that $\cX=|W_3\rangle\otimes_K|W_3\rangle\in\otimes^3\C^4$.  It will be convenient to use Dirac notation, where 
\begin{eqnarray}
|00\rangle=|0\rangle, \;\; |01\rangle=|1\rangle,\;\;|10\rangle=|2\rangle,\;\; |11\rangle=|3\rangle.
\end{eqnarray}
Then
\begin{eqnarray*}
&&\cX=(|001\rangle+|010\rangle+|100\rangle)\otimes_K(|001\rangle+|010\rangle+|100\rangle)=\\
&&|003\rangle+|012\rangle+|102\rangle+|021\rangle+|030\rangle+|120\rangle+|201\rangle+|210\rangle+|300\rangle.
\end{eqnarray*}
 The above three tensor is symmetric on $\C^4$ and  it corresponds to the following polynomial  of degree three,  $f(x_1,x_2,x_3,x_4)=3x_1^2x_4+6x_1x_2x_3$.
Observe next that \cite{CCDJW10, CF18}:
\begin{eqnarray*}
&&6x_1^2x_4=(x_1+x_4)^3-(x_1-x_4)^3-2x_4^3,\\
&&24x_1x_2x_3=(x_1+x_2+x_3)^3-(-x_1+x_2+x_3)^3-(x_1-x_2+x_3)^3-(x_1+x_2-x_3)^3.
\end{eqnarray*}
This implies that $r_{\textrm{s}}(\cX)\le 7$.  

We now follow the arguments of \cite{YCGD10} to show that $r(\cX)\ge 7$.  First observe that the four frontal sections of $\cX$ form the following four matrices:
\begin{eqnarray*}
A_1=|00\rangle \langle 00|=\left[\begin{array}{cccc}1&0&0&0\\0&0&0&0\\0&0&0&0\\0&0&0&0\end{array}\right],\quad  
A_2=|01\rangle \langle 01|+|10\rangle \langle 10|=\left[\begin{array}{cccc}0&1&0&0\\1&0&0&0\\0&0&0&0\\0&0&0&0\end{array}\right],\\
A_3=|02\rangle \langle 02|+|20\rangle \langle 20|=\left[\begin{array}{cccc}0&0&1&0\\0&0&0&0\\1&0&0&0\\0&0&0&0\end{array}\right], \\  
A_4=|03\rangle \langle 03|+|12\rangle  \langle 12|+|21\rangle \langle 21|+
|30\rangle \langle 30|=\left[\begin{array}{cccc}0&0&0&1\\0&0&1&0\\0&1&0&0\\1&0&0&0\end{array}\right].
\end{eqnarray*}
Matrices $A_1,A_2,A_3,A_4$ are linearly independent.  Furthermore $\det (A_4+a_1A_1+a_2A_2+a_3A_3)=1$, which is obtained by expanding
 this determinant by the rows $4,3,2,1$.

Assume to the contrary that $r(\cX)=r<7$.  As $r_3(\cX)=4$ we have that $r\ge 4$.
Let $B_1,\ldots,B_r$ be $r$ linearly independent rank-one matrices so that they span the subspace $\bV\subset \C^{4\times 4}$, which contains $A_1,\ldots,A_4$.  As $A_1$, $A_2$ and $A_3$ are independent we have a basis in $\bV$ consisting of $A_1,A_2,A_3$ and $C_1,\ldots,C_{r-3}\in\{B_1,\ldots,B_r\}$.  Express $A_4$ in this basis to deduce that $A_4+\sum_{i=1}^3 a_i A_i=\sum_{j=1}^{r-3} c_j C_j$.  That is $r( A_4+\sum_{i=1}^3 a_i A_i)\le r-3\le 3$. This contradicts the equality $\det (A_4+a_1A_1+a_2A_2+a_3A_3)=1$.

We now show that the rank of $\cY=|W_3\rangle^{\otimes 2}\in \otimes^6\C^2$ is eight.  We first give a simple decomposition of $\cY$ as a sum of $8$ rank-one tensors as in \cite{CJZ18}.  Recall that the generic rank of tensors in $\otimes^3\C^2$ is two.
Hence most rank-one perturbation of $|W_3\rangle$ have rank two.  For example, for $\z=|0\rangle$ the two tensors $|W_3\rangle + \z^{\otimes 3}$ and  $|W_3\rangle +\frac{1}{2} \z^{\otimes 3}$ have rank two.  Next observe that 
\begin{eqnarray*}
|W_3\rangle^{\otimes 2}=\big(|W_3\rangle^{\otimes 2}+\z^{\otimes 3}\big)^{\otimes 2} -\big(|W_3\rangle +\frac{1}{2} \z^{\otimes 3}\big)\otimes \z^{\otimes 3} -\z^{\otimes 3}\otimes
\big(|W_3\rangle +\frac{1}{2} \z^{\otimes 3}\big).
\end{eqnarray*}
Use the inequality \eqref{ranktenprodin} for each tensor product appearing in the right hand side of the above identity
to deduce that $r(|W_3\rangle^{\otimes 2})\le 4+2+2=8$.  

We now outline briefly the main arguments in \cite{CF18} to show that $r(|W_3\rangle^{\otimes 2})\ge 8$.
Recall that $r(|W_3\rangle^{\otimes 2})\ge r(\cX)=7$.  Assume to the contrary 
\begin{eqnarray*}
 |W_3\rangle^{\otimes 2}=\sum_{i=1}^7 \otimes_{j=1}^6\ba_{j,i}.
\end{eqnarray*}
We claim that for each $i\in[7]$ either $\ba_{1,i},\ba_{2,i},\ba_{3,i}\in$ span$(|0\rangle)$
or $\ba_{4,i},\ba_{5,i},\ba_{6,i}\in$ span$(|0\rangle)$. Suppose the opposite case. Then we may assume that this dichotomy does not hold for $i=7$.  Since each copy of $|W_3\rangle$ is symmetric, by permuting the first and the last $3$ components of 
$|W_3\rangle^{\otimes 2}$, we can assume that $\ba_{1,7}$ and $\ba_{6,7}$ are not in span$(|0\rangle)$.  Contract  $|W_3\rangle^{\otimes 2}$ with respect to the first coordinate using a vector $\x$ orthogonal to $\ba_{1,7}$:
\begin{eqnarray}
\x\times |W_3\rangle^{\otimes 2}=(\x\times |W_3\rangle)\otimes |W_3\rangle=\sum_{i=1}^6 (\x\trans \ba_{1,i})\otimes_{j=2}^6 \ba_{j,i}.
 \end{eqnarray} 
 Observe next that since $\x\ne c|1\rangle$ it follows that the rank of the $2\times 2$ matrix $\x\times |W_3\rangle$ is two.   Use \eqref{CJZthm}
to deduce that $r\left((\x\times |W_3\rangle)\otimes |W_3\rangle\right) =6$. 
 
  The second part of Proposition 9 in \cite{CF18} states that the following six $3$-mode rank-one tensors are linearly dependent:
 \begin{eqnarray}
 \ba_{2,i}\otimes \ba_{3,i}\otimes\ba_{p,i}, \quad i\in[6] \textrm{ for } p\in\{4,5,6\},
 \end{eqnarray}
 and the following six $3$-mode rank-one tensors are linearly independent:
 \begin{eqnarray}
 \ba_{p,i}\otimes \ba_{q,i}\otimes\ba_{r,i},\quad i\in[6] \textrm{ for } p\in\{2,3\}, 4\le q<r\le 6.
 \end{eqnarray}
 Next contract $|W_3\rangle^{\otimes 2}$ on the last mode with respect to $\y$ orthogonal to $\ba_{6,7}$ and use \cite[Proprosition 9]{CF18} to deduce that the following six vectors are linearly independent: $\ba_{2,i}\otimes \ba_{3,i}\otimes \ba_{4,i}, i\in[6]$.  This contradicts to the previous statement that these six tensors are linearly dependent.  
 
 Thus we showed that for each $i\in[7]$ either $\ba_{1,i}$, $\ba_{2,i}$, $\ba_{3,i}\in$ span$(|0\rangle)$
or $\ba_{4,i}$, $\ba_{5,i}$, $\ba_{6,i}\in$ span$(|0\rangle)$.  We now contradict this statement.  Assume first that $\ba_{1,i},\ba_{2,i},\ba_{3,i}\in$ span$(|0\rangle)$ for each $i\in[7]$.  Then $|W_3\rangle^{\otimes 2}=|0\rangle^{\times 3}\otimes \cZ$ for some $\cZ\in \otimes^3\C^2$.  Thus $r(\cZ)\le 3$
and $r(|W_3\rangle^{\otimes 2})\le 3$ which is impossible.  Similarly, one cannot have  $\ba_{4,i},\ba_{5,i},\ba_{6,i}\in$ span$(|0\rangle)$ for $i\in[7]$.
Hence $r(|W_3\rangle^{\otimes 2})\ge 8$.
\end{proof}

We now discuss briefly the ranks of $|W_3\rangle^{\otimes k} \in\otimes^k \C^{8}$ and
$\otimes_K^k |W_3\rangle\in\otimes^k\C^{2^k}$.  It is shown in \cite{YCGD10} that $r(\otimes_K^k |W_3\rangle)\ge 2^{k+1}-1$, similar to the arguments we gave for the case $k=2$.  Hence $r(|W_3\rangle^{\otimes k})\ge 2^{k+1}-1$.  In particular,
$r(|W_3\rangle^{\otimes 3})\ge 15$.   
It is known \cite[Theorem]{Zui17} that
$ r(\otimes_K^3|W_3\rangle)=16$. 
Combine this result with \cite{CF18} to obtain
that
$16\le  r(|W_3\rangle^{\otimes 3})\le 20$.  In \cite{CCDJW10} it is shown that
\begin{eqnarray}
r(|W_d\rangle^{\otimes k})\ge r(|W_3\rangle^{\otimes k})+(d-3)(2^k-1).
\end{eqnarray}

A real sequence $\{a_k\}, k\in\N$ is called \emph{subadditive} if $a_{p+q}\le a_{p}+a_{q}$ for every $p,q\in\N$.  Fekete's subadditive lemma claims that
for any subadditive sequence ${a_k}$ with $k\in\N$, the modified sequence converges $\lim\limits_{k\to\infty}\frac{a_k}{k}=a$ with $a\in[-\infty,\infty)$.

Let $\cT\in \C^{\n}$.  The inequality \eqref{rnaktenKrnin} yields that the two sequences
$\log r(\otimes_K^k \cT)$ and $\log r(\cT^{\otimes k})$ are subadditive.  Let
\begin{eqnarray*}
r_{\textrm{lim}_K}(\cT)=\lim_{k\to\infty} \Big(r(\otimes_K \cT)\Big)^{\frac{1}{k}}\qquad\text{and}\qquad
r_{\textrm{lim}}(\cT)=\lim_{k\to\infty} \Big(r(\cT^{\otimes k})\Big)^{\frac{1}{k}}.
\end{eqnarray*}
 By definition
 $r_{\textrm{lim}_K}(\cT) \le r_{\textrm{lim}}(\cT)$,
 while  Corollary 12 in \cite{CJZ18} claims that
$r_{\textrm{lim}}(\cT) \le r_{\textrm{b}}(\cT)$. 
Hence the above results for $|W_d\rangle$ yield the equalities 
\begin{equation}r_{\textrm{lim}_K}(|W_d\rangle)=r_{\textrm{lim}}(|W_d\rangle)=2.\end{equation}
\subsection{Computational methods for symmetric rank of symmetric tensors}\label{subsec:compsymten}
Recall that the symmetric rank  of a symmetric tensor is the Waring rank of the corresponding homogeneous polynomial.  Hence one can use theoretical methods of algebraic geometry and the available software as Bertini \cite{BHSW06}.
The Waring rank of a homogeneous polynomial of degree $d$ in two variables can be determined very efficiently using Sylvester's algorithm \cite{Syl86}. The paper  \cite{BGI} discusses Sylvester's algorithm in the modern language of algebraic geometry.  The authors discuss also algorithms to find 
small border ranks of symmetric tensors.  In \cite{AF20} the authors provide methods
of linear algebra to find the rank of symmetric tensors 
similar to the algorithm for determining the rank of  a general tensor discussed in \S\ref{subsec:estrankpoleq}.
\subsection{Generic identifiability of symmetric tensors}\label{subsec:genidst}
In this subsection we discuss the identifiability property for symmetric tensors, which is similar to our discussion of the identifiability property for general tensors in \S\ref{subsec:genidgt}.
Assume that $d\ge 3$ and $n\ge 2$.  Suppose that $r>1$ is an integer, and there exists a tensor $\cS\in\cS^d\C^{n}$ of symmetric rank $r$.  Denote
\begin{equation}
\bG_r:(\C^n)^r\to \rS^d\C^n,\quad \bG_r(\x_1,\ldots,\x_r)=\sum_{i=1}^r \otimes^d\x_i.
\end{equation}
A tensor $\cS$ is called a \emph{random tensor} in $\bG_r((\C^n)^r)$ if $\cS=\bG_r(\x_1,\ldots,\x_r)$, and the coordinates of vectors $\x_1,\ldots,\x_r$  are sampled from independent Gaussian complex valued distribution.  We say that the identifiability property holds if for a random $\cS$ of the above form $r=r_{\textrm{s}}(\cS)$, and the decomposition $\cS=\bG_r(\x_1,\ldots,\x_r)$ is unique up to a permutation of the summands $\otimes^d\x_i$.

Recall the inequality \eqref{symrankin}: $r_{\textrm{gen}}(d,n)\ge \Bigg\lceil\frac{{n+d-1\choose d}}{n}\Bigg\rceil$.  Observe that if $\Bigg\lceil\frac{{n+d-1\choose d}}{n}\Bigg\rceil$ is not an integer then by counting the number of parameters we deduce that identifiability property fails for $r=r_{\textrm{gen}}(d,n)$.  Similarly,  the identifiability property fails for $r>r_{\textrm{gen}}(d,n)$.  The fundamental result in \cite{COV} states that the identifiability property holds for $r<r_{\textrm{gen}}(d,n)$ except the following cases:
\begin{eqnarray*}
d=6, \, n=3, \textrm{ and } r=9;\\
d=4, \, n=4, \textrm{ and } r=8;\\
d=3, \, n=6, \textrm{ and } r=9.
\end{eqnarray*}
In the above exceptional cases a generic tensor of the corresponding rank has two distinct Waring decompositions.
We list additional references on the identifiability property: \cite{ACM20,ACV,BaC13}.

\section{Nuclear rank of a tensor}
\label{sec:nucrank}
Section \ref{sec:nucrank} is mainly devoted to the notion of the nuclear rank of a tensor.  In subsection \ref{subsec:specnrm} we discuss the spectral norm and the
geometric measure of entanglement.  We point out the concentration law 
concerning the geometric measure of entanglement. 
It states that this measure of entanglement 
of a random symmetric quantum state generated with respect to the Haar measure
is close to the maximal possible value.
Subsection \ref{subsec:nucnorm} introduces the  nuclear norm and nuclear rank. 
The minimal nuclear decomposition of a tensor plays the role analogous to the
singular value decomposition of a matrix. 
Hence from the point of view of applications in quantum physics,
the nuclear rank of a tensor seems to be the right analog of the rank of a matrix.
\\
In subsection \ref{subsec:face}  we discuss the faces of the unit ball with respect
to  the nuclear norm.  Lemma \ref{facdes} characterizes the exposed faces of such a unit ball.
We consider also the restriction of the nuclear norm to symmetric tensors,
 which gives rise to the definition of a symmetric nuclear rank.
Subsection \ref{subsec:matnnrm} concerns the exposed faces and facets of unit balls
with respect to matrix nuclear norm and spectral norm.  
Theorem \ref{exfcspnucnrm} characterizes the exposed faces of these unit balls.
 In subsection \ref{subsec:GHZ} we show that the nuclear rank of $|GHZ\rangle$ state 
  is $2$.  Subsection \ref{subsec:gn3qub} discusses the  generic and maximum nuclear rank of symmetric states of a three-qubit system.  Theorem \ref{dimfw3} characterizes the face nuclear norm in $2\times 2\times 2$ symmetric tensors
 which contains the state $|W\rangle$.

\subsection{Geometric measure of entanglement and spectral norm}
\label{subsec:specnrm}
Denote the space of all product states in $\C^{\n}$ by 
\begin{eqnarray}
\Pi(\n)=\big\{\cP\in \C^{\n},\cP=\otimes_{j=1}^d\x_{j}, \x_j\in\C^{n_j}, j\in[d],
\|\cP\|_2 =1\big\}.
\end{eqnarray}
For any multipartite state  $|\psi\rangle$ 
represented by a tensor $\cT\in\C^{\n}$ normalized by 
a fixed Hilbert--Schmidt norm,
$||\cT||_2=1$, 
its entanglement can be characterized by the Fubini--Study distance
of $|\psi\rangle$ 
to the set of product states \cite{ZB02,WG03}.
This quantity can be
related to the \emph{spectral norm} of $\cT$,
\begin{equation}\label{defspecnrm}
\|\cT\|_{\infty}=\max\big\{|\langle\cT,\cP\rangle|, \cP\in\Pi(\n)\big\}.
\end{equation}
In analogy to the bipartite case, corresponding to matrices,
one defines 
the geometric measure of entanglement of the state $\cT$ by
$\sqrt{2(1-\|\cT\|_{\infty})}$
which corresponds to the minimal Hilbert--Schmidt distance
between the projector
$\rho_{\psi} =|\psi\rangle \langle \psi|$
and the projector on a separable state
-- see  Section  \ref{subsec:entangle}.  

 We now make a few comments on the spectral norm of $\cT$. 
First, note that  $\|\cT\|_{\infty}=\max\{\Re\langle\cT,\cP\rangle, \cP\in\Pi(\n)\}$.  Next, observe that for $d=2$, i.e., matrices, $\|T\|_{\infty}$ is the leading singular value 
$\sigma_{\rm max}$ of the matrix $T$,
which is also the spectral norm of $T$, viewed as a linear transformation from $\C^{n_2}$ to $\C^{n_1}$. Note that
  $\|T\|_{\infty}$ sometimes denotes the operator norm of a matrix $T$, where $\C^{n_1}$ and $\C^{n_2}$ are endowed with
   $\infty$-norms, which is different from $\sigma_1(T)$.

Assume that $\|\cT\|_{\infty}=\Re\langle\cT,\cP\rangle$ for some $\cP\in \Pi(\n)$.
Then $\|\cT\|_{\infty}\cP$ is the best rank-one approximation of $\cT$: $\|\cT-\|\cT\|_{\infty}\cP\|\le \|\cT-\cX\|$, where $\cX$ is rank-one tensor \cite{LMV00}.  For a matrix  $A\ne 0$ a best rank-one approximation 
$\|A\|_{\infty}\cP$ is the term $\sigma_1\bu_1\bv_1^*$ in the SVD decomposition \eqref{SVDdec} \cite{SC10} or \cite[Corollary 4.13.2]{Fribook}. Furthermore, $r(A-\|A\|_{\infty}\cP)=r(A)-1$.  This equality is not true for tensors with $d\ge 3$ indices \cite{SC10}.

Assume that  $\cT$ has real entries. Then we can define the real spectral norm as 
$\|\cT\|_{\infty,\R}=\max\{|\langle\cT,\cP\rangle|, \cP\in\Pi(\n)\cap \R^{\n}\}$.  
By definition, the following inequality holds,  
$\|\cT\|_{\infty,\R}\le \|\cT\|_{\infty}$
which is saturated for   bipartite states, $d=2$,
represented by matrices $T$.
 However, 
 for $d\ge 3$ one can have a strict inequality already for the space of $3$-qubits \cite{FL18}.  
That is, the closest product state to a real state may be complex-valued.
The computation of spectral norms of $\|\cT\|_{\infty,\R}$  and $\|\cT\|_{\infty}$ for $d\ge 3$ is NP-hard \cite{HL13, FL18}.

In the case of a bipartite state 
represented by a matrix with the Hilbert--Schmidt norm fixed,
the smaller spectral norm, the larger quantum entanglement -- see Section \ref{subsec:entangle}.
The similar reasoning holds for multipartite states represented by a tensor.
  Hence it makes sense to introduce the following measure of entanglement, equivalent to the geometric measure of entanglement \cite{GFE09}:
\begin{eqnarray}\label{defetaT}
\eta(\cT):=-\log_2 \|\cT\|_{\infty}^2.
\end{eqnarray}
One can estimate $\eta(\cT)$ from above as follows.  Expand $\cT$ in terms of the orthonormal basis in $\C^{\n}$ which consists of product states corresponding to a choice of an orthonormal basis $\e_{1,i},\ldots,\e_{n_i,i}$ in each $\cH_{n_i}$ for $i\in[d]$.  
  The absolute value of the coefficient of $\cT$ with respect to $\otimes_{j=1}^d \e_{k_i,n_i}$ is $|\langle\cT,\otimes_{j=1}^d \e_{k_i,n_i}\rangle|\le |\cT|_{\infty}$.
As $\dim \C^{\n}=N(\n)$ and $\|\cT\|^2=1$  we easily deduce that $\|\cT\|_{\infty}^2\ge \frac{1}{N(n)}$.  A slightly better estimation   $\|\cT\|_{\infty}^2\ge \frac{\max(n_1,\ldots,n_d)}{N(n)}$ is given in \cite{LNSTU}.  Hence
\begin{eqnarray*}
0\le \eta(\cT)\le \log_2 \frac{N(\n)}{\max(n_1,\ldots,n_d)}.
\end{eqnarray*}
Assume that $\n=2^{\times d}$.  Then $\eta(\cT)\le d-1$.   It is shown in \cite{GFE09}, using the concentration of the Haar measure on the manifold of states in $\C^{2^{\times d}}$, that  
\begin{eqnarray}
\mathrm{P}\Big[\eta(\cT)\ge d-2\log_2(d)-3\Big]\ge 1-e^{-d^2}, \qquad \text{for }d\ge 11,
 \end{eqnarray}
 where the probabilities are considered with respect to the unitarily invariant measure
 on the space of $d$-qubit states  induced by the Haar measure  on $U(2)^{\otimes d}$.
 The physical interpretation is that most of the $d$-qubits are strongly entangled  for $d\gg 1$.
 A generalization to $\n=n^{\times d} $ for a fixed number $n\ge 3$ and $d\gg 1$ is given in \cite{DM18}.
 
 We now consider the symmetric tensors $\rS^d\C^n\subset \C^{n^{\times d}}$.  
 A fundamental result of Banach \cite{Ban38} claims
 \begin{eqnarray}
 \|\cS\|_{\infty}=\max\big\{|\langle \cS,\x^{\otimes d}\rangle|, \x\in\C^{n}, \|\x\|=1\big\},\quad \cS\in\rS^d\C^n.
 \end{eqnarray}
 That is, the geometric measure of entanglement of a symmetric state is achieved at a symmetric product state.
 This characterization was rediscovered in \cite{Hubetall09}. 
  
 Assume that $n\ge 2$ is fixed and $d\gg 1$.  We claim that the typical 
  symmetric states are much less entangled
   than the general states of $d$ qunits with respect to 
the discrete measure of entanglement given by the generic rank of a tensor. 
 Recall \eqref{genqunit} and the resulult by  Alexander-Hirschowitz \cite{AH95},
 \begin{eqnarray}
  r_{\textrm{gen}}(n^{\times d})\ge \frac{n^d}{d(n-1)+1}>\frac{n^{d-1}}{d},\\
  \nonumber\\
  r_{\textrm{gen}}(d,n)= \Bigg\lceil\frac{{n+d-1\choose d}}{n}\Bigg\rceil=\Bigg\lceil\frac{{n+d-1\choose n-1}}{n}\Bigg\rceil = O(d^{n-1}) .
 \end{eqnarray}
 Thus $ r_{\textrm{gen}}$ has an exponential growth in $d$
 in contrast to the  polynomial growth of  $r_{\textrm{gen}}(d,n)$. 
  This fact can be explained by observing that the dimension of $\otimes^d\C^n$ is exponential in $d$, while the dimension of $\rS^d\C^n$ is polynomial in $d$.
   
  It is shown in \cite{FK18} that
 \begin{eqnarray}
 0\le \eta(\cS)\le \log_2 {n+d-1\choose d} =\log_2  {n+d-1\choose n-1}.
 \end{eqnarray}
 Thus, for $n=2$ we have that $\eta(\cS)\le \log_2 (d+1)$.  
 There is still the concentration
 law which shows that most of symmetric  tensor for fixed $n$ and $d\gg 1$ concentrate at the upper bound given above \cite{FK18}.  In particular, for symmetric $d$-qubits one has the inequality:
 \begin{eqnarray}
 \mathrm{P}\Big[\eta(\cT)\ge \log_2d-\log_2(\log_2 d) -3\Big]\ge 1-\frac{1}{2d^{5/2}}, \quad \text{for }d\ge 42.
\end{eqnarray}
Above results show that for a fixed $n\ge 2$ and $d\gg 1$ a symmetric state
is typically much less entangled with respect to the geometric measure of entanglement
than a generic states of the same dimension.

The computation of the spectral norm of $\cS\in\rS^d\C^n$ is NP-hard in $n$ for $d=3$ \cite{FW18}.  However, for a fixed $\cS$, the computation of $\|\cS\|_{\infty}$ is polynomial in $d$ \cite{FW18}.  This result is obtained by showing that the computation of $\|\cS\|_{\infty}$ can be done by solving polynomial equations for the critical points of the function $\Re \langle \cS,\otimes^d\x\rangle$ restricted to the unit sphere $\|\x\|=1$.
\subsection{Nuclear norm and nuclear rank}
\label{subsec:nucnorm}
Denote by the nuclear norm $\|\cdot\|_1$ the dual norm to the spectral one on $\C^{\n}$.
From the definition of the spectral norm  it follows that the unit ball of the nuclear norm is the convex hull of $\Pi(\n)$.  As each $\cP\in\Pi(\n)$ is the extreme point on the unit sphere of the Hilbert--Schmidt norm, we deduce that each $\cP$ is an extreme point on the unit sphere of the nuclear form.  One can show that the nuclear norm has the following minimum characterization \cite{FL18}:
\begin{equation}\label{defnucnrm}
\|\cT\|_1=\min\Big\{\sum_{i=1}^r \prod_{j=1}^d \|\x_{i,j}\|, \;\cT=\sum_{i=1}^r \otimes_{j=1}^d \x_{i,j}\Big\}.
\end{equation}
Viewing  $\sum_{i=1}^r \prod_{j=1}^d \|\x_{i,j}\|$ as \emph{energy} of the expression of $\cT=\sum_{i=1}^r \otimes_{j=1}^d \x_{i,j}$, then $\|\cT\|_1$ is the minimal energy to decompose the tensor $\cT$ into a sum
 of rank-one tensors.

It is well known that for $d=2$, the nuclear norm
 reduces to the trace norm, $||T||_{1}= \tr\sqrt{TT^*}$,
 which is equal to the sum of  singular values of the matrix $T\in\C^{n_1\times n_2}$  -- see \cite{FL18}.  
 For $d\ge 3$, the computation of  the nuclear norm is NP-hard, since the computation of the (dual) spectral norm is NP-hard \cite{FL18}. 
 An interesting formula for the nuclear norms of special type tensors is given in \cite[Theorem 3]{SA14}.
 
One can find numerically $\|\cT\|_1$ for $\cT\in \C^{2}\times\C^{m}\times \C^{n}$ as follows:  The two first mode sections of $\cT$ are $T_1,T_2\in\C^{m\times n}$.  Let $\x=(x_1,x_2)\trans \in \C^{2}$ be a vector of length $1$: $|x_1|^2+|x_2|^2=1$.
Then $\x\times \cT=\cT(\x)= x_1 T_1 +x_2 T_2$ and
$\|\cT\|_{\infty}=\max\big\{\|\cT(\x)\|_{\infty}, \|\x\|=1\big\}.$
Note that $\cT(\x)$ is a matrix, so we can use software to find the singular value of $\cT(\x)$.  Due to numerical errors one needs to find all $\x$ where $\|\cT(\x)\|_1$ is a local maximum for $\x$ of norm $1$.

The minimal decomposition of $\cT$ with respect to the nuclear norm reads:
\begin{equation}\label{defmindecnn}
\cT=\sum_{i=1}^r \otimes_{j=1}^d \x_{i,j}, \quad \|\cT\|_1=\sum_{i=1}^r \prod_{j=1}^d \|\x_{i,j}\|.
\end{equation}
The \emph{nuclear rank} of $\cT\ne 0$, denoted as $r^{\textrm{nucl}}(\cT)$\label{eq:def_nuclearrank}, is the minimal $r$ in the above minimal decomposition.
It is assumed that $r^{\textrm{nucl}}(0)=0$. 
 By definition one has
  $r(\cT)\le r^{\textrm{nucl}}(\cT)$,
hence $r^{\textrm{nucl}}(\cT)$ 
can be interpreted as yet
 another measure of the entanglement of any $d$-partite quantum pure state
 represented by tensor $\cT$.
 In the particular case $d=2$, corresponding to bipartite systems,
 one arrives at the standard matrix rank, $r^{\textrm{nucl}}(T)=r(T)$.

Thus we can discuss similar notions for nuclear rank as for the regular rank:
\begin{enumerate}
\item What is the value of the maximum nuclear rank, denoted as $r_{\textrm{max}}^{\textrm{nucl}}(\n)$\label{eq:def_maximumnuclearrank}, and a good upper bound on its value?
\item What is a generic nuclear rank, denoted $r_{\textrm{gen}}^{\textrm{nucl}}(\n)$\label{eq:def_genericnuclearrank} and what is its value?
\item Does the border rank notion exist for nuclear norm?
\item Are there efficient algorithms to compute the nuclear rank?
\end{enumerate}

We now discuss some answers to these problems.  In order to do this we need to 
 recall some notions of convex sets in $\R^N$. 
\subsection{Faces of unit balls in $\C^{\n}$}\label{subsec:face}
We now recall several standard notions of convex sets applied to a unit ball   of any complex norm $\nu:\C^{\n}\to[0,\infty)$: $\rB_{\nu}=\{\cT\in\C^{\n}, \nu(\cT)\le 1\}$.
 It is convenient to view $\C^{\n}$ as a real space $\R^{\n}\times\R^{\n}$ of dimension $2N(\n)$.  That is $\cT=(\Re\cT,\Im\cT)$.  Then, any real functional $\phi:\C^{\n}\to \R$ is induced by $\cX\in\C^{\n}$: $\phi(\cT)=\Re \langle \cT,\cX \rangle$.  We denote this linear functional by $\phi_{\cX}$.  For $\cX,\cY\in \rB_{\nu}$ the set  $[\cX, \cY]=\{t\cX+(1-t)\cY, t\in[0,1]\}$ is called a \emph{closed interval} in $\rB_{\nu}$.
 A closed convex subset $\bF\subset B_{\nu}$ is called a \emph{face} if any open interval, $(\cX, \cY)=\{t\cX+(1-t)\cY, t\in(0,1)\}$, that lies in $\rB_{\nu}$ and intersects $\bF$ lies completely in $\bF$.  We denote that $\bF$ is a face of $\rB_{\nu}$ by $\bF\triangleleft \rB_{\nu}$.
 Note that $\emptyset$ and $\rB_{\nu}$ are faces of $\rB_{\nu}$.  Other faces of $\rB_{\nu}$ are called \emph{proper faces}. A proper face $\bF$ lies on the boundary of $B_{\nu}$, the unit sphere with respect to the nuclear norm $\rS_{\nu}=\{\cT\in\C^{\n}, \nu(\cT)=1\}$.   For example, any extreme point of $\rB_{\nu}$ is a zero dimensional face.  A dimension of a given convex set $C\subset \R^N$, is the dimension of the linear subspace spanned by affine combinations of the elements in $C$.
 As $B_{\nu}$ is a norm ball, for each tensor $\cT\in\rS_{\nu}$  one has a supporting hyperplane at $\cT$.  This supporting hyperplane can be neatly given by the dual norm 
$\nu^\vee(\cT)=\max\Re \big\{\langle \cT,\cX\rangle, \;\cX\in \rB_{\nu}\big\}.$
 Then for a given $\cT\in\rS_{\nu}$, each supporting hyperplane of $B_{\nu}$ at $\cT$ is $\phi_{\cX}$ such that $\Re \langle \cT,\cX\rangle= \nu^\vee(\cX)$.
 
 A proper face $\bF\triangleleft  \rB_{\nu}$ is called an \emph{exposed} face if it is an intersection of $\rB_{\nu}$ with a supporting hyperplane.  That is, each  
 $\cX\in\C^{\n}\setminus \{0\}$ induces an exposed face
\begin{eqnarray}\label{Xfacedef}
\bF(\cX)=\big\{\cY\in \rB_{\nu},\; \nu^\vee(\cX)=\Re\langle\cY,\cX\rangle\big\}.
\end{eqnarray}
It is known that there exist compact closed convex sets which have nonexposed faces. For example, take the standard real Hilbert norm in $\rB_{\|\cdot\|}\subset\R^N$, and a point $\x\in\R^N$ outside this ball.  Now take the Minkowski sum of $\rB_{\|\cdot\|}$ and the interval $[-\x,\x]$.  Then there exist extreme points of this balanced convex set, corresponding to the norm $\nu$, which are not exposed. (In $\R^2$ there are $4$ nonexposed extreme points.)

 A facet of $\rB_{\nu}$ is a maximal set-theoretic proper face of $\rB_{\nu}$.  By separation, every face is contained in an exposed face and thus facets are
automatically exposed \cite{SSS11}. 
 
Let $\rB_1(\n)\subset \C^{\n}$ be the unit ball of the nuclear norm, which is the convex set spanned by $\Pi(\n)$.  (Since $\Pi(\n)$ is closed it follows from Caratheodory's theorem that this convex set is closed.)  Denote by $U(n)\subset \C^{n\times n}$ the unitary group acting on $\C^n$.  Let $U(\n)$ be the product group $U(n_1)\times \cdots\times U(n_d)$ which acts on $\C^{\n}$.  First observe that $\Pi(\n)$ is the orbit of one product state $\cP\in\Pi(\n)$ under the action of $U(\n)$: $\Pi(\n)=U(\n)\cP$.
 Hence $B_1(\n)$ is an orbitope \cite{SSS11}.   Since the nuclear norm is the dual norm of the spectral norm it follows that
 \begin{eqnarray*}
 \|\cX\|_{\infty}=\max\big\{\Re \langle \cX,\cY\rangle, \cY\in \rB_1(\n)\big\}, \quad \forall \cX\in\C^{\n}.
 \end{eqnarray*} 

Thus we obtain the description of exposed faces of $\rB_1(\n)$:
\begin{lemma}\label{facdes}  Fix a state $\cX\in\C^{\n}$.  Let 
\begin{eqnarray}
\Pi(\cX)=\big\{\cP\in\Pi(\n), \Re\langle \cX,\cP\rangle=\|\cX\|_{\infty}\big\}.
\end{eqnarray}
Then $\Pi(\cX)$ is a closed set, and its convex hull is the exposed face $\bF(\cX)$ given by \eqref{Xfacedef}.  Vice versa, every exposed face of $B_1(\n)$ is of the form $\bF(\cX)$. 
\end{lemma}
\begin{proof}
Every exposed face is of the form $\bF(\cX)$.  Without loss of generality we can assume that $\cX$ is a state.  Assume now that $\cX$ is a state and consider $\bF(\cX)$.  Since the linear functional $\xi(\cT)=\Re\langle \cT,\cX\rangle$ is a supporting hyperplane of $B_1(\n)$ it follows that $\bF(\cX)$ is a facet.  Let $\Pi(\cX)$ be defined as above.  Then $\Pi(\cX)$ is a closed subset of $\Pi(\n)$.  Assume that $\cY$ is in a convex hull of $\Pi(\cX)$:
\begin{eqnarray*}
\cY=\sum_{i=1}^{r}\alpha_i \cP_i, \cP_i\in \Pi(\cX),\quad \alpha_i>0,\sum_{i=1}^r \alpha_i=1.
\end{eqnarray*}
Then $\|\cY\|_1\le \sum_{i=1}^r\alpha_i \|\cP_i\|_1=\sum_{i=1}^r \alpha_i=1$.
Furthermore
\[\Re \langle\cX,\cY\rangle=\sum_{i=1}^r \alpha_i \langle\cP_i,\cX\rangle=\|\cX\|_1.\]
Thus $\cY\in \bF(\cX)$.  

Assume that $\cY\in\bF(\cX)$.  As $\cY\in B_1(\n)$, $\cY$ is a convex combination of the extreme points of $B_1(\n)$: 
\begin{eqnarray*}
\cY=\sum_{i=1}^{r}\alpha_i \cP_i, \cP_i\in \Pi(\n),\quad \alpha_i>0, \sum_{i=1}^r \alpha_i=1.
\end{eqnarray*}
Hence $\Re \langle\cX,\cY\rangle=\sum_{i=1}^r \alpha_i \langle\cX,\cP_i\rangle\le \|\cX\|_1$.  Since $\cY\in\bF(\cX)$ it follows that $\cP_i\in\Pi(\cX)$ for $i\in[r]$. 
\end{proof}

As in \cite[Proposition 4.3]{FL16} one can generalize Lemma \ref{facdes} to an exposed face of $B_{\nu}$, where $\Pi(\n)$ is replaced by the set of the extreme points of $B_{\nu}$.
The following corollary of Lemma \ref{facdes} is given by  \cite[Lemma 4.1]{FL18}:
\begin{corollary}\label{nucnrmcert}
Let $\cT\in\C^{\n}\setminus\{0\}$ and assume that  $\cT=\sum_{i=1}^r \otimes_{j=1}^d \x_{i,j}$, where $\otimes_{j=1}^d \x_{i,j}\ne 0$ for $i\in[r]$. Then $\|\cT\|_1\le \sum_{i=1}^r \prod_{j=1}^d \|\x_{i,j}\|$.  Equality holds if and only if there exists $\cB\in\C^{\n}\setminus\{0\}$ such that  $\Re \langle\cB, \otimes_{j=1}^d\x_{i,j}\rangle=\|\cB\|_{\infty}\prod_{j=1}^d \|\x_{i,j}\|$ for $i\in[r]$.
\end{corollary} 

It is plausible to assume that  the generic nuclear rank corresponds to a generic facet of $\rB_1(\n)$.    More precisely, $r_{\textrm{gen}}^{\textrm{nucl}}(\n)$ is $1$ plus the dimension of the generic facet of $\rB_1(\n)$.  
By definition we know that
$r_{\textrm{gen}}^{\textrm{nucl}}(\n)\ge r_{\textrm{gen}}(\n)$. Caratheodory's theorem implies that $r_{\textrm{max}}^{\textrm{nucl}}(\n)$ is at most $1$ plus the dimension of the  facet of $\rB_1(\n)$ with maximum dimension.  This implies that $r_{\textrm{max}}^{\textrm{nucl}}(\n)\ge r_{\textrm{max}}(\n)$.

To find the generic nuclear rank one can do as follows:  Choose at random state $\cT\in\C^n$.  Then $\cY=\frac{1}{\|\cT\|_1}\cT$ will be an interior point of a generic facet $\bF$ of $B_1(\n)$.  Let $r$ be the number of rank-one components in a numerical minimal decomposition of $\cT$ as in Corollary \ref{nucnrmcert}.  Then  $r$ is the value of $r_{\textrm{gen}}^{\textrm{nucl}}(\n)$.  One can find numerically the nuclear norm of $\cT$ using an algorithm suggested in \cite{DFLW17}.

One of the main advantages of the nuclear rank of a tensor is that it behaves as the rank of a  matrix, in the sense that, the nuclear rank is a lower semicontinuous function \cite{FL18}.  Hence in the case of the nuclear rank there is no need to introduce its border rank.

Consider the subspace of symmetric tensors $\rS^d\C^n\subset \C^{n^{\times d}}$.
Then the dual version of the theorem of Banach \cite{Ban38} claims \cite{FL18}:
\begin{eqnarray}
\|\cS\|_1=\min\Big\{\sum_{i=1}^r \|\x_i\|^d, \; \cS=\sum_{i=1}^r \x_i^{\otimes d}\Big\}, \quad \cS\in\rS^d\C^n.
\end{eqnarray}
The minimal $r$ in the above minimal decomposition of $\cS\in\rS^d\C^n$ is called the 
\emph{symmetric nuclear rank} and is denoted
as $r_{\textrm{s}}^{\textrm{nucl}}(\cS)$\label{eq:def_symmetricnuclearrank}. 
Observe that 
  $r_{\textrm{s}}^{\textrm{nucl}}(\cS)\ge r^{\textrm{nucl}}(\cS)$,  
   since in the definition of the former quantity 
    we restrict the decomposition of $S$ to any combination of symmetric tensors of
     rank one.
 
 Denote $\rB_{1,s}(n^{\times d})=\rB_{1}(n^{\times d})\cap\cS^d\C^n$.  Then $\rB_{1,s}(n^{\times d})$ represents the unit ball in sense of the nuclear norm restricted to symmetric tensors.  The above characterization of $\|\cS\|_1$ yields that the extreme points of $\rB_{1,s}(n^{\times d})$ are $\Pi_s(n^{\times d})=\Pi(n^{\times d})\cap \rS^d\C^n$.  

 To have a better understanding of how generic and maximum nuclear ranks are related to facets of unit balls we discuss the matrix case. 
\subsection{Exposed faces and facets of matrix nuclear and spectral norms}\label{subsec:matnnrm} 
In this subsection we consider the case of $m\times n$ matrices, where $2\le m\le n$.
Recall that the inner product in $\C^{m\times n}$ is $\tr AB^*$, for $A,B\in\C^{m\times n}$.
For $A\in\C^{m\times n}\setminus\{0\}$ the singular value decomposition 
reads  
\begin{eqnarray*}
A=\sum_{i=1}^r \sigma_i(A)\bu_i\bv_i^*, \; r=r(A), \quad  \sigma_1(A)\ge\ldots\ge\sigma_r(A)>0=\sigma_{r+1}(A)=\cdots,\\ 
 A\bv_i=\sigma_i(A)\bu_i,\; A^*\bu_i=\sigma_i(A)\bv_i,\;
\bu_i\in\C^{m},\bv_i\in\C^n, \bu_i^*\bu_j=\bv_i^*\bv_j=\delta_{ij}, i,j\in[r].
\end{eqnarray*}
The vectors $\bu_i,\bv_i$ are called the left and the right singular vectors of $A$ corresponding to the i-th singular value $\sigma_i(A)$.  Furthermore,  $\|A\|_{\infty}=\sigma_1(A)$ and $\|A\|_1=\sum_{i=1}^m \sigma_i(A)$.
 
The following results are likely to be known, but  we prove them for completeness:
\begin{theorem}\label{exfcspnucnrm} Assume that $2\le m\le n$.  Denote by $\rB_{\infty}(m,n), \rB_1(m,n)\subset \C^{m\times n}$ the unit balls with respect to
 the spectral and nuclear norm respectively.  Then
\begin{enumerate}
\item Every exposed face of $\rB_1(m,n)$ has dimension $k^2-1$ , for $k\in[m]$.  It is given by $X\in\C^{m\times n}$ normalized by the condition $1=\sigma_1(X)=\cdots =\sigma_k(X)>\sigma_{k+1}(X)$.  Let $\bu_1,\ldots,\bu_k\in\C^m$ and $\bv_1,\ldots,\bv_k\in\C^n$ be two orthonormal systems corresponding to the left and the right singular eigenvectors corresponding  to the singular value $1$ of $X$.  Then the face $\bF(X)$ is a convex combination of rank-one matrices of the form $\bu\bv^*$, where $\bu$ is a unit vector in span$(\bu_1,\ldots,\bu_k)$ and $\bv=X^*\bu$.
For $k=1$, $\bF(X)=\{\bu_1\bv_1^*\}$ is an extreme point of $\rB_1(m,n)$.  The face $\bF(X)$ is a facet if and only if $k=m$.
\item  An exposed face of  $\bF\triangleleft \rB_{\infty}(m,n)$ is of dimension $2(m-k)(n-k)$ for $k\in[m]$.  It is of the following form:  Fix orthonormal systems $\bu_1,\ldots,\bu_k\in\C^m$ and $\bv_1,\ldots,\bv_k\in\C^n$.  Then
\begin{eqnarray*}
\bF=\big\{X\in\C^{m\times n}, \|X\|_{\infty}=1, X\bv_i=\bu_i, X^*\bv_i=\bu_i, i\in[k]\big\}.
\end{eqnarray*}
$\bF$ is a facet if and only if $k=1$ and $\bF$ contains an extreme point if and only if $k=m$.  Every extreme point $X\in\rB_{\infty}(\n)$ is an exposed face and is satisfies $X X^*=I_m$.
\end{enumerate}
\end{theorem}
\begin{proof}  (1) Let us consider an exposed face of $\rB_1(m,n)$.   By Lemma \ref{facdes}  it is of the form $\bF(X)=\{A\in\C^{m\times n}, \|A\|_1=1, \tr AX^*=\sigma_1(X)\}$ for some $X\in \C^{m\times n}\setminus\{0\}$.  Without loss of generality we can assume that $\sigma_1(X)=1$.
Suppose that $1=\sigma_1(X)=\cdots=\sigma_k(X)>\sigma_{k+1}(X)$.
Let $\bu_1,\ldots, \bu_k\in\C^m$ and $\bv_1,\ldots,\bv_k\in\C^n$ are two sets of orthonormal left and right singular vectors of $X$ corresponding to the first singular value  of $X$.  It is straightforward to show \cite{Fribook} that $\Re \tr X \cP^*=\sigma_1(X)$ for $\cP=\bu\bv^*\in \Pi(m,n)$ if and only if $\bu$ is a unit vector in span$(\bu_1,\ldots,\bu_k)$ and $\bv=X^*\bu$.   Suppose that $Z\in \bF(X)$.  Then the singular value decomposition of $Z$ is $Z=\sum_{j=1}^r \sigma_j(Z)\x_j\y_j^*$, where $r=r(Z)$.   Recall that $\|Z\|_1=\sum_{j=1}^r \sigma_j(Z)=1$.  Lemma \ref{facdes} yields that $\x_i\in\C^m$ has unit length,  $\x_i\in$ span$(\bu_1,\ldots\bu_k)$ and $\y_i=X^*\x_i$ for $i\in[r]$.  That is $\bF(X)$ is a convex hull of $\bu\bv^*$, where $\bu$ of length one is in span$(\bu_1,\ldots,\bu_k)$ and $\bv=X^*\bu$.  We claim that the dimension of this face is $k^2-1$.  Indeed, without loss of generality we may assume that $m=n=k$ and $X=I_k$. 
 Then the face corresponds to
  all density matrices $\rho$ of order $k$, 
 which are  hermitian, positive semidefnite,
 $\rho=\rho^*\ge 0$,  and normalized, Tr$\rho=1$.
 The real dimension of this convex set is the real dimension of all $k\times k$ hermitian matrices of trace $1$, which is $k^2-1$.

The face 
 $\bF(X)$ is  maximally exposed  if and only if $k=m$, i.e., $XX^*=I_m$.
Indeed, 
if $k<m$ then extend the orthonormal systems $\{\bu_1,\ldots,\bu_k\}, \{\bv_1,\ldots,\bv_k\}$ to an orthonormal system $\{\bu_1,\ldots,\bu_m\}, \{\bv_1,\ldots,\bv_m \}$ and set $C=\sum_{i=1}^m \bu_i\bv_i^*$.  It now follows that $\bF(X)\subsetneq \bF(C)$.

(2)  
Recall that an exposed face of $\rB_{\infty}(\n)$ is $\bF(Y)=\{X\in\rS_{\infty}(\n), \tr XY^*=\|Y\|_1\}$.  Without loss of generality we can assume that $\|Y\|_1=1$.  Assume that the SVD decomposition reads $Y=\sum_{i=1}^r \sigma_1(Y)\bu_i\bv_i^*$, where $r=r(Y)$, $\sigma_1(Y)\ge \cdots\sigma_r(Y)>0$ and $\sum_{i=1}^r \sigma_1(Y)=1$.  Assume that $X\in\rS_{\infty}(\n)$.  Then $\Re\tr X(\bu_i \bv_i^*)^*\le \|X\|_{\infty}=1$.  Equality holds if and only if $X\bv_i=\bu_i$ and $X^*\bu_i=\bv_i$ for $i\in[r]$.  Thus $\bF(Y)$ consists of all $X\in \rB_{\infty}(\n)$ satisfying  $X\bv_i=\bu_i$ and $X^*\bu_i=\bv_i$ for $i\in[r]$.  In particular, $\sigma_1(X)=\cdots=\sigma_k(X)=1\ge \sigma_{k+1}(X)$.  By choosing an orthonormal bases $\bu_1,\ldots,\bu_m\in\C^m$ and $\bv_1,\ldots,\bv_n\in\C^n$ we see that $X$ is a direct sum $I_k\oplus X'$, where $X'\in \C^{(m-k)\times (n-k)}$ and $\|X'\|_{\infty}\le 1$.  Hence the real dimension of the face $\bF(Y)$ is $2(m-k)(n-k)$.  

Observe first that $\bF(Y)$ is a facet if $Y$ is a rank-one matrix.  In this case the dimension of the facet is $2(m-1)(n-1)$.  The face $\bF(Y)$ is zero dimensional if and only if $r(Y)=m$.  In this case $\bF(Y)=\{X\}$, where $X=\sum_{i=1}^m \bu_i\bv_i^*$.  Thus $X$ is an extreme point of $\rB_{\infty}$.  It is left to show that every extreme point $X$ of  $\rB_{\infty}$ is of this form.  Let $X\in\rB_1(\n)$ and consider the full SVD decomposition of $X=\sum_{i=1}^m \sigma_i(X)\bu_i\bv_i^*$, where $\bu_1,\ldots,\bu_m$ and $\bv_1,\ldots,\bv_m$ are orthonormal vectors.  Furthermore, $1=\sigma_1(X)\ge \cdots\ge \sigma_m(X)\ge 0$.  Assume to the contrary that $1>\sigma_m(X)$.  Choose $\varepsilon>0$ such that $1>\sigma_m(X)+\varepsilon$.  Then 
\begin{eqnarray*}
X(\pm\varepsilon)=(\sigma_m(X)\pm\varepsilon)\bu_m\bv_m^*+\sum_{i=1}^{m-1}
\sigma_i(X)\bu_i\bv_i^*\in \rB_1(\n).
\end{eqnarray*}
Since 
$X=\frac{1}{2}(X(\varepsilon)+X(-\varepsilon))$ its is not an extreme point.
\end{proof}
 \subsection{The nuclear rank of the $|GHZ\rangle$ state}\label{subsec:GHZ}
We now discuss the nuclear rank of $3$-qubits.  First observe that the state $|GHZ\rangle$ has nuclear rank $2$. Indeed, up to a normalization constant we have 
\begin{eqnarray}
|GHZ\rangle=\e_1^{\otimes 3}+\e_2^{\otimes 3}=\e_1\otimes\e_1^{\otimes 2}+\e_2\otimes \e_2^{\otimes 2}.
\end{eqnarray}
Written in physics notation, 
$|GHZ\rangle =|000\rangle + |111\rangle$,
this state has a two-term representation, 
so its rank is not more than two,
 $\||GHZ\rangle\|_1\le 2$.
 
   On the other hand, when we unfold $|GHZ\rangle$ in mode $1$ we obtain
    $T\in\C^{2\times 4}$,  for which $\|T\|_1=2$. 
    Hence the above decomposition of $|GHZ\rangle$ is a minimal decomposition.  As $r(|GHZ\rangle)=2$ we deduce that $r^{\textrm{nucl}}(|GHZ\rangle)=2$.  Let us try to find the facet that contains $|GHZ\rangle$.  We claim that we can choose the supporting plane $\xi(\cT)=\Re\langle |GZW\rangle,\cT\rangle$.  Indeed, $|GHZ\rangle$ is a symmetric tensor corresponding  to symmetric polynomial $x_1^3+x_2^3$.  Then $\max\{\Re{x_1^3+x_2^3}, |x_1|^2+|x_2|^2=1\}$ is achieved at the vectors $\zeta\e_1,\zeta\e_2$ where $\zeta^3=1$.  By Corollary \ref{nucnrmcert} we see that $\frac{1}{2}|GHZ\rangle\in \bF(|GHZ\rangle)$.

\subsection{Generic and maximum nuclear rank of symmetric $3$-qubit states}\label{subsec:gn3qub}
We now discuss in details the exposed facets and faces of the nuclear norm ball in $\rS^3\C^2$.  Denote by $B_1(3,2)\subset \rS^3\C^2 $ the unit ball of the nuclear norm.  Denote by $\partial B_1(3,2)$ the boundary of $B_1(2,3)$, which is the unit sphere of the nuclear ball.
Recall that the complex dimension of $\rS^3\C^2$ is $4$.  Hence, its real dimension is $8$.
Any supporting hyperplane of $B_1(3,2)$ is of the form
$\{S\in B_1(3,2),\Re\langle \cS,\cT\rangle=\|\cT\|_{\infty}\}$ for some $\cT\in \rS^3\C^2\setminus\{0\}$.
For simplicity of notation we will identify $\cT$ with a homogeneous polynomial of degree $3$ in $x_1,x_2$.

The extreme points of $B_1(3,2)$ are rank-one symmetric tensors of the $\bu^{\otimes 3}$ with $\|\bu\|=1$. Such symmetric tensor determines $\bu$ up to a third root of unity. That is, we can replace $\bu$ by $\zeta\bu$, where $\zeta^3=1$.  Thus all extreme points can be identified with the $3$-dimensional real sphere $\rS^3$ in the complex space $\C^2$, quotient by the action of multiplication by third roots of unity.  The supporting hyperplane is a corresponding nonzero homogeneous polynomial $(ax_1+bx_2)^3$.  Thus the nuclear rank of an extreme point of $B_1(3,2)$ is $1$, and the dimension of the exposed face is $0$.  

We next discuss exposed face of dimension $1$.  Let us first consider the example of
the  $|GHZ\rangle$ state.  It corresponds to the polynomial $x_1^3+x_2^3$,
up to  multiplication by a constant.  Hence $\||GHZ\rangle\|_{\infty}=1$ and the maximum is achieved for the two extreme points of $B_1(3,2)$: $\be_1^{\otimes 3},\be_2^{\otimes 3}$,
and any convex combination of these two points.  Thus, the exposed face corresponding to $|GHZ\rangle$ is a convex combination of two unit orthogonal vectors in $\C^2$.  The variety of all these $1$-dimensional exposed faces is of real dimension $4$  --  three for $\bx$ of length one, and an extra dimension for an orthogonal vector $\by$ of length one.
 It is not known to the authors whether there exist
  additional exposed faces of dimension $1$.

 To find the spectral norm of a tensor one may follow the approach which is described in \cite{FW18}. 
 For a generic tensor in $\rS^3\C^2$ one can find all critical points for $\max\{\Re \langle \cS,\bu^{\otimes 3}\rangle, \|\bu\|=1\}$.
It will have usually at most $5$ critical $\bu$ viewed as a point on the Riemann sphere. Those points that correspond to $2$ points of maximum will yield a supporting hyperplane of $B_1(3,2)$ which supports an exposed face of dimension $1$. There are exceptional cases that are discussed in \cite{FW18}.

Let us now consider the three-qubit state $|W_3\rangle=(|100\rangle+ |010\rangle +|001\rangle) / \sqrt{3}$.  
Recall that  $|W_3\rangle$ is a symmetric tensor corresponding to the polynomial $f=\sqrt{3}x_1^2x_2$.   Hence 
\begin{equation}\||W_3\rangle\|_{\infty}=\max\big\{|f(x_1,x_2)|=\sqrt{3}|x_1|^2|x_2|, |x_1|^2+|x_2|^2=1\big\}=\frac{2}{3},\end{equation}
as $|x_1|=\frac{\sqrt{2}}{\sqrt{3}}$, $|x_2|=\frac{1}{\sqrt{3}}.$
Recall that among all $3$-qubit states $|W_3\rangle$ has the minimal spectral norm, i.e., it has the highest geometrical measure of entanglement \cite{TWP09}.

In \cite[\S6]{FL18} we gave a nuclear decomposition of $|W_3\rangle$ with four terms.  We showed that $\||W_3\rangle\|_1=3/2=\||W_3\rangle\|_\infty^{-1}$. The last equality follows from \cite[Theorem 2.2]{DFLW17}, as $|W_3\rangle$ has the minimal spectral norm,   In particular, $|W_3\rangle$ has the maximum nuclear norm among the $3$-qubit states.

The four rank-one symmetric tensors of norm one that appear in the nuclear decomposition of $|W_3\rangle$, given in \cite[\S6]{FL18} are all rank-one symmetric states for which $|W_3\rangle$ achieves its spectral norms.  Thus $|W_3\rangle\in \bF(|W_3\rangle)$.

Let us now consider $\bF(|W_3\rangle)$ in the ball of the nuclear norm in $\rS^3\C^2$. Extreme points are all tensors $\bu^{\otimes 3},\|\bu\|=1$, such that
$\Re \langle |W_3\rangle, \bu^{\otimes 3}\rangle=\| |W_3\rangle\|_{\infty}$.  It is straightforward to show that $\bu=\Big(\zeta\sqrt{\frac{2}{3}}, \bar\zeta^2\frac{1}{\sqrt{3}}\Big)^\top$, where $|\zeta|=1$. 

Indeed, each rank-one tensor $\bu^{\otimes 3}, \bu=(a,b)\trans$, corresponds to the polynomial 
\begin{eqnarray*}
(ax_1+bx_2)^3=a^3x_1^3+a^2b(3x_1^2x_2)+ab^2(3x_1x_2^2)+b^3(x_2^3).
\end{eqnarray*}
Hence in the basis $\big\{x_1^3, 3x_1^2x_2, 3x_1x_2^2, x_3^3\big\}$ of homogeneous polynomials of degree $3$ in $x_1,x_2$ tensor $\bu^{\otimes 3}$ is represented by a vector 
$(a^3,a^2b, a b^2, b^3)$. Assume that $a\ne 0$.  Then this vector is $a^3(1,z,z^2,z^3)$ and $z=b/a$.  For the extreme points of the exposed face corresponding to $\bF(|W_3\rangle)$ one has $z=\eta \frac{1}{\sqrt{2}}$, where $\eta=\bar\zeta^3$.  So $\eta$ has an arbirary value on the unit circle in $\C$.
In particular, if we choose four pairwise distinct points on the unit circle $\eta_1,\ldots,\eta_4$, the four points $(1,\eta_i,\eta_i^2, \eta_i^3), i\in[4]$ are linearly independent, as their determinant has Vandermonde form and does not vanish.

Since we consider the convex hull of the extreme points of $\bF(|W_3\rangle)$, we need to know what is the real dimension of this convex set.  
We claim that the dimension is $4$:
\begin{theorem}\label{dimfw3}
Let 
\begin{eqnarray*}
\bF(|W_3\rangle)=\big\{\cS\in B_1(3,2), \Re \langle\cS,|W_3\rangle\rangle=\||W_3\rangle\|_{\infty}\big\}
\end{eqnarray*}
be the face of the ball corresponding to the real functional $\cS\mapsto \Re \langle\cS,|W_3\rangle\rangle$.  Then
\begin{enumerate}
\item The real dimension of $\bF(|W_3\rangle)$ is $4$.
\item The state $|W_3\rangle$ has the following nuclear decomposition:
\begin{eqnarray*}
|W_3\rangle=\frac{\sqrt{3}}{2}\Bigg(\Big(\sqrt{\frac{2}{3}}\e_1+\frac{1}{\sqrt{3}}\e_2\Big)^{\otimes 3}+\Big(\sqrt{\frac{2}{3}}\zeta \e_1+\frac{1}{\sqrt{3}}\bar\zeta^2\e_2\Big)^{\otimes 3}
\\+\Big(\sqrt{\frac{2}{3}}\bar\zeta\e_1+\frac{1}{\sqrt{3}}\zeta^2\e_2\Big)^{\otimes 3}\Bigg), \;\zeta=e^{2\pi\bi/9}.
\end{eqnarray*}
As the above nuclear decomposition is a sum of three symmetric tensors of rank-one, and $\rank |W_3\rangle=3$, we deduce:
 the nuclear rank of $|W_3\rangle$ is $3$, which is equal to its rank.
\item The nuclear rank of any $\cS\in  \bF(|W_3\rangle)$ is at most $4$.
\item The subgroup of $G$ of the two dimensional unitary group $U(2)$ that fixes either $|W_3\rangle$ or 
$\bF(|W_3\rangle)$ is one dimensional subgroup of the form 
\begin{eqnarray*}
G=\Big\{U\in\C^{2\times 2}, U=\left[\begin{array}{cc}\zeta&0\\0&\bar\zeta^2\end{array}\right], \zeta\in\C, |\zeta|=1\Big\}.
\end{eqnarray*}
\item The semialgebraic set of faces of the form $U\bF(|W_3\rangle), U\in U(2)$ has dimension $7$.
\end{enumerate}
\end{theorem}
\begin{proof}
(1) Recall that our three symmetric tensors, which are the extreme points of $\bF(|W_3\rangle)$ correspond to linear forms $(ax_1+bx_2)^3$, where $a=\zeta s$ and $b=\bar\zeta^2 t$ where $s=\sqrt{2}/\sqrt{3}$ and $t=1/\sqrt{3}$.  By considering new variables $y_1=sx_1$ and $y_2=tx_2$, we have that the qubic forms are $(\zeta y_1+\bar\zeta^2 y_2)^3$.  Hence the Veronese coordinates of these cubic forms are
$(\zeta^3, 1, \bar\zeta^3, \bar\zeta^6)$.  Letting $\xi=\bar\zeta^3$ and rearranging the coordinates we have that the coordinates are $(1,\xi,\bar\xi, \xi^2)$, where $|\xi|=1$.  Since we consider convex (or affine) combination we can drop the first coordinate.  Thus we are looking at convex combinations of the vectors of $(\xi,\bar\xi,\xi^2)$, where $\xi$ is on the unit circle in $\C$.  Let us try to find a basis in the space of all real linear combinations whose sum is $1$.  
The vector $(-\xi,\overline{-\xi}, (-\xi)^2)=(-\xi,-\bar\xi, \xi^2)$ is such a point. Thus
\begin{eqnarray}
\big((\xi,\bar\xi,\xi^2)+(-\xi,-\bar\xi, \xi^2)\big)/2=(0,0,\xi^2), \quad |\xi|=1.
\end{eqnarray} 
 The set of all convex combinations of unimodular points, $|\eta|=1$,
  forms the unit disk, $|z|\le 1$, of real dimension $2$. 
   Thus the convex combination of all vectors 
of the form $(0,0,\xi^2)$ has a real dimension $2$.  Note that $(0,0,-\xi^2)$ is also in the convex set.  Hence
\begin{eqnarray}
\frac{1}{2}(\xi,\bar\xi,0)=\big((\xi,\bar\xi,\xi^2)+(0,0,-\xi^2)\big)/2.
\end{eqnarray}
Next, observe that the convex hull of the above vectors  is $1/2(z,\bar z,0)$ where $|z|\le 1$.  Thus its real dimension is $2$. Hence the convex combinations of vectors 
$1/2(z,\bar z,0)$ and $(0,0,w)$, where $z$ and $w$ are in a unit disk in $\C$ has real dimension $4$.  Therefore the real dimension of $\bF(|W_3\rangle)$ is $4$.  

\noindent
(2) Straightforward.

\noindent
(3) We claim that any point in $\bF(|W_3\rangle)$ is a convex combination of at most $4$ extreme points.   Note that any convex combination of the exteme points is of the form $(z,\bar z, w)$.  Hence it is enough to consider the convex combinations of vectors $(\xi,\xi^2)$ for $|\xi|=1$.
As the ambient subspace is $4$-real dimensional, Caratheodory's theorem shows that every point is a convex combination of at most $5$ extreme points.  Assume that $(z,w)$ is in this convex set and 
\begin{eqnarray*}
(z,w)=\sum_{i=1}^5 a_i(\xi_i,\xi_i^2), \quad a_i>0, \sum_{i=1}^5a_i=1.
\end{eqnarray*}
We assume that $\xi_i\ne \xi_j$ for $i\ne j$.
If the real span of $(\xi_1,\xi_1^2),\ldots,(\xi_5,\xi_5^2)$ is $3$-dimensional we are done.
Then, there must be a real nonzero  linear combination $\sum_{i=1}^5 b_i (\xi_i,\xi_i^2)=0$ such that $\sum_{i=1}^5 b_i=0$.  By considering the linear combination $\sum_{i=1}^5 (a_i+tb_i)(\xi_i,\xi_i^2)$  we can choose a positive $t>0$ such that $a_i+tb_i\ge 0$ for all $i\in [5]$ such that at least one $a_i+tb_i=0$.

Thus we are left with the case that every minimal convex combination that contains $(w,z)$ has exactly $5$ extreme points with positive coefficients and which have $4$ linear independent extreme points. 
Assume for simplicity that these linearly independent (over real domain) elements are $(\xi_1,\xi_1^2),\ldots,(\xi_4,\xi_4^2)$.  Then  $(\xi_5,\xi_5^2)=\sum_{i=1}^4 c_i(\xi_i,\xi_i^2)$. Suppose first that $\sum_{i=1}^4 c_i=1$.  Since $(\xi_5,\xi_5^2)$ is an extreme point we must have that $c_i<0$ for some $i\in[4]$.   Observe next that
\begin{eqnarray*} 
(z,w)=\sum_{i=1}^5 a_i(\xi_i,\xi_i^2)=\sum_{i=1}^4 (a_i+a_5c_i)(\xi_i,\xi_i^2).
\end{eqnarray*}
Let $a_i(t)=a_i+tc_i$ for $i\in[4]$ and $a_5(t)=a_5-t$.  Then 
\begin{eqnarray*}
\sum_{i=1}^5 a_i(t)=1, \quad(z,w)=\sum_{i=1}^5 a_i(t)(\xi_i,\xi_i^2).  
\end{eqnarray*}
Start to increase $t$ from $0$ to $a_5$ until $a_i(t)$ is zero.
In this case we obtained that $(z,w)$ is a convex combination of at most $4$ extreme points in $\bF(|W_3\rangle)$.

Thus, it is left to discuss the case where $(\xi_1,\xi_1^2),\ldots,(\xi_4,\xi_4^2)$ are linearly independent and $(\xi_5,\xi_5^2)$ is not an affine combination of 
$(\xi_1,\xi_1^2),\ldots,(\xi_4,\xi_4^2)$.  Note that the convex span of $(\xi_1,\xi_1^2),\ldots,(\xi_4,\xi_4^2)$ has real dimension $3$.  Furthermore, the convex span of $(\xi_1,\xi_1^2),\ldots,(\xi_5,\xi_5^2)$ has dimension $4$.

Let us denote by $\Delta=\Delta(w,z)$ these $5$ tuples of points $((\xi_{i_1},\xi_{i_1}1^2),\cdots,(\xi_{i_5},\xi_{i_5}^2))$, where $\{i_1,\ldots,i_5\}=[5]$.  This is an open set in $C^5$, where $C$ is a complex closed curve,
 $C:=\{(\zeta,\zeta^2)\in\C^2, |\zeta|=1\}$.  That is $\Delta\subset C^5$ is an open set, which is not $C^5$.  Now consider a boundary point of $\Delta$ in $C^5$.  This boundary point consists of $5$ tuples $\eta=((\eta_1,\eta_1^2),\ldots,(\eta_5,\eta_5^2)), |\eta_i|=1, i\in[5]$.  As this boundary point is a limit of  points in $\Delta$ it follows that $\eta\in \Delta$.  But then, either the five points $(\eta_i,\eta_i^2), i\in[5]$ span at most $3$-dimensional  subspace, or there are $4$ extreme points which are linearly independent and the fifth point is an affine combination of the remaining four.
In these cases we obtain that $(z,w)$ is a convex combination of $4$ extreme points.

\noindent
(4) We now do the dimension count: How big is the subgroup of $U(2)$ that fixes $\bF(|W_3\rangle)$.  It is equivalent to the subgroup that fixes $|W_3\rangle$.
It has a real dimension at least $1$: Assume that $U\be_1=\zeta\be_1, U\be_2=\bar\zeta^2\be_2$ for $|\zeta|=1$, which implies  $U^{\otimes 3}|W_3\rangle=|W_3\rangle$.  So for sure we have a one-parameter group that fixes $\bF(|W_3\rangle)$.  

We now show that the above subgroup is the only subgroup that fixes $|W_3\rangle$.
Indeed, suppose that
\begin{eqnarray*}
\cS=U^{\otimes 3}|W_3\rangle=\x_1\otimes\x_1\otimes\x_2+\x_1\otimes\x_2\otimes\x_1+\x_2\otimes\x_1\otimes\x_1,
\;\x_i^*\x_j=\delta_{ij}, i,j\in[2].
\end{eqnarray*}
where the two pairs of vectors $\be_1,\x_1$ and $\be_2,\x_2$ are linearly independent.
 Consider the rank-one matrix $\cS \times \bar\x_2=\x_1\otimes\x_1$, where  the contraction is on the third mode.  Assume that $\cS=|W_3\rangle$.  Since  $\be_1$ and $\x_1$ are linearly independent it follows that $\x_2^*\be_1\ne 0$.    A straightforward computation implies that the rank of the  matrix $|W_3\rangle\times \bar \x_2$ is $2$.  This contradicts the assumption that $\cS=|W_3\rangle$. Hence the subgroup $G$ is of the form given above.
 
 \noindent
(5)  Observe that $U\bF(|W_3\rangle),  U\in U(2)$ must be also an exposed face of  dimension $4$.  As the subgroup of $U(2)$ that fixes $\bF(|W_3\rangle)$ is one dimensional, it follows that the dimension of union of all faces of the form  $U\bF(|W_3\rangle)$ is $3+4=7$.
Recall that $7$ is the dimension of $\partial B(3,2)$.
\end{proof}

Let us conclude the work with a short list of open questions:
\begin{enumerate}
\item Is the nuclear rank of $\cS\in\bF(|W_3\rangle)$ at most $3$?
\item Is $\bF(|W_3\rangle)$ a face of maximum real dimension?
\item Is $U\bF(|W_3\rangle), U\in U(2)$ the set of all faces of $B_1(3,2)$?
\item Is the generic face of $B_1(3,2)$ of dimension $3$, such that the subgroup of $U(2)$ that fixes generic faces is a finite group?  
In such a  case the semialgebraic set of $U\bF,U\in U(2)$ has dimension $3+4=7$. 
\end{enumerate}

\section*{Acknowledgments}
It is a pleasure to thank Jaros{\l}aw Buczy{\'n}ski
for numerous discussions and constructive suggestions.
We are also grateful to Zbigniew Pucha{\l}a and Oliver Reardon-Smith
for their remarks on the paper 
and to all the referees for their detailed 
comments which allowed us to improve the work.
This research was supported by  National Science Center in Poland
under the Maestro grant number DEC-2015/18/A/ST2/00274, by Foundation for Polish Science under the grant Team-Net NTQC, and Simons collaboration grant for mathematicians.


\appendix
\section{Basic notions of quantum theory}
\label{AppA}

\smallskip
In this Appendix we present definitions of some notions used in quantum theory
and 
discussed in this work. To make it easier for the reader to study the literature 
of the subject we are going to use the Dirac notation presented in Section \ref{subsec:Diracnot}. 
 In short,  the  `ket'  
$|\psi\rangle$ denotes a complex vector of a fixed size $n$ represented by a column, 
the `bra' $\langle \psi |$ is a conjugated (dual) vector forming a row, 
the scalar product is written as a bra--ket,
$\langle \psi |\phi\rangle \in \C$,
while $|\phi \rangle \langle \psi|$ forms an operator,
represented by a rank-one square matrix of size $n$.

The term {\sl state} is understood as a mathematical tool used to calculate the
probability of a given outcome of any measurement.
In the classical probability theory one uses probability vectors, $p=\{p_1, \dots p_n\}$,
such that $p_i\ge 0$ and $\sum_{i=1}^n p_i=1$. The natural number $n$ describes the number of distinguishable events and is fixed.
The set of classical states forms the probability simplex of dimension $n-1$.
A probability vector $p$ with a single component equal to unity,
representing a vertex of the simplex,  is called an extremal point or
 a classical pure state, and corresponds to a certain event.
 Any point inside the simplex can be considered as a classical mixed state.

One of the key notions of  quantum theory is the {\sl quantum state},
which characterizes  the way a physical system was prepared
and allows one to compute the
probability of an outcome of any quantum measurement. 
Let us fix a natural number $n$, assumed here to be finite,
and discuss first a special class of states.
A pure quantum state is represented by 
 a ray in an $n$ dimensional complex Hilbert space ${\mathcal H}_n$.

\medskip

{\bf Definition.} 
{\sl Consider a complex vector $|\psi\rangle \in {\mathcal H}_n$,
  normalized as $||\psi||^2=\langle \psi|\psi\rangle=1$ and
  an arbitrary complex  phase, $e^{i \alpha}$, with $\alpha\in [0, 2 \pi]$.
  A {\sl pure quantum state} denotes the equivalence class,
   $|\psi\rangle \sim e^{i \alpha}|\psi\rangle$.}

\medskip

The space of all pure states forms a complex projective space, $\C P^{n-1}$,
of $2(n-1)$ real dimensions. In the simplest case, $n=2$,
often called a single {\sl qubit} (i.e. a {\sl qu}antum {\sl bit}) system, 
this space  forms a sphere, $\C P^{1}=S^2$,
in physics called the {\sl Bloch sphere}.
The set of pure quantum states is continuous,
in contrast to the discrete set of classical pure states - the corners of the
probability simplex.

\smallskip
A hermitian operator $P_{\psi}= |\psi \rangle \langle \psi|=P_{\psi}^2$ 
is a projection operator onto a pure state  $|\psi \rangle$. 
In physics literature  the term `pure state' may denote a state
 $|\psi \rangle$ or the corresponding projector   $P_{\psi}$
and the meaning depends on the context.
This difference is not  relevant in physics,  
since the vectors representing rays  are in one-to-one correspondence
with the projectors.

Any convex combination of such projectors,
$\rho =\sum_{j=1}^k  q_j  |\psi_j \rangle \langle \psi_j|$,
forms a mixture of pure states, 
where $q$ represents a probability vector of an arbitrary length $k$.
Such a mixture, called a {\sl density matrix} or a {\sl mixed state} or
just a {\sl quantum state}, 
can be introduced in a more formal way.

\medskip

{\bf Definition.} 
{\sl A square matrix $\rho$ of order $n$ is
called  a density matrix if it is hermitian, $\rho=\rho^*$,  positive semi-definite, 
$\rho\ge 0$, and normalized, ${\rm Tr}\rho=1$.}

\medskip

Any convex mixture of 1-dimensional projection operators
$P_{\psi}$ satisfies above properties.
Note that the diagonal entries of a density operator
represented in any basis 
are real and  form a classical state - the probability vector $p$ of length $n$
with components $p_i=\rho_{ii}$, $i=1,\dots n$.
Let $\rho=\sum_{j=1}^n \lambda_j |\chi_j\rangle \langle \chi_j|$
be the eigen-decomposition of the state $\rho$, where the eigenvalues
 $\lambda$   form a probability vector.
The rank  $r$ of the state $\rho$ is equal to the number of its positive eigenvalues.
The case $r=1$ corresponds to the projector,
$\rho=|\chi_1\rangle \langle \chi_1|=\rho^2$.
It forms an extremal point of the set of all density matrices
of order $n$, which explains the term `pure state'.
In Euclidean geometry, induced by the Hilbert-Schmidt distance, 
$d_{\rm HS}(X,Y)=||X-Y||$,
the space of all mixed states for $n=2$ forms a solid Bloch ball
with the Bloch sphere consisting of pure states at its boundary.
In higher dimensions $n\ge 3$ there exist points
from the boundary of the set of mixed states which are not pure
-- see e.g. \cite{BZ17}.

\medskip

In some physical systems one can observe an internal structure and 
can identify its subsystems. Assume first, that the system is {\sl bipartite},
so two parties, called $A$ and $B$ are distinguished.
Let ${\mathcal H}_A$ and ${\mathcal H}_B$ denote Hilbert spaces
used to describe subsystems $A$ and $B$, respectively.
Then the  bipartite system $AB$ is described by
a quantum state from the composite Hilbert space
with a tensor product structure,
 ${\mathcal H}_{AB}= {\mathcal H}_A \otimes {\mathcal H}_B$.
 Once both subsystems $A$ and $B$ are well defined and the 
 above splitting  of ${\mathcal H}_{AB}$ is fixed,
one can introduce the notion of 
separable and entangled states  \cite{Bel64,HHHH09},

\medskip

{\bf Definition.} 
{\sl A bipartite pure quantum  state
$|\psi_{AB}\rangle \in   {\mathcal H}_A \otimes {\mathcal H}_B$
 is called separable, if it has the product form,
$|\psi_{AB}\rangle =|\phi_A\rangle \otimes |\phi_B\rangle$,
where  $|\phi_A\rangle  \in {\mathcal H}_A$ and  $|\phi_B\rangle  \in {\mathcal H}_B$.}

\medskip

{\bf Definition.} 
{\sl A bipartite pure quantum  state
$|\psi_{AB}\rangle \in   {\mathcal H}_A \otimes {\mathcal H}_B$
 is called entangled if it is not separable,
  so it is not of the product form.}

\medskip

The above definitions  do not depend on the choice of the
local bases in both subspaces,
but they do depend on the splitting  of ${\mathcal H}_{AB}$
into ${\mathcal H}_A $ and ${\mathcal H}_B$. Let $n$ and $m$ denote
the dimensions of these two subspaces, respectively.
The entanglement is then invariant with respect to the
{\sl local unitary transformations} from the group $U(n) \otimes U(m)$,
but it can change under global unitary transformations,
 $|\psi\rangle \to V|\psi\rangle$, with  $V\in U(nm)$. 

\medskip
 
 Note that a separable pure state represents independent events,
  for which the joint probability of events measured seperately in both subsystems
   has the product form,
while correlated events correspond to entangled states.
 In the general case of density matrices
 the definition of separability is slightly more involved \cite{We89}.
 
\medskip

{\bf Definition.} 
{\sl A bipartite  quantum (mixed) state
$\rho^{AB}$ acting on the composite space ${\mathcal H}_A \otimes {\mathcal H}_B$
 is called separable,  if it can be represented as a convex combination of
 product states,}
 
 \begin{equation}\rho^{AB}_{\rm sep}= \sum_{j=1}^k   q_j  \;  \rho_j^A \otimes \rho_j^B,\end{equation}
 
 \noindent
 {\sl 
 where $q$ is a probability vector of length $k$,
 while $\{\rho_j^A\}$ and $\{\rho_j^B\}$ denote collections of $k$ quantum states
 acting on Hilbert spaces  ${\mathcal H}_A$ and  ${\mathcal H}_B$, respectively.}

\medskip

{\bf Definition.} 
{\sl A bipartite quantum (mixed) state
$\rho^{AB}$ acting on the composite space ${\mathcal H}_A \otimes {\mathcal H}_B$
 is called entangled if it is not separable}.

\medskip
 
 It is easy to see that in the case of pure states, which are extremal
 and cannot be represented by a mixture of other pure states,
 both definitions of entanglement are consistent. Furthermore, 
  entanglement does not depend on the choice
 of the local bases also for mixed states,
  but it depends on the partition of the total system
 into two parts. In the case of a two-qubit system, $n=m=2$,
 the necessary and sufficient conditions for separability 
 of a given density matrix $\rho$ of size $4$ are known
 \cite{HHHH09,BZ17}, but already for $n=m=3$
 the separability problem remains open.

\medskip
 
The above notions are easy to generalize for {\sl multipartite systems}
which consist of $d\ge 3$ subsystems.
For instance,
 a separable pure state of a $d$--partite system 
has the product form,
$|\psi\rangle =|\phi_1\rangle \otimes |\phi_2\rangle \otimes \cdots \otimes |\phi_d\rangle$.

 \end{document}